\documentclass[10pt]{amsart}

\usepackage{amssymb}
\usepackage{amsthm}
\usepackage{amsfonts}
\usepackage{amsmath}
\usepackage{pmboxdraw}
\usepackage{graphicx}
\usepackage{texdraw}
\usepackage{graphpap}
\usepackage{wasysym}
\usepackage{enumitem}
\usepackage{color}
\usepackage[OT2,T1]{fontenc}
\usepackage{verbatim} 
\usepackage{multirow}
\usepackage{mathtools}
\usepackage[colorlinks=true, citecolor=blue, filecolor=black, linkcolor=black, urlcolor=black]{hyperref}
\usepackage{cite}
\usepackage{bm}
\usepackage{bbm}
\usepackage{todonotes}
\usepackage{kantlipsum}
\usepackage{subcaption}

\theoremstyle{plain}
\newtheorem{thm}{Theorem}[section]

\newtheorem{lem}[thm]{Lemma}

\newtheorem{prop}[thm]{Proposition}

\theoremstyle{remark}

\newtheorem{rem}[thm]{Remark}
\newtheorem{defi}[thm]{Definition}

\newcommand{\R}{{\mathbb R}}

\newcommand{\re}{\operatorname{Re}}

\newcommand{\Tr}{\operatorname{Tr}}

\newcommand{\1}{\chi}

\newcommand{\Ai}{\operatorname{Ai}}
\newcommand{\ext}{\operatorname{ext}}

\def\FF{\mathcal{F}}

\def\C{\mathbb{C}}

\def\R{\mathbb{R}}

\numberwithin{equation}{section}

\newcommand{\RN}[1]{%
	\textup{\uppercase\expandafter{\romannumeral#1}}%
}

\usepackage{tikz}
\usetikzlibrary{arrows}
\usetikzlibrary{decorations.pathmorphing}
\usetikzlibrary{decorations.markings}
\usetikzlibrary{patterns}
\usetikzlibrary{automata}
\usetikzlibrary{positioning}
\usepackage{tikz-cd}
\tikzset{->-/.style={decoration={
			markings,
			mark=at position #1 with {\arrow{latex}}},postaction={decorate}}}

\tikzset{-<-/.style={decoration={
			markings,
			mark=at position #1 with {\arrowreversed{latex}}},postaction={decorate}}}

\usetikzlibrary{shapes.misc}\tikzset{cross/.style={cross out, draw, 
		minimum size=2*(#1-\pgflinewidth), 
		inner sep=0pt, outer sep=0pt}}

\usepackage{pgfplots}

\begin{document}

\title[Free energy expansions of a conditional complex Ginibre ensemble]{Free energy expansions of a conditional G{\SMALL in}UE and 
\\ large deviations of the smallest eigenvalue of the LUE
}

\author{Sung-Soo Byun}
\address{Department of Mathematical Sciences and Research Institute of Mathematics, Seoul National University, Seoul 151-747, Republic of Korea}
\email{sungsoobyun@snu.ac.kr}

\author{Seong-Mi Seo}
\address{Department of Mathematics, Chungnam National University, Daejeon 34134, Republic of Korea.}
\email{smseo@cnu.ac.kr}

\author{Meng Yang}
\address{Department of Mathematics, School of Sciences,  Great Bay University, Dongguan 523000, China}
\email{my@gbu.edu.cn}

\begin{abstract}
We consider a planar Coulomb gas ensemble of size $N$ with the inverse temperature $\beta=2$ and external potential $Q(z)=|z|^2-2c \log|z-a|$, where $c>0$ and $a \in \C$. 
Equivalently, this model can be realised as $N$ eigenvalues of the complex Ginibre matrix of size $(c+1) N \times (c+1) N$ conditioned to have deterministic eigenvalue $a$ with multiplicity $cN$. 
Depending on the values of $c$ and $a$, the droplet reveals a phase transition: it is doubly connected in the post-critical regime and simply connected in the pre-critical regime.
In both regimes, we derive precise large-$N$ expansions of the free energy up to the $O(1)$ term, providing a non-radially symmetric example that confirms the Zabrodin-Wiegmann conjecture made for general planar Coulomb gas ensembles.  
As a consequence, our results provide asymptotic behaviour of moments of the characteristic polynomial of the complex Ginibre matrix, where the powers are of order $O(N)$. 
Furthermore, by combining with a duality formula, we obtain precise large deviation probabilities of the smallest eigenvalue of the Laguerre unitary ensemble. 
A key ingredient for the proof lies in the fine asymptotic behaviour of a planar orthogonal polynomial, extending a result of Betola et al \cite{BBLM15}. This result holds its own interest and is based on a refined Riemann-Hilbert analysis using the partial Schlesinger transform.
\end{abstract}

\date{\today}

\maketitle


\section{Introduction}

\subsection{Models and numerology}

In this work, we obtain precise asymptotic behaviour up to the $O(1)$ term in the context of the following interrelated topics. 

\begin{itemize}
    \item Zabrodin-Wiegmann prediction on the partition functions of planar Coulomb gas ensembles: a case study for a conditional complex Ginibre ensemble breaking the rotational symmetry.  
    \smallskip 
    \item Asymptotic behaviour of moments of the characteristic polynomial of the complex Ginibre ensemble. 
    \smallskip 
    \item Large deviation probabilities of the smallest eigenvalue of the Laguerre unitary ensemble. 
\end{itemize}

\noindent Due to a certain duality relation (Proposition~\ref{Prop_equivalence}) these topics are indeed equivalent, and readers may find a particular viewpoint most interesting based on their individual interests.
In the aforementioned topics, a specific phase transition occurs, yielding  distinct geometric/probabilistic implications in each context, see Subsection~\ref{Subsec_phase transition} for details, cf. Figures~\ref{Fig_droplet} and ~\ref{Fig_MPLDP}. 

\medskip 

Let us be more precise in introducing our models and formulations. 

\subsubsection{GinUE and its characteristic polynomials}
We begin with the complex Ginibre ensemble (GinUE) $\textbf{\textup{G}}_N$, an $N\times N$ matrix whose entries are given by independent centered complex Gaussian random variables with variance $1/N$, see \cite{BF22} for recent reviews.  
It is well known that the eigenvalues $\{ z_j \}_{j=1}^N$ of $\textbf{\textup{G}}_N$ follow the joint probability distribution 
\begin{equation} \label{GinUE jpdf}
\frac{1}{Z_N^{ \rm Gin } } \prod_{j>k=1}^N|z_j-z_k|^2 \prod_{j=1}^N e^{-N|z_j|^2} \,dA(z_j), \qquad Z_N^{ \rm Gin } = N! \frac{ G(N+1) }{ N^{N(N+1)/2} }, 
\end{equation}
where $dA(z)=d^2z/\pi$ is the area measure and $G$ is the Barnes $G$-function \cite[Section 5.17]{NIST}. 
Here, $Z_N^{ \rm Gin }$ is the normalisation constant, known as the partition function that makes \eqref{GinUE jpdf} a probability measure. 
As $N \to \infty$, the eigenvalues $\{z_j \}_{j=1}^N$ tend to be uniformly distributed on the unit disk, known as the circular law, see e.g. \cite{Ja23} for a recent progress. 
Note here that we have a simple weight function $e^{-N|z|^2}$ in \eqref{GinUE jpdf}, which enables explicit computations of the GinUE statistics. 
For instance, the evaluation of $Z_N^{\rm Gin}$ follows from Andr\'{e}ief’s formula together with the norm of the associated orthogonal polynomial (which is in this case monomial).

For the GinUE model, we shall investigate the moment of its characteristic polynomial
\begin{equation} \label{GinUE char poly}
\mathbb{E} \Big| \det(\textbf{\textup{G}}_N -z) \Big|^{2\gamma}, \qquad \gamma \ge 0.  
\end{equation}
For a fixed value of $\gamma \ge 0$, the asymptotic behaviour of \eqref{GinUE char poly} was recently obtained by Webb and Wong \cite{WW19} for the bulk regime $|z|<1$, and by Deaño and Simm \cite{DS22} for the edge regime $|z|=1+O(1/\sqrt{N})$.
This asymptotic behaviour can be applied to construct a Gaussian multiplicative chaos measure \cite{RV14}, see also \cite{Lam20}.
We mention that the moment of the characteristic polynomial in Hermitian random matrix theory has been extensively studied, see e.g. \cite{We15, BWW18, LOS18, Ch19, Kra07,CG21,CFWW21}.
Along the similar spirit of \cite{WW19,DS22}, we study the asymptotic behaviour of \eqref{GinUE char poly}, but instead of a fixed $\gamma$, we consider the exponentially varying regime $\gamma=O(N)$. 
Namely, we examine the regime where $\gamma$ is scaled as $\gamma= cN$ for a fixed parameter $c > 0$. 
In this case, a phase transition occurs at the critical value $|z|=\sqrt{c+1}-\sqrt{c}$, and we will investigate all the regimes that arise.

\subsubsection{Partition function of a determinantal Coulomb gas}

As a second formulation, we consider a conditional point process. 
For this purpose, we fix a parameter $c>0$ and consider the GinUE of size $(c+1) N \times (c+1) N$ conditioned to have deterministic eigenvalue $a \ge 0$ with multiplicity $cN$. 
Then the remaining $N$ random eigenvalues $\{z_j\}_{j=1}^N$ follow the distribution 
\begin{equation} \label{iGinUE jpdf}
\frac{1}{Z_N(a,c) } \prod_{j>k=1}^N|z_j-z_k|^2 \prod_{j=1}^N |z_j-a|^{2cN} e^{-N|z_j|^2} \,dA(z_j).   
\end{equation} 
This model is called the induced GinUE \cite{FBKSZ12,AKS23} and its partition function $Z_N(a,c)$ is also called a massive partition function in the quantum chromodynamics related literature \cite{AV03}. 
Note that if $c>0$, the ensemble is rotationally symmetric only for the case $a=0$, and in this case, again explicit computations lead to 
\begin{equation} \label{ZN(0,c) evaluation}
Z_N(0,c)= N! \frac{G(N+cN+1)}{ G(cN+1) } N^{ -(c+\frac12)N^2-\frac{1}{2}N }. 
\end{equation}

From the statistical physics viewpoint, the model \eqref{iGinUE jpdf} can be realised as a planar Coulomb gas with inverse temperature $\beta=2$ (also known as the random normal matrix model \cite{AHM11,CZ98}) and the external potential 
\begin{equation} \label{Q insertion}
Q(z)= |z|^2-2c \log|z-a|.  
\end{equation} 
We mention that potentials of this type are sometimes called Hele-Shaw potentials, see \cite{AT24} for a recent work. 
As $N \to \infty$, the ensemble \eqref{iGinUE jpdf} tends to be uniformly distributed on a certain droplet $S \equiv S_Q$, see Figure~\ref{Fig_droplet}. 
The precise shape of the droplet is characterised in \cite{BBLM15}, and it reveals a topological phase transition, see Subsection~\ref{Subsec_phase transition}.  
In the theory of Coulomb gases, the large-$N$ expansion of the free energy $\log Z_N(a,c)$ is a fundamental topic \cite{ZW06,LS17}, as the coefficients of this expansion provide essential potential theoretic/geometric properties of the model, see Subsection~\ref{Subsec_ZW} for more details. 

\begin{figure}[t]
	\begin{subfigure}{0.3\textwidth}
		\begin{center}			\includegraphics[width=\textwidth]{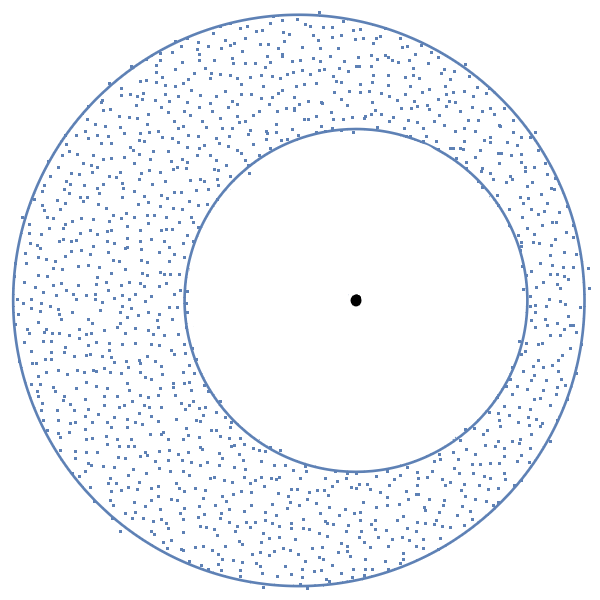}
		\end{center}
		\subcaption{Post-critical; $a=1/4$}
	\end{subfigure}	
  \quad
	\begin{subfigure}{0.3\textwidth}
		\begin{center}	
	\includegraphics[width=\textwidth]{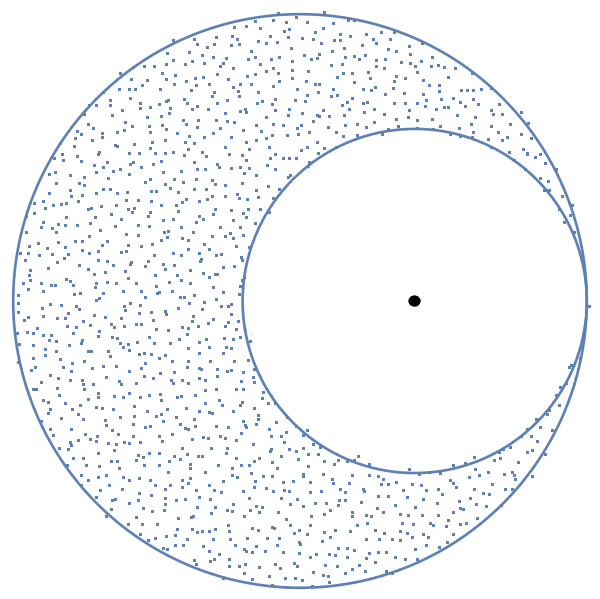}
		\end{center}
		\subcaption{Critical; $a=1/2$}
	\end{subfigure}	
 \quad
	\begin{subfigure}{0.3\textwidth}
		\begin{center}	
	\includegraphics[width=\textwidth]{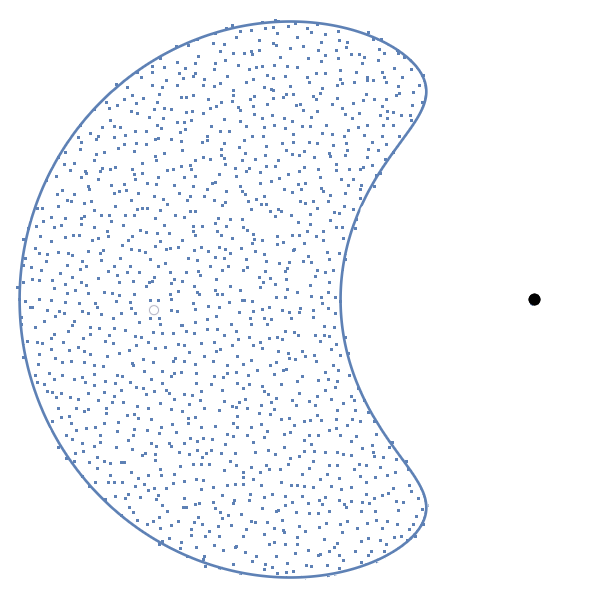}
		\end{center}
		\subcaption{Pre-critical; $a=1$}
	\end{subfigure}	
	\caption{Illustration of the droplet, where $c=9/16$. The black dot indicates the point $a.$ } \label{Fig_droplet}
\end{figure}

\subsubsection{Laguerre unitary ensemble and its smallest eigenvalue}

The third formulation is in the context of a classical Hermitian random matrix model \cite{Fo10}. 
We consider an $N\times N$ Wishart matrix, also known as the Laguerre unitary ensemble (LUE) $\textbf{W}_N= \textbf{R}_N \textbf{R}_N^*, $  where $\textbf{R}_N$ of size $N \times (\alpha+1)N$ is a rectangular complex Ginibre matrix. 
Here $\alpha \ge 0$ is the rectangular parameter.  
Then the joint probability distribution of eigenvalues $ \{ \lambda_j\}_{j=1}^N$ of $\textbf{W}_N$ is proportional to 
\begin{equation} \label{LUE}
 \prod_{j>k=1}^N  |\lambda_j-\lambda_k|^{ 2 } \prod_{j=1}^N \lambda_j^{ \alpha N } e^{ - N \sum_{j=1}^N \lambda_j }, \qquad (\lambda_N>\dots>\lambda_1 >0).  
\end{equation}
It is well known that the as $N \to \infty$, the empirical measure of the LUE follows the Marchenko-Pastur distribution
\begin{equation} \label{MP}
\frac{1}{2\pi}\frac{\sqrt{(\lambda_{+}-x) (x-\lambda_{-})  }}{x}\cdot  \mathbbm{1}_{ [\lambda_{-},\lambda_{+} ] }(x), \qquad   \lambda_{\pm}:=(\sqrt{\alpha+1}\pm 1)^2.
\end{equation} 
In this context, we focus on the statistics of the smallest eigenvalue $\lambda_1$. 
Such a statistic of the LUE, or a more general sample covariance matrix where the Gaussian entries are replaced by an i.i.d. random variable, finds several applications in random matrix theory, see e.g. \cite{WG13,BLX23,GH23} and references therein.

\begin{figure}[t]
	\begin{subfigure}{0.3\textwidth}
		\begin{center}			\includegraphics[width=\textwidth]{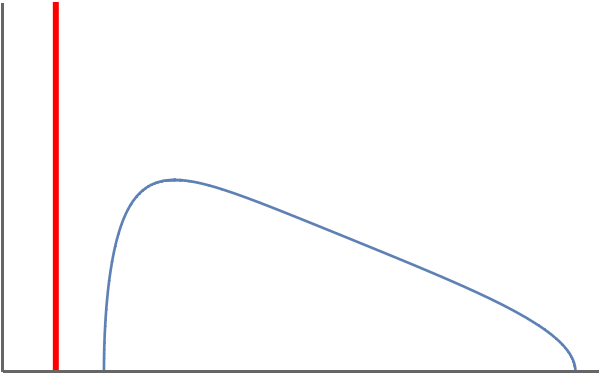}
		\end{center}
		\subcaption{Pulled; $t=\lambda_--1$}
	\end{subfigure}	
  \quad
	\begin{subfigure}{0.3\textwidth}
		\begin{center}	
	\includegraphics[width=\textwidth]{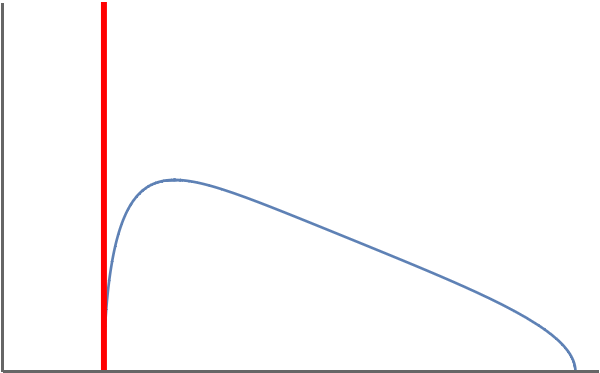}
		\end{center}
		\subcaption{Critical; $t=\lambda_-$}
	\end{subfigure}	
 \quad
	\begin{subfigure}{0.3\textwidth}
		\begin{center}	
	\includegraphics[width=\textwidth]{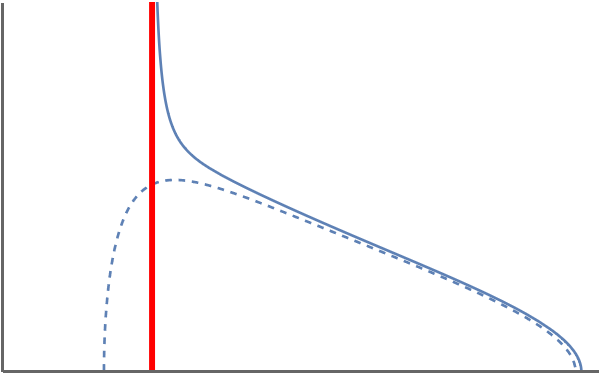}
		\end{center}
		\subcaption{Pushed; $t=\lambda_-+1$}
	\end{subfigure}	
	\caption{Illustration of the Marchenko-Pastur law \eqref{MP} and constrained spectral density \eqref{constrained LUE density}, where $\alpha=5$. Here, the vertical (full red) line indicates the hard wall $x=t$ of the LUE.  
 } \label{Fig_MPLDP}
\end{figure}

\subsubsection{Duality relation}
We now discuss the equivalence of the above three formulations. 
First, notice that by definition, the characteristic polynomial and partition functions are related as  
\begin{align}  \label{rel char poly ZN}
\mathbb{E} \Big| \det(G_N-z) \Big|^{2cN} = \frac{1}{ Z_N^{ \rm Gin } } \int_{ \C^N }\prod_{j>k=1}^N |z_j-z_k|^{2} \prod_{j=1}^{N} |z-z_j|^{2cN} e^{-N\, |z_j|^2 }  \,dA(z_j) =  \frac{ Z_N(|z|,c) }{ Z_N^{ \rm Gin } } .
\end{align}
The following equivalence was introduced in \cite{FR09,NK02}.
This is a restatement of \cite[Proposition 3.1]{DS22} after several transformations, and for the reader's convenience, we provide the details in Section~\ref{Subsec_critical}.

\begin{prop}[\textbf{Duality relation}] \label{Prop_equivalence} 
Let $\lambda_1$ be the smallest eigenvalue of the LUE in \eqref{LUE}. 
For a fixed $c>0$, we put $\alpha=1/c$. 
Then for any $x \in \R$, we have 
\begin{align}
\begin{split} \label{equivalence btw three}
\mathbb{ P }\Big[ \lambda_1 > \frac{x^2}{c} \Big] & =     e^{ -cN^2 x^2 }   \frac{ Z_N(x,c) }{ Z_N(0,c) } \bigg|_{N \to N/c}.
\end{split}
\end{align} 
\end{prop}

Due to this duality relation and the well-known asymptotic expansion \eqref{Barnes G asymp} of the Barnes $G$-function, it becomes evident that the asymptotic behaviours of the three aforementioned formulations are equivalent. 
We mention that the duality relation originates from the supersymmetry method, cf. \cite{NK02,Gr16,Fyo18}. 
Remarkably, this relation expresses the integral \eqref{GinUE char poly} over an $N\times N$ non-Hermitian random matrix in terms of an integral over a $\gamma \times \gamma$ Hermitian random matrix. 
In our present case, where $\gamma=cN$, we further make the change of variables $cN \mapsto N$, resulting in the formula \eqref{equivalence btw three} with the parameter $\alpha=1/c$ of the LUE. 
It is noteworthy that such a duality relation finds application in various contexts of random matrix theory, see e.g. \cite{LZ23,SS23,SS23a}.

\subsection{Free energy expansion; Zabrodin-Wiegmann prediction} \label{Subsec_ZW}

Among the above three formulations, let us discuss the free energy expansion from a viewpoint of a more general Coulomb gas theory. 
In general, the partition function of the random normal matrix model is given by 
\begin{align} \label{ZN V general}
Z_{N}^V := \int_{\C^N} \prod_{j>k=1}^N |z_j-z_k|^{2} \prod_{j=1}^{N}  e^{-N \, V(z_j) }  \,dA(z_j),  
\end{align} 
where $V: \C \to \R$ is a given external potential. 

To describe the asymptotic behaviour of $Z_N^V$, we recall some potential theoretic notions \cite{ST97}. 
Given a compactly supported probability measure $\mu$ on $\mathbb{C}$, the weighted logarithmic energy $I_V[\mu]$ associated with the potential $V$ is given by 
\begin{equation} \label{energy}
I_V[\mu]:= \int_{ \mathbb{C}^2 } \log \frac{1}{ |z-w| }\, d\mu(z)\, d\mu(w) +\int_{ \mathbb{C} } V \,d\mu .
\end{equation}
For a general $V$, there exists the unique minimizer $\sigma_V$ called the equilibrium measure. 
Furthermore, due to Frostman's theorem, it is of the form
\begin{equation} \label{eq msr form}
d\sigma_V(z) = \Delta V(z) \,\mathbbm{1}_{S_V}(z) \, dA, \qquad \Delta =\partial \bar{\partial}, 
\end{equation}
where the compact support $S_V$ is called the droplet. 
It is conjectured that if the droplet is \emph{connected}, the partition function $Z_N^V$ has the asymptotic expansion of the form  
\begin{equation} \label{Z expansion}
\begin{split}
\log Z_N^V &= - I_V[\sigma_V] N^2 +\frac12 N \log N + \Big( \frac{\log(2\pi)}{2}-1 - \frac12 \int_\C \log(\Delta V) \,d\sigma_V \Big) N
\\
&\quad + \frac{6-\chi}{12} \log N + \frac{\log(2\pi)}{2} + \chi \, \zeta'(-1) + \FF_V + o(1).
\end{split}
\end{equation}
Here $\chi$ is the Euler characteristic of the droplet and $\zeta$ is the Riemann zeta function. 
We refer the reader to \cite[Section 1.1]{BKS23} and \cite[Section 5.4]{BF22} for the development of the expansion \eqref{Z expansion}. 

Among the variety of literature on the expansion \eqref{Z expansion}, we mention that Lebl\'{e} and Serfaty \cite{LS17} proved the expansion \eqref{Z expansion} up to the order $O(N)$ in the context of a more general $\beta$ ensemble, see also \cite{Se23,BBNY19} for quantitative error bounds. 
The topology-dependence of the $O(\log N)$ term was introduced in the work of Jancovici et al. \cite{JMP94,TF99} through exactly solvable examples such as (induced) Ginibre and spherical ensembles. 
The $O(1)$ term reflects a conformal geometric property of the droplet. More precisely, in \cite{ZW06}, Zabrodin and Wiegmann made use of the Ward's identities in conformal field theory (see e.g. \cite[Appendix 6]{KM13}) and proposed a remarkable prediction, suggesting that the term $\mathcal{F}_V$ can be expressed in terms of the zeta-regularized determinant of the exterior droplet. 
In particular, for a quasi-harmonic potential $V$ (i.e. $\Delta V$ is a constant), this conjecture reads 
\begin{equation}
\FF_V= -\frac12 \log \textup{det}_\zeta ( \Delta_{ \C \setminus S_V } ),
\end{equation}
cf. Remark~\ref{Rem_detLap}. 
For a general potential, one needs an additional term related to the first correction of the global density. 

The expansion of the form \eqref{Z expansion} was obtained in a recent work \cite{BKS23} for a radially symmetric potential $V$ with $\Delta V > 0$ in $\mathbb{C}$. This strictly sub-harmonic assumption is crucial in \cite{BKS23}, as it leads to a droplet that is either a disc or an annulus. The asymptotic behaviour of the partition function $Z_N^V$ has been further investigated in \cite{ACC23} to include the case where the ensemble exhibits a spectral gap (i.e. $V$ can be such that $\Delta V<0$). In particular, it was shown in \cite{ACC23} that if the droplet has multiple components, a non-trivial oscillatory term (a new ``displacement term'') emerges in the $O(1)$ term in the expansion, indicating that the expansion of the form \eqref{Z expansion} does not hold. Furthermore, in \cite{ACC23}, the case with the harmonic measure perturbation of the order $O(1/N)$ in the potential and Fisher-Hartwig singularities \cite{Kl17} has been investigated, which are closely related to the fluctuation theorem \cite{ACC23a}.

\begin{rem}[Potential with a hard edge]
In the above discussions, we have focused on the case that $V$ is supported on the whole complex plane, leading to a droplet with a soft edge. 
In contrast, if $V$ is confined to a subset of the droplet, the expansion of the free energy takes on a notably different form. This hard edge regime finds applications in different contexts including the counting statistics \cite{BC22,Ch22,ACCL22,ACCL23,ABES23} and hole probabilities \cite{Fo92,AS13,Ad18,GN18,Ch23,Ch23a}, cf. Remark~\ref{Rem_sqrtN}.
Moreover, when $V$ is supported on a Jordan curve, the asymptotic behaviour of the free energy has also been investigated, see e.g. \cite{WZ22,CJ23,JV23} and references therein. 
\end{rem}


In this work, we obtain a precise expansion in the form of \eqref{Z expansion} for the potential given by \eqref{Q insertion}. 
To our knowledge, this provides the first non-radially symmetric example (aside from more explicitly solvable models, such as the elliptic GinUE, which is solvable by virtue of classical special function theory) that confirms this conjecture.
Indeed, the elliptic GinUE \cite[Section 2.3]{BF22}, indexed by the non-Hermiticity parameter $\tau \in [0,1)$, is the only known non-radially symmetric example for which the free energy expansion in the form of \eqref{Z expansion} is known. However, in this case, the free energy is simply the same as that of the GinUE up to an additive constant of $ \frac{N}{2} \log(1-\tau^2)$.

Due to the lack of the rotational symmetry, it requires a different approach compared to \cite{BKS23,ACC23}, and we implement the partial Schlesinger transform \cite{BL08} to refine the Riemann-Hilbert analysis in \cite{BBLM15,LY23,KLY23}. 
Furthermore, as explained in the previous subsection, it follows from Proposition~\ref{Prop_equivalence} that this result implies the asymptotic behaviour of moments of the characteristic polynomials of the GinUE, as well as the large-deviation probabilities of the smallest eigenvalue of the LUE.

\subsection{Phase transition} \label{Subsec_phase transition}

As previously mentioned, the model we consider undergoes an interesting phase transition. We now precisely describe such a transition.
Let $S \equiv S_Q$ be the droplet in \eqref{eq msr form} associated with the potential $Q$ given in \eqref{Q insertion}. 
The droplet reveals a topological phase transition at the critical value
\begin{equation} \label{c cri a cri}
c_{\rm cri}:= \frac{(1-a^2)^2}{4a^2}, \qquad a_{ \rm cri }:= \sqrt{c+1}-\sqrt{c} .
\end{equation}
In the post- and pre-critical regimes, the parameters $a$ and $c$ are assumed to be fixed (i.e. independent of $N$).
The droplet is then given as follows \cite[Section 2]{BBLM15}. 
\begin{itemize}
    \item \textbf{Post-critical regime}: $c < c_{ \rm cri }$, i.e. $a< a_{\rm cri}.$
    In this case,  the droplet is given by 
    \begin{equation}
    S=\overline{\mathbb{D}(0,\sqrt{1+c})} \setminus \mathbb{D}(a,\sqrt{c}),
    \end{equation}
    where $\mathbb{D}(z,r)$ is a disk with center $z$ and radius $r>0.$
    \smallskip 
    \item \textbf{Pre-critical regime}\footnote{Compared to \cite{BBLM15}, we have replaced the notations: $\alpha \to q$ and $\rho \to R$.}: $c > c_{ \rm cri }$, i.e. $a > a_{\rm cri}.$  In this case, the droplet is a simply connected domain whose boundary is given by the image of the unit circle under the conformal map 
\begin{equation} \label{f conformal}
f(z)=R\,z-\frac{\kappa}{z-q}-\frac{\kappa}{q}, \qquad R=\frac{1+a^2q^2}{2a q}, \qquad \kappa=\frac{(1-q^2)(1-a^2q^2)}{2aq}.
\end{equation}
Note that $R>0$ is the conformal radius of the droplet. 
Here, $q \equiv q(a)$ satisfies $f(1/q)=a$ and it is given by the  unique solution of the algebraic equation  
  \begin{equation} \label{q equation}
q^6 -\Big( \frac{a^2+4c+2}{2a^2} \Big) q^4+\frac{1}{2a^4} =0  
 \end{equation}
such that $0<q<1$ and $\kappa>0.$ 
\end{itemize}

See Figure~\ref{Fig_droplet} for the shape of the droplet. 
From the above description of the droplet, one can see that the Euler characteristic $\chi$ of the droplet $S_Q$ is given by 
\begin{equation} \label{Euler char droplet}
\chi= \begin{cases}
0 &\textup{for the post-critical case,}
\smallskip
\\
1 &\textup{for the pre-critical case}.
\end{cases}
\end{equation}
This will play an important role in the free energy expansion. 

Let us mention that there are some more examples of the droplets revealing a topological phase transition, see e.g. \cite{By23a,ABK21,BM15,CK22}. 
A notable feature of the phase transition is the emergence of a singular boundary point. In our present case, it is a merging (double) point that falls into the class of Sakai's regularity theory \cite{Sa92,LM16}. 
Such singular boundaries are of particular interest from the viewpoint of the non-standard universality classes, and have been studied for several  different models, see e.g. \cite{BGM18,BK12,MS19,AKMW20,BLY21} and references therein.

\medskip 
Beyond the post- and pre-critical regimes, it is also natural to study the behaviour of the ensemble \eqref{iGinUE jpdf} in the critical regime, see \cite{KLY23} for a recent work on the local statistics. 
Let us first define the critical regime. 

\begin{defi}[Critical scaling regime]
The critical case corresponds to the scaling regime
$$
1-a(a+2\sqrt{c})=O(N^{-2/3}).
$$  
We consider a parameter $s \in \R$ describing the critical regime, see \cite[Eq.(1.36)]{BBLM15}. 
Then the parameter $a$ satisfies 
\begin{align} \label{a cri asymp}
a = a_{\rm cri}- \frac{ ( \sqrt{c+1}-\sqrt{c})^{1/3} }{ 2c^{1/6} (c+1)^{1/6} }\frac{s}{N^{2/3}} +O\Big(\frac{1}{N^{4/3}}\Big).
\end{align} 
As $s$ varies, it parametrises the transition between the critical, post-critical, and pre-critical regimes, with the limits $s \to +\infty$ and $s \to -\infty$ reducing to the post-critical and pre-critical regimes, respectively. See Remark~\ref{Rem_Topoly transition} for a discussion of this transition behaviour in the free energy expansion.
\end{defi}

Turning to the LUE statistics, note that in the setup in Proposition~\ref{Prop_equivalence} with $x=a \ge 0$, it follows from \eqref{MP} that
\begin{equation}
\frac{x^2}{c} <\lambda_- = \Big( \sqrt{\frac{c+1}{c}}-1\Big)^2, \quad \textup{if and only if} \quad a< \sqrt{c+1}-\sqrt{c}=a_{ \rm cri }. 
\end{equation}
Hence, the topological phase transition of the droplet can be naturally interpreted from the viewpoint of lower and upper large deviation probabilities.
This transition is also called the push-to-pulled transition, see Figure~\ref{Fig_MPLDP}. 
Furthermore, while the order $O(N^{-2/3})$ in the critical scaling regime might be less intuitive in the two-dimensional model \eqref{iGinUE jpdf}, it becomes clear if we interpret this in the context of the LUE given the standard square root decay of the Marchenko-Pastur law \eqref{MP} at the soft edge.
In this context, a universal pulled-to-pushed transition of the third order has been observed, see e.g. \cite{MS14,CFLV18}. 
See Remark~\ref{Rem_Topoly transition} for more details.


\section{Main results}

In this section, we present our main results. 
Recall that the partition function $Z_N(a,c)$ is given by 
\begin{align}
 Z_N(a,c) \equiv Z_N^Q  := \int_{\C^N} \prod_{j>k=1}^N |z_j-z_k|^{2} \prod_{j=1}^{N}  e^{-N Q(z_j) }  \,dA(z_j), 
\end{align}
where $Q$ is given by \eqref{Q insertion}.
As discussed in Subsection~\ref{Subsec_ZW}, the logarithmic energy \eqref{energy} should appear in the leading term of the free energy expansion.
We first evaluate the energy associated with the potential $Q$ explicitly. 

\begin{prop}[\textbf{Evaluation of the energy}] \label{Prop_energy eval}
Let $Q$ be given by \eqref{Q insertion}. Then we have the following. 
\begin{itemize}
    \item For the post-critical case, we have  
    \begin{equation} \label{energy post}
I_Q[\sigma_Q]= \mathcal{I}^{ \rm post }(a,c) := \frac{3}{4}+\frac{3c}{2}+\frac{c^2}{2}\log c -\frac{(c+1)^2}{2}\log(c+1) - ca^2. 
\end{equation}
\item For the pre-critical case, we have 
\begin{align}
\begin{split} \label{energy pre}
I_Q[\sigma_Q] =  \mathcal{I}^{ \rm pre }(a,c) & :=   \frac{3}{8} + \frac{a^2}{8} + \frac{3}{8a^2q^4} - \frac{5}{8q^2} + \Big(\frac{3}{4} + \frac{a^2}{8}\Big)a^2q^2 - \frac{3a^4q^4}{8}
\\
&\quad +\log(2aq)+2 c \log(2aq^2) +  \log \frac{ (1+a^2q^2-2a^2q^4)^{c^2} }{ (1+a^2q^2)^{ (c+1)^2 } }   ,
\end{split}
\end{align}
where $q = q(a)$ is given by \eqref{q equation}.
\end{itemize}
Furthermore, for a fixed $c>0$, suppose that $a> a_{ \rm cri }$. Then we have  
\begin{align} \label{energy inequality}
 \mathcal{I}^{ \rm post }(a,c) \le \mathcal{I}^{ \rm pre }(a,c), 
\end{align}
where the equality holds when $a=a_{\rm cri}$. 
\end{prop}

See Figure~\ref{Fig_energy} for the graphs of the energy. We mention that the inequality \eqref{energy inequality} plays an important role in Theorem~\ref{Thm_LUE} below. 
The evaluation of energy is closely related to the equilibrium measure problem. Using potential theory and conformal mapping methods, Proposition~\ref{Prop_energy eval} will be established in Section~\ref{Section_energy}.

\begin{figure}[t]
    \centering
    \includegraphics[width=0.5\textwidth]{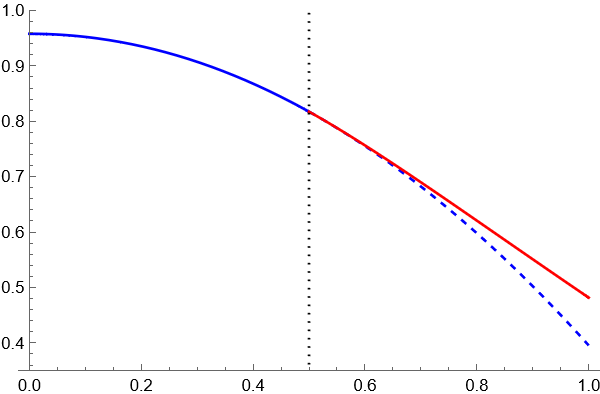}
    \caption{The plot shows the graph of the energy $a \mapsto I_Q[\sigma_Q]$, where $c=9/16$. Here, the dotted vertical line represents $a=a_{\rm cri}=1/2$. The graph (full line) for $a<a_{\rm cri}$ follows \eqref{energy post}, while for $a>a_{\rm cri}$ it follows \eqref{energy pre}. The dotted line for $a>a_{\rm cri}$ is the continuation of the graph \eqref{energy post}. }
    \label{Fig_energy}
\end{figure}

Recall that the Bernoulli number $B_k$ is given in terms of the generating function as 
\begin{equation}
\frac{t}{e^t-1} = \sum_{k=0}^\infty B_k \frac{t^k}{k!}.
\end{equation}
We are now ready to state our main result. 

\begin{thm}[\textbf{Free energy expansion for the post- and pre-critical cases}] \label{Thm_main}
Let $Q$ be given by \eqref{Q insertion}. Then as $N \to \infty,$ we have 
\begin{equation}
\begin{split}  \label{Z expansion main}
\log Z_N(a,c) & = -I_Q[\sigma_Q] N^2 + \frac12 N \log N + \Big( \frac{\log(2\pi)}{2}-1 \Big) N 
\\
&\quad + \frac{6-\chi}{12} \log N + \frac{\log (2\pi)}{2} +\chi \, \zeta'(-1) +  \mathcal{F}(a,c) + \mathcal{E}_N , 
\end{split}
\end{equation}
where $I_Q[\mu_Q]$ is the energy given in Proposition~\ref{Prop_energy eval}, $\chi$ is the Euler characteristic of the droplet $S_Q$ given in \eqref{Euler char droplet} and $\zeta$ is the Riemann zeta function.
Here $\mathcal{F}(a,c)$ and the error term $\mathcal{E}_N$ are given as follows.
\begin{itemize}
    \item For the post-critical case, we have 
    \begin{equation} \label{det Lap post}
    \mathcal{F}(a,c) = \mathcal{F}^{ \rm post }(a,c) : =  \frac{1}{12} \log \Big( \frac{c}{1+c}\Big)
    \end{equation}
    and 
    \begin{equation} \label{error for post}
    \mathcal{E}_N =  \sum_{k=1}^M \bigg( \frac{B_{2k}}{ 2k(2k-1) } \frac{1}{N^{2k-1}} + \frac{ B_{2k+2} }{ 4k(k+1) } \Big( \frac{1}{(c+1)^{2k}}-\frac{1}{c^{2k}}\Big) \frac{1}{N^{2k}} \bigg) +O(\frac{1}{N^{2M+1}})
    \end{equation}
    for any $M>0$, where $B_k$ is the Bernoulli number. \smallskip 
    \item For the pre-critical case, we have 
    \begin{equation} \label{det of Lap pre}
     \mathcal{F}(a,c) = \mathcal{F}^{ \rm pre }(a,c) : =   \frac{1}{24} \log \Big( \frac{(1+a^2q^2-2a^2q^4)^4}{ (1+a^2q^2)^4 (1-q^2)^3 (1-a^4q^6) } \Big)   
    \end{equation}
    and $\mathcal{E}_N = O(N^{-1}).$ 
\end{itemize}
\end{thm}

\medskip 

As explained above, up to the $O(N)$ term, Theorem~\ref{Thm_main} is a special case of \cite[Corollary 1.1]{LS17}. 
It is obvious that the entropy in the $O(N)$ term of \eqref{Z expansion} vanishes since $\Delta Q = 1$. The constant term $\FF(a, c)$ coincides with the prediction \eqref{Z expansion}, as we discuss now.


\begin{rem}[Regularized determinant of Laplacian] \label{Rem_detLap}

Let $0<\lambda_{1,D} \le \lambda_{2,D} \le \dots$ be the eigenvalues of the Dirichlet Laplacian $\Delta$ on a domain $D \subset \C$. 
Then the spectral zeta function is defined by 
\begin{equation}
\zeta_{\Delta} (s;D):= \sum \lambda_{j,D}^{-s}. 
\end{equation}
This is a building block to define the zeta-regularized determinant of $\Delta$: 
\begin{equation}
\textup{det}_\zeta( \Delta_D ):= \exp(-\zeta_\Delta'(0;D)),
\end{equation}
see e.g. \cite{HZ99} for more details. 
The spectral determinant can also be expressed in terms of several different domain functionals in conformal geometry, such as Brownian loop measure \cite{LSW03,LW04}, Loewner energy \cite{Wa19a, Wa19b}, and the Grunsky operator \cite{TT06,JV23}.
It is also used to describe large deviation principles (see e.g. \cite{Wa19a,PW23}), aligning with the same spirit as our result, especially from the viewpoint of the LUE statistics. 

If the derivative of the potential $\partial V$ is rational, $\textup{det}_\zeta( \Delta_{\C \setminus S_V} )$ can be made more explicit as discussed in \cite[Section 6.1]{ZW06}.
As a consequence, for the post-critical case, it can be observed that  
\begin{equation}
 \log \textup{det}_\zeta( \Delta_{ \C \setminus S } )= -\frac{1}{6} \log\Big( \frac{c}{1+c}\Big) 
\end{equation}
since the boundary of the droplet is a union of two circles.
For the pre-critical case\footnote{The conformal map $z(w)$ was introduced below Eq.(5.22) of \cite{ZW06}, which in our case is $f(z)$.  }, let
\begin{equation} \label{z pm}
z_\pm = q \pm i\sqrt{ \frac{\kappa}{R} }  
\end{equation}
be the critical points of the conformal map $f$ given in \eqref{f conformal}. 
Then we have 
\begin{align}
\begin{split}
\log \textup{det}_\zeta ( \Delta_{ \C \setminus S } ) &= \frac1{12} \log \Big( \frac{R^4 \,f'(1/z_+) f'(1/z_-) }{ f'(1/q)^2 } \Big) 
= - \frac{1}{12} \log \Big( \frac{ (1+a^2q^2-2a^2q^4) ^4 }{ (1-q^2)^3 (1-a^4q^6) (1+a^2q^2)^4 } \Big) .
\end{split}
\end{align}
Therefore, one can observe that $\mathcal{F}(a,c)=-\frac12 \log \textup{det}_\zeta ( \Delta_{ \C \setminus S } )$, confirming the prediction of Zabrodin and Wiegmann. 

We mention that in our present case, since we consider the quasi-harmonic potential $\Delta Q \equiv 1$, the equilibrium measure \eqref{eq msr form} has a flat metric. 
On the other hand, if $\Delta Q$ is not a constant, the determinant of Laplacian with respect to a non-trivial conformal metric or its conformal transformation law can be obtained via the Polyakov-Alvarez conformal anomaly formula \cite{Po81,Al83}.
\end{rem}


\begin{rem}[Absence of the $O(\sqrt{N})$ term for $\beta=2$] \label{Rem_sqrtN}
As a side remark, it is worth mentioning the belief for a general $\beta$ ensemble that there exists an $O(\sqrt{N})$ term, with a coefficient proportional to $\log(\beta/2)$ called the surface tension. 
This conjecture was made in an unpublished note of Lutsyshin and first appeared in \cite{CFTW15}, see also \cite{SS15,Se23,Kl16}. 
However, since the coefficient is expected to be proportional to $\log (\beta/2)$, the absence of the $O(\sqrt{N})$ term in the determinantal case $\beta=2$ is, in this context, a statement of prediction. 
The absence of the $O(\sqrt{N})$ term has been verified for the rotational symmetry case \cite{BKS23,ACC23}. However, one might question whether this absence truly results from $\beta=2$ or from rotational symmetry. Nonetheless, as per Theorem~\ref{Thm_main}, one can observe that even without the rotational symmetry, there is no $O(\sqrt{N})$ term for $\beta=2$. (It is worth noting however that the $O(\sqrt{N})$ term does arise for $\beta=2$ when considering the hard wall constraints of the potential \cite{ACCL22,BC22,Ch22,Ch23}).
\end{rem}

\medskip 

We now discuss the critical regime. 
It can be expected from the duality relation (Proposition~\ref{Prop_equivalence}) that the expansion of the free energy in the critical regime is closely related to the Tracy-Widom distribution: 
\begin{equation} \label{def of TW distribution}
 F_{ \rm TW }(t):= \exp\Big( -\int_t^{\infty} (x-t) \boldsymbol{q}(x)^2\,dx \Big),
\end{equation}
where $\boldsymbol{q}$ is the Hastings–McLeod solution to Painlevé II equation
\begin{equation}
\boldsymbol{q}''(s)=s\boldsymbol{q}(s)+2\boldsymbol{q}(s)^3, \qquad \boldsymbol{q}(s) \sim \Ai(s) \quad \mbox{as }s\to \infty.
\end{equation}
Then we have the following. 

\begin{prop}[\textbf{Free energy expansion for the critical regime}] \label{Prop_critical expansion}
For a given fixed parameter $c>0$, let $a$ be scaled as \eqref{a cri asymp}. 
Then as $N \to \infty,$ we have
\begin{align}
\begin{split} \label{ZN exp critical}
\log Z_N(a,c) &= -\Big( \frac{3}{4}+\frac{3c}{2}+\frac{c^2}{2}\log c -\frac{(c+1)^2}{2}\log(c+1) - ca^2 \Big) N^2 
\\
&\quad +\frac12 N \log N +\Big( \frac{\log(2\pi)}{2}-1 \Big)N  +\frac12 \log N
\\
&\quad + \frac{\log(2\pi)}{2}+\frac{1}{12} \log \Big( \frac{c}{1+c}\Big)  + \log F_{ \rm TW }( c^{-2/3}s ) +O(N^{-2/3}) ,
\end{split}
\end{align}
where $F_{ \rm TW }$ is the Tracy-Widom distribution. 
\end{prop}
This will be shown in Section~\ref{Subsec_critical}. 
Note that the expansion \eqref{ZN exp critical} is not of the form \eqref{Z expansion}. Namely, by the scaling \eqref{a cri asymp}, we have 
\begin{equation}
ca^2 N^2= c a_{ \rm cri }^2 N^2+ \widetilde{C}_1(s) N^{4/3} + \widetilde{C}_2(s) N^{2/3} + \widetilde{C}_3(s) +O(N^{-2/3})
\end{equation}
for some constants $\widetilde{C}_k$ ($k=1,2,3$). 

The Painlevé II critical asymptotic behaviour of the associated planar orthogonal has been discovered in \cite{BBLM15,KLY23}. 
However, the asymptotic behaviour presented in \cite{BBLM15,KLY23} is not enough to derive Proposition~\ref{Prop_critical expansion}, particularly to derive the $O(1)$ term. 
On the other hand, Proposition~\ref{Prop_critical expansion} can be readily derived utilizing the duality relation (Proposition~\ref{Prop_equivalence}), the Marchenko-Pastur law \eqref{MP}, and the well-established edge universality of the random Hermitian matrix ensemble \cite{DKMVZ99}. 
This aligns with the probability theoretic intuition of the free energy expansion: the law of large numbers (determining the position of the left edge of the Marchenko-Pastur law) gives rise to the leading order of the free energy, while fluctuations (governed by the Tracy-Widom distribution) contribute to the constant term.

\begin{rem}[Free energy expansion under the topology transitions] \label{Rem_Topoly transition}
Recall the well-known tail behaviour of the Tracy-Widom distribution: as $x \to +\infty,$
\begin{align}
F_{\rm TW}(-x)& = \frac{ 2^{ 1/24 } e^{ \zeta'(-1) } }{ x^{1/8} } e^{-x^3/12} \Big( 1+ \frac{3}{ 2^6x^3 } +O(x^{-6}) \Big),
\\
1-F_{ \rm TW }(x) & = \frac{1}{32 \pi x^{3/2} } e^{ -4x^{3/2}/3  } \Big(1+O(x^{-3/2}) \Big). 
\end{align} 
Using this, we have 
\begin{equation}
\lim_{s \to +\infty } \log F_{\rm TW}(c^{-2/3}s) =0. 
\end{equation}
Thus in this limit, \eqref{ZN exp critical} matches with Theorem~\ref{Thm_main} for the post-critical regime. 
On the other hand, in the opposite limit, we have
\begin{equation}
\log F_{\rm TW}(c^{-2/3}s) =  \frac{1}{24} \log 2 + \zeta'(-1) -\frac18 \log |s| -\frac{|s|^3}{12}  + O(|s|^{-3}) , \qquad s \to -\infty.  
\end{equation}
Note that by \eqref{a cri asymp}, at least formally, the proper scaling for the pre-critical regime should be $s=O(N^{2/3})$ in the critical scaling. 
With this scaling, one can notice the additional term
\begin{equation}
-\frac{1}{12} \log N +\zeta'(-1) 
\end{equation}
appearing in Theorem~\ref{Thm_main} for the pre-critical regime when $\chi=1.$

One can observe such a transition in the opposite direction, from the pre-critical regime. Namely, in the scaling regime \eqref{a cri asymp}, by \eqref{q equation}, we have
\begin{align*}
q= 1 + \frac{ s }{ 4c^{1/6}(c+1)^{1/6} ( \sqrt{c+1}- \sqrt{c} )^{2/3} } N^{-2/3} +O(N^{-1}). 
\end{align*}
This gives that $\mathcal{F}^{ \rm pre }$ in \eqref{det of Lap pre} has the asymptotic expansion 
\begin{equation}
\mathcal{F}^{ \rm pre }(a,c) = \frac1{12} \log N +O(1).  
\end{equation}
Hence, one can again observe the additional $\frac{1}{12} \log N$ term.  
\end{rem}

\bigskip

We now turn back to the other formulations of the problem. 
First of all, as a consequence of Theorem~\ref{Thm_main}, we have the asymptotic behaviour of the moments of characteristic polynomial.
For this, we shall use the notations in  \cite[Eqs.(2.17),(4.6)]{BBLM15}. Let 
\begin{equation} \label{def of F} 
F(z)=\frac{1}{2R}\Big[z+|\beta|+\sqrt{(z-\beta)(z-\overline\beta)}\Big]  
\end{equation}
be the inverse of $f$ in \eqref{f conformal}, where 
\begin{equation} \label{def of beta}
\beta = f(z_+)=  Rq-\frac{\kappa}{q} + 2i \sqrt{\kappa R}  
\end{equation}
is a critical value of $f$.

\begin{thm}[\textbf{Moments of characteristic polynomial of the GinUE}] \label{Thm_ch poly GinUE}
Let $\textbf{\textup{G}}_N$ be the complex Ginibre matrix of size $N$ and let $c>0$ be fixed. 
Then as $N \to \infty,$ 
\begin{equation} \label{moment expansion main}
\mathbb{E} \Big| \det(\textbf{\textup{G}}_N -z) \Big|^{2cN} =N^{ \frac{1-\chi}{12} } \, e^{ (\chi-1)\zeta'(-1) } \,\mathcal{G}(|z|)   \, \exp\Big( \mathcal{H}(|z|) N^2  +\widetilde{\mathcal{E}}_N \Big)   
\end{equation}
where $\mathcal{H}(z)$, $\mathcal{G}(z)$ and $\widetilde{\mathcal{E}}_N$ are given as follows.
\begin{itemize} 
    \item If $|z|<\sqrt{c+1}-\sqrt{c}$, we have $\chi=0$, 
\begin{equation}
\mathcal{H}(z)= \mathcal{H}^{\rm post}(z):=c\,z^2-\frac{3c}{2}+\frac{(c+1)^2}{2}\log(c+1) -\frac{c^2}{2}\log c 
\end{equation}
and 
\begin{equation}
\mathcal{G}(z) = \mathcal{G}^{ \rm post }(z) : = \Big( \frac{c}{1+c} \Big)^{ \frac{1}{12}  }. 
\end{equation}
Here for any $M>0$,
\begin{equation}
\widetilde{\mathcal{E}}_N= \sum_{k=1}^M  \frac{ B_{2k+2} }{ 4k(k+1) } \Big( \frac{1}{(c+1)^{2k}}-\frac{1}{c^{2k}}-1\Big) \frac{1}{N^{2k}} +O(\frac{1}{N^{2M+2}}) .
\end{equation}
  \item If $|z|>\sqrt{c+1}-\sqrt{c}$, we have $\chi=1$,  
\begin{align}
\begin{split}
\mathcal{H}(z) =\mathcal{H}^{ \rm pre }(z) & :=\frac{3}{8} - \frac{z^2}{8} - \frac{3F(z)^4}{8z^2} +\frac{5 F(z)^2}{8} - \Big(\frac{3}{4} + \frac{z^2}{8}\Big)\frac{z^2}{F(z)^2} + \frac{3}{8} \frac{z^4}{F(z)^4}
\\
& \quad + \log \Big(  \frac{ F(z)^{2c^2-1} ( F(z)^2+z^2 )^{(c+1)^2} }{  (2z)^{2c+1}  (F(z)^4 +z^2 F(z)^2 -2z^2)^{c^2}   }    \Big)
\end{split}
\end{align}
and 
\begin{equation}
\mathcal{G}(z)= \mathcal{G}^{ \rm pre }(z) := \Big(  \frac{ F(z)^{4} (F(z)^4+z^2F(z)^2 -2z^2)^4 }{ (F(z)^2+z^2)^4 (F(z)-1)^3 (F(z)^6-z^4) } \Big)^{ \frac{1}{24} }  , 
\end{equation}
where $F$ is given by \eqref{def of F}. Here $\widetilde{\mathcal{E}}_N=O(1/N).$
\end{itemize}
\end{thm}

Note that Theorem~\ref{Thm_ch poly GinUE} immediately follows from Proposition~\ref{Prop_energy eval}, Theorem~\ref{Thm_main} and Lemma~\ref{Lem_Reference Z}. Here we have used also the fact that $q=1/F(a)$.

\begin{rem}[Comparison with the small insertion] \label{Rem_small insertion}
Let us stress again that the case $c=O(1/N)$ was studied in \cite{WW19} for the bulk case $|z|<1$ and also in \cite{DS22} for the edge case $|z|=1+O(1/\sqrt{N})$.  
From the viewpoint of the induced model \eqref{iGinUE jpdf}, the regime $c=O(1/N)$ represents the case in which the point charge insertion is finite. 
In this situation, the conditional process does not lead to macroscopic changes, and consequently, the droplet remains the unit disk of the circular law.
It particular, it was shown in \cite{WW19} that for the bulk case $|z|<1$, 
\begin{equation} \label{WW formula}
\mathbb{E} \Big| \det(\textbf{\textup{G}}_N -z) \Big|^{2\gamma} = N^{ \frac{\gamma^2}{2} } \exp\Big( \gamma N (|z|^2-1) \Big) \frac{(2\pi)^{ \gamma/2 }}{ G(\gamma+1) } (1+o(1)). 
\end{equation}
Note that the leading order in the exponent is $O(N)$, contrasting with the order of $O(N^2)$ in \eqref{moment expansion main}. 
This difference arises from the fact that the $O(N^2)$ term originates from the energy of the equilibrium measure, which is a macroscopic quantity.
Nonetheless, one can observe, at least up to a multiplicative constant, that the asymptotic formula \eqref{WW formula} coincides with our formula \eqref{moment expansion main} by simply setting $c=\gamma/N$:
\begin{equation}
 \exp\Big( \mathcal{H}^{ \rm post }(|z|) N^2 \Big) \Big|_{ c=\gamma/N } = N^{ \frac{\gamma^2}{2} } \exp\Big( \gamma N (|z|^2-1) \Big) \cdot O(1). 
\end{equation}
However, the multiplicative constant term does not match. 
Nonetheless, if we make use of the Barnes $G$-function in the asymptotic formula, by \eqref{Barnes G asymp}, one can see that the asymptotic behaviour
\begin{equation}
\mathbb{E} \Big| \det(\textbf{\textup{G}}_N -z) \Big|^{2\gamma}= N^{ -\gamma N  }  \exp\Big( \gamma N\,|z|^2  \Big)  \frac{ G(\gamma+N+1) }{ G(\gamma+1) G(N+1) } (1+o(1)) 
\end{equation}
matches both \eqref{WW formula} and \eqref{moment expansion main} for $\gamma$ fixed and $\gamma=cN$, respectively.
\end{rem}

\medskip 

We now discuss the large deviation of the LUE smallest eigenvalue. 
In general, the large deviation principles of the extremal eigenvalues of random matrices have been extensively studied in the literature. 
For instance, the statistics of the maximal eigenvalues of the Gaussian ensembles have been studied in \cite{BDG01,DM06, DM08}. 
Also, for the Laguerre ensembles, which are closely related to our present case, there has been extensive work on large deviation probabilities for both the smallest and largest eigenvalues \cite{AGKWW14,KC10, PS16,VMB07,WG13,WG14,WKG15,MV09}, see Remark~\ref{Rem_KC formula}. 
For such integrable random matrices, the Coulomb gas approach has been mainly implemented, and essentially, the large deviation speed and rate can be derived by computing the energy associated with the potential under the presence of a hard wall. An advantage of this approach is that it can be applied to a general $\beta$ ensemble \cite{Fo12,FW12,DR16}. 
However, it is limited in deriving the leading-order asymptotic behaviour. 
In addition, within the context of large deviation probabilities of the extremal eigenvalues, there is a universal pulled-to-pushed transition of the third order, see e.g. \cite{MS14,CFLV18}. 
We also point out that such large deviation probabilities are closely related to the maximal height of $N$ non-intersecting Brownian excursions \cite{SMCF13,FMS11,KIK08}.



\begin{thm}[\textbf{Large deviation probabilities of the smallest eigenvalues of the LUE}]  \label{Thm_LUE}
Let $\lambda_1$ be the smallest eigenvalue of the LUE in \eqref{LUE}. Then as $N\to \infty$, we have the following. 
\begin{itemize}
    \item \textup{\textbf{(Pulled-regime)}} If $t < \lambda_-$, we have  
    \begin{equation} \label{LDP pulled}
    \log  \mathbb{ P }\Big[ \lambda_1 > t \Big] =  O(N^{-\infty}).
    \end{equation}
    Here, $O(N^{-\infty})$ means $O(N^{-m})$ for any positive integer $m.$
    \smallskip 
    \item \textup{\textbf{(Pushed-regime)}} 
    If $t > \lambda_-$, we have  
    \begin{equation} \label{LDP pushed}
    \log  \mathbb{ P }\Big[ \lambda_1 > t \Big] = - \Phi(t;\alpha) N^2 -\frac{1}{12} \log (\alpha N) +\zeta'(-1) + \Psi(t;\alpha) +O(\frac{1}{N}),
    \end{equation}
    where 
    \begin{align}
    \Phi( t;\alpha) &=  \alpha^2 \Big( \mathcal{I}^{ \rm pre }( \sqrt{t/\alpha} ,1/\alpha)- \mathcal{I}^{ \rm post }( \sqrt{t/\alpha} ,1/\alpha) \Big) ,
    \\
     \Psi(t;\alpha) & = \mathcal{F}^{ \rm pre }( \sqrt{t/\alpha} ,1/\alpha) -  \mathcal{F}^{ \rm post }(\sqrt{t/\alpha} ,1/\alpha)  . 
    \end{align}
    Here, $\mathcal{I}^{ \rm pre }$, $\mathcal{I}^{ \rm post }$, $ \mathcal{F}^{ \rm pre }$ and $\mathcal{F}^{ \rm post }$ are given by \eqref{energy post}, \eqref{energy pre} \eqref{det Lap post} and \eqref{det of Lap pre}.
\end{itemize}
\end{thm}
 
Theorem~\ref{Thm_LUE} follows from Proposition~\ref{Prop_equivalence} and Theorem~\ref{Thm_main} (with $N \mapsto N/c= \alpha N$). 
Note that the positivity of the rate function $\Phi(t;\alpha)>0$ follows from the inequality \eqref{energy inequality}. 
See Figure~\ref{Fig_LUE} for the graph of $\Phi$. 

\begin{figure}[t]
    \centering
    \includegraphics[width=0.45\textwidth]{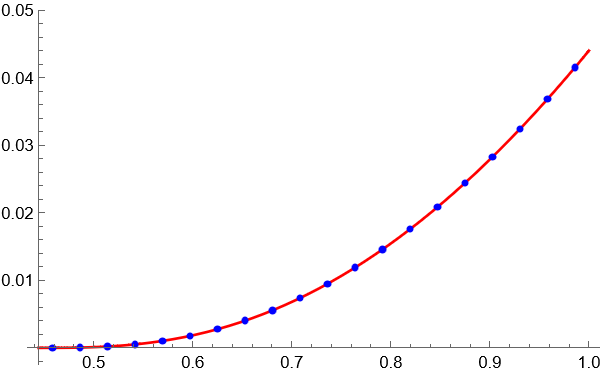} \quad  
        \includegraphics[width=0.45\textwidth]{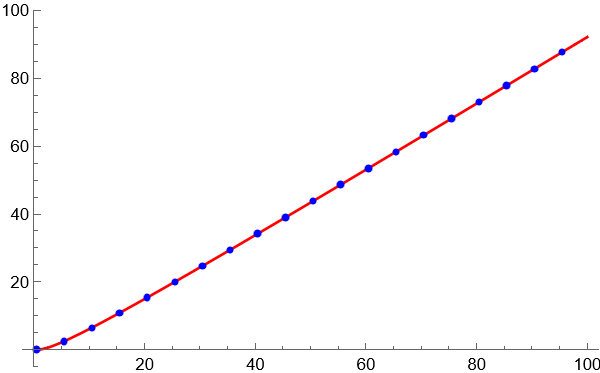}
    \caption{The full red line represents the graph of the rate function $t \mapsto \Phi(t;\alpha)$, where $\alpha=16/9$ and $t \geq \lambda_-=4/9$. The blue dots indicate the values of the function $t \mapsto S(t)-S(\lambda_-)$. Here, one can also observe the behaviours \eqref{rate function 0} and \eqref{rate function infty} in the left and right figures, respectively.  }
    \label{Fig_LUE}
\end{figure}

By the explicit formulas \eqref{energy post} and \eqref{energy pre}, one can observe that 
\begin{equation} \label{rate function 0}
\Phi(t;\alpha)  \sim \frac{\sqrt{\alpha+1}}{ 12 \lambda_-^2 } \,(t-\lambda_-)^3, \qquad t\to \lambda_-, 
\end{equation}
which agrees with the tail probability of the Tracy-Widom distribution.
This is known as a third-order phase transition appearing in a more general context, and we refer the reader to \cite{MS14} for a review. 
Let us also mention that in the opposite limit, we have 
\begin{equation}  \label{rate function infty}
\Phi(t;\alpha) \sim  t -\alpha \log t, \qquad t \to \infty.
\end{equation}
Here, the linear growth $t$ comes from the post-critical energy \eqref{energy post}, while the logarithmic corrections $ \alpha \log t$ comes from the pre-critical energy \eqref{energy pre}, cf. \eqref{energy infintiy}.  

\begin{rem}[Comparison with the Katzav-Castillo formula from a Coulomb gas approach] \label{Rem_KC formula}
In \cite{KC10}, Katzav and Castillo used a Coulomb gas method and derived the leading order asymptotic behaviour
\begin{equation} \label{KC formula}
\log \mathbb{P}\Big[\lambda_1 > t \Big] = -  \Big( S(t)-S(\lambda_-)\Big) N^2 
  +o(N^2),  \qquad (t > \lambda_-) 
\end{equation}
where $S$ is given by 
\begin{align}
\begin{split}
S(t) & = \frac{U(t)+t}{2} -\frac{ (U(t)-t)^2 }{ 32 } +\frac{\alpha}{4} ( \sqrt{U(t)}-\sqrt{t} )^2 
\\
&\quad - \log \Big( \frac{ U(t)-t }{4} \Big)  +\frac{\alpha^2}{4} \log \Big( t \, U(t) \Big) -\alpha (\alpha+2) \log \Big( \frac{ \sqrt{U(t)}+\sqrt{t} }{2 } \Big) . 
\end{split}
\end{align} 
Here,  
\begin{equation}
U(t)= \frac{4}{3} \Big(t+2(\alpha+2) \Big)  \cos^2\Big( \frac{\theta+2\pi}{3} \Big),\qquad   \theta= \arctan \Big(   \sqrt{ \frac{ (t+2(\alpha+2) )^3-27 \alpha^2 t }{27 \alpha^2 t }    } \Big) . 
\end{equation}  
Indeed, this value $U(t)$ is the right edge of the \emph{constrained} spectral density (see \cite{VMB07})
\begin{equation} \label{constrained LUE density}
\frac{ \sqrt{U(t)-x} }{ 2\pi \sqrt{x-t} } \frac{ x-\alpha \sqrt{t/U(t)} }{  x } \mathbbm{1}_{[t,U(t)]}(x), 
\end{equation}
the density of the LUE \eqref{LUE} conditioned on $\lambda_1>t$. 
By comparing \eqref{LDP pushed} and \eqref{KC formula}, we have 
\begin{equation} \label{comparison with KC formula}
 \Phi( t;\alpha)= S(t)-S(\lambda_-), 
\end{equation}
see Figure~\ref{Fig_LUE}. This identity should follow from the explicit solution of the cubic equation \eqref{q equation} (in the variable $q^2$). 
We also refer to \cite[Section 6.2]{Ku19} for the interpretation of the formula \eqref{KC formula} from the viewpoint of the recursion scheme. 
By \eqref{comparison with KC formula}, one can observe that our result extends the result \eqref{KC formula}, incorporating not only polynomial but also constant corrections of the large deviation probability $\mathbb{P}[\lambda_1 >t]$. 
Let us emphasise that while the Coulomb gas approach yields explicit formula for the leading-order asymptotic behaviour, applicable not only to the LUE but also to general Laguerre-$\beta$ ensembles, this approach is difficult to implement for deriving precise asymptotic behaviour.

In the pulled-regime, the precise order for $O(N^{-\infty})$ in \eqref{LDP pulled} is expected to be exponentially small: for a certain $\widetilde{c}(t)$,  
\begin{equation}
\log \, \mathbb{ P }\Big[ \lambda_1 > t\Big] =  e^{- \widetilde{c}(t) N  } (1+o(1)).
\end{equation}
See \cite[Eq.(16)]{KC10} for the Coulomb gas prediction on the constant $\widetilde{c}(t)$. 
We also refer to \cite[Eq.(1.4)]{Fo12} for the constant $\widetilde{c}(t)$ derived from the tail behaviour of the LUE density.
However, capturing such an exponentially decaying behaviour seems challenging based on the Riemann-Hilbert analysis we implement in this work.
\end{rem}

\begin{rem}[Asymptotic expansions of partition functions in one dimension] \label{Rem_1D partition}
Compared to the two-dimensional point process, there has been extensive work on the partition functions of log gases in one-dimension. 

For instance, in the same spirit as the present work, Riemann-Hilbert analysis proves to be a robust method for the determinantal point process \cite{EM03}, requiring thorough and technical analysis but particularly strong in obtaining the precise form of coefficients and addressing singularities \cite{DIK11, DIK14}.
In this vein, a general result was achieved in \cite{Ch19,CG21}, where the authors derived the precise asymptotic form of the partition functions associated with Gaussian, Laguerre, and Jacobi type (one-cut regular) weights together with jump and root-type singularities. In particular, Theorem~\ref{Thm_LUE} can be re-derived by employing \cite[Theorem 1.2]{CG21}, alongside additional work involving potential-theoretic computations. Nonetheless, the resulting expressions differ from those in Theorem~\ref{Thm_LUE}, as our formulas stem from the equilibrium measure of the two-dimensional model, thus involving the parameter satisfying a non-trivial cubic equation \eqref{q equation}. This demonstrates the remarkably powerful nature of the duality formula, allowing us to obtain two-dimensional results from one-dimensional results. It is noteworthy that our method, outlined in the next section, takes the opposite approach; namely, we derive the one-dimensional result (Theorem~\ref{Thm_LUE}) as a consequence of the two-dimensional result (Theorem~\ref{Thm_main}). Furthermore, we emphasise that Theorem~\ref{Thm_OP fine asymp}, concerning the strong asymptotic behaviour of a planar orthogonal polynomial, which holds its own significance that extends \cite{BBLM15}, cannot be derived as a consequence of an existing one-dimensional result.

Among different approaches, let us also mention that the topological recursion of Eynard and Orantin \cite{EO07} provides an efficient method to derive the structural form of the asymptotic expansion  of one-dimensional point processes \cite{BG13,BG24}. We also refer to \cite{Ma23,IMS18} and references therein for the implementation of such approaches for various cases, including Toeplitz determinants.
\end{rem}


\subsection*{Organisation of the paper}

The rest of this paper is organised as follows. In Section~\ref{Section_Overall}, we introduce the overall strategy and complete the proof of Theorem~\ref{Thm_main}. However, two main ingredients, Proposition~\ref{Prop_energy eval} for the energy evaluations, and Theorem~\ref{Thm_OP fine asymp} for the fine asymptotic behaviour of the orthogonal polynomials, will be established in the later sections.
Section~\ref{Section_energy} is devoted to proving Proposition~\ref{Prop_energy eval} based on the logarithmic potential theory and conformal mapping method. In Section~\ref{Section_RH analysis}, we prove Theorem~\ref{Thm_OP fine asymp} using Riemann-Hilbert analysis and a partial Schlesinger transform.  
Section~\ref{Subsec_critical} is a separate part, where we provide the derivation of the duality formula of the form \eqref{equivalence btw three} and also the proof of Proposition~\ref{Prop_critical expansion} on the free energy expansion in the critical regime.


\section{Proof of Theorem~\ref{Thm_main} } \label{Section_Overall}

This section reaches its culmination with the proof of Theorem~\ref{Thm_main}. For reader's convenience, we begin by providing a summary of the overall strategy.

\begin{itemize}
\item As a first step, we express the partition function as an integral of its derivative with respect to a deformation parameter. Here, the integration constants are given by the reference partition functions, whose asymptotic behaviour can be computed precisely (Lemma~\ref{Lem_Reference Z}). With proper choices of reference partition functions, this implies that it suffices to derive the asymptotic expansion of the derivative of free energies (Proposition~\ref{Prop_ZN derivative}), for which we also need the derivatives of the coefficients (Lemma~\ref{Lem_deri of N2 1 terms}).
\smallskip
\item Using the $\tau$-function, the derivative of the partition functions can be expressed in terms of the solution to the Riemann-Hilbert problem for the associated planar orthogonal polynomial (Lemma~\ref{Lem_tau function rep}). This allows us to express the derivative of the partition function in terms of the coefficient of the solution to the Riemann-Hilbert problem as the spectral variable $z \to \infty$ (Proposition~\ref{Prop_ZN YN coeff}).
\smallskip
\item As a consequence of the previous step, the free energy expansion can be derived using the asymptotic expansion of the orthogonal polynomial near infinity (Proposition~\ref{Prop_Asymptotic coefficient}). Then, by refining the Riemann-Hilbert analysis in \cite{BBLM15}, we obtain the fine asymptotic behaviour of the orthogonal polynomial outside the motherbody (Theorem~\ref{Thm_OP fine asymp}). Consequently, by computing the residue near infinity, we complete the proof of Theorem~\ref{Thm_main}. 
\end{itemize}

\medskip 

Let us now be more precise in introducing our strategy. 
Let $p_j$ be the monic polynomial satisfying 
\begin{equation} \label{Planar OP}
\int_\C p_j(z) \overline{p_k(z)} e^{-NQ(z)}\,dA(z) = h_j\,\delta_{jk} , 
\end{equation}
where $Q$ is given by \eqref{Q insertion} and $\delta_{jk}$ is the Kronecker delta. 
Then we have 
\begin{equation} \label{ZN orthogonal norm}
Z_N(a,c) =  N! \prod_{j=0}^{N-1} h_{j},  
\end{equation}
see e.g. \cite[Chapter 5]{BF22}. 

\medskip 

We now begin by explaining the above strategy in more detail. Each step in the above item is given in Subsections~\ref{Subsec_Deform}, \ref{Subsec_tau}, and \ref{Subsec_fine asymp}, respectively.

\subsection{Deformations of partition functions} \label{Subsec_Deform}

An important idea we employ in the asymptotic analysis of the partition functions is to consider the partition function $a \mapsto Z_N(a,c)$ as a function of the deformation parameter $a \ge 0$. 
Next, we require reference partition functions for which we can compute explicit formulas, along with their precise asymptotic expansions. 
The chosen reference partition functions are in the extremal, rotationally symmetric case. 
Additionally, it is crucial to select a reference partition function whose droplet is topologically equivalent to the droplet of the regime we intend to compute. This choice is essential to simplify the proof by avoiding the critical regime during this deformation. Consequently, we set $a=0$ and $a=\infty$ for the post- and pre-critical cases, respectively. The concrete statement of the above discussion is provided in the following lemma.
We mention that the idea of implementing deformations of partition functions (or structured determinants) has been utilized in Hermitian random matrix theory \cite{Kra07,DIK08,CK15} as well as in the Coulomb gas theory on a Jordan domain \cite{TT06,JV23}. A similar idea can also be found in the context of the transport method developed for general Coulomb gases, see \cite[Chapter 4]{Se19} and \cite{Se24}.

\begin{figure}[t]
    \centering
\begin{tikzpicture}
    \node at (0,0) {  \includegraphics[width=\textwidth]{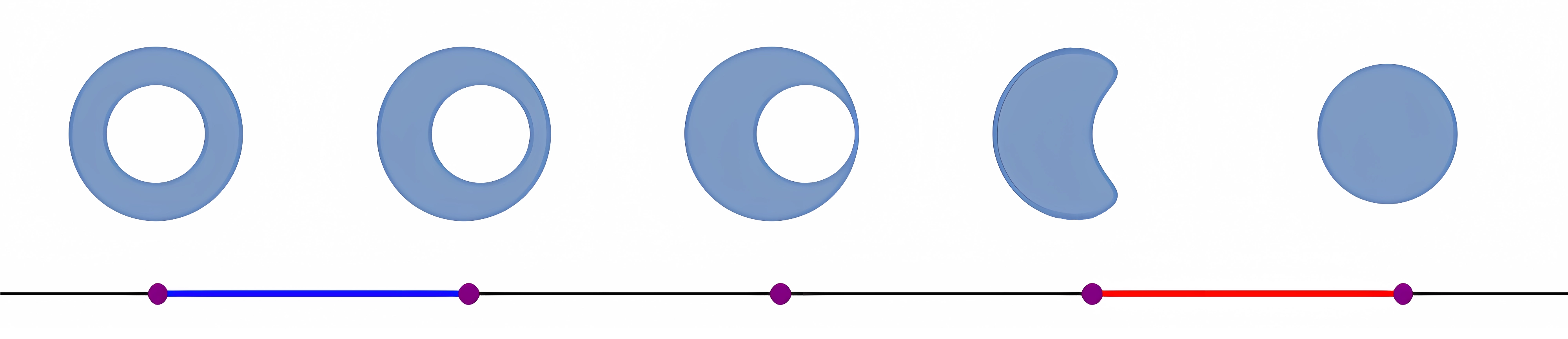}};

    \node at (-6.7, -2.25) {$a=0$};
    \node at (-3.35, -2.25) {$a<a_{ \rm cri }$};
    \node at (0, -2.25)  {$a=a_{ \rm cri }$};
    \node at (3.35, -2.25)  {$a>a_{ \rm cri }$};
    \node at (6.7, -2.25)  {$a=\infty$};
\end{tikzpicture}
    \caption{The plot shows deformations of the droplet. The leftmost $(a=0)$ and rightmost $(a=\infty)$ are the rotationally symmetric cases, for which we use the associated partition functions as reference. The thick lines indicate the integral domains in \eqref{ZN ref post} and \eqref{ZN ref pre}, respectively.}
    \label{Fig_transition}
\end{figure}
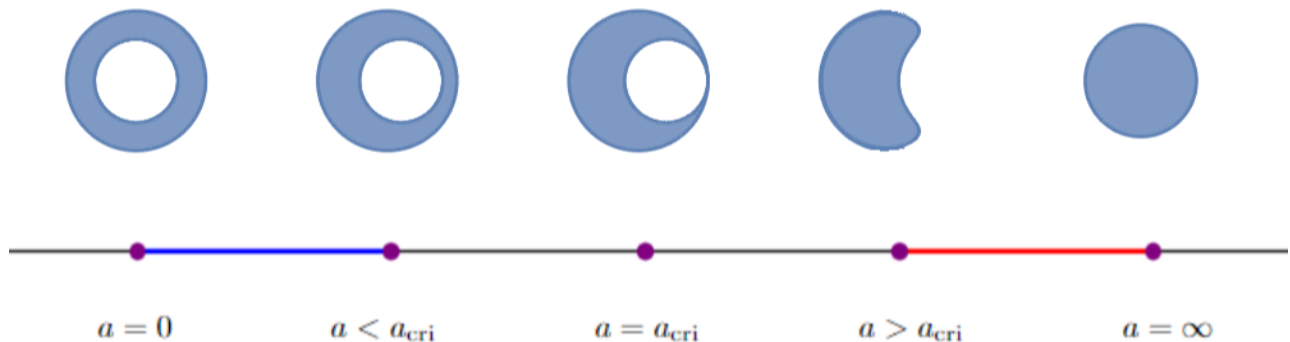

\begin{lem}[\textbf{Deformations and reference partition functions}]  \label{Lem_Reference Z}
We have 
\begin{align}
 \label{ZN ref post}
\log Z_N(a,c)& =\log Z_N(0,c)+ \int_0^a \frac{d}{dt} \log Z_N(t,c) \,dt 
\\
&=   \log Z_N^{ \rm Gin } + (2c \log a ) N^2  - \int_a^\infty \Big( \frac{d}{dt}  \log Z_N(t,c) - \frac{2cN^2}{t} \Big)\,dt.    \label{ZN ref pre}
\end{align}
The reference partition functions are evaluated as 
\begin{equation} \label{ZN reference evaluations}
Z_N(0,c)= N! \frac{G(N+cN+1)}{ G(cN+1) } N^{ -(c+\frac12)N^2-\frac{1}{2}N }, \qquad  Z_N^{ \rm Gin }=  N! \,G(N+1)  N^{ -\frac12 N^2-\frac{1}{2}N }, 
\end{equation}
and they satisfy the asymptotic behaviour as $N \to \infty$:
\begin{align}
\begin{split} \label{Z ind Gin asy}
\log Z_N(0,c) & = -\Big(  \frac{3}{4}+\frac{3c}{2}+\frac{c^2}{2}\log c -\frac{(c+1)^2}{2}\log(c+1) \Big) \, N^2 +\frac12 N \log N +\Big( \frac{\log(2\pi)}{2}-1 \Big)N
\\
&\quad   +\frac12 \log N+ \frac{\log(2\pi)}{2}+\frac{1}{12} \log \Big( \frac{c}{1+c}\Big) 
\\
&\quad +\sum_{k=1}^\infty \bigg( \frac{B_{2k}}{ 2k(2k-1) } \frac{1}{N^{2k-1}} + \frac{ B_{2k+2} }{ 4k(k+1) } \Big( \frac{1}{(c+1)^{2k}}-\frac{1}{c^{2k}}\Big) \frac{1}{N^{2k}} \bigg) ,
\end{split}
\\
\begin{split} \label{Z Gin asy}
 \log  Z_N^{ \rm Gin }& = -\frac34\, N^2+\frac12 N \log N +\Big( \frac{\log(2\pi)}{2}-1 \Big)N
\\
& \quad +\frac{5}{12} \log N+ \frac{\log(2\pi)}{2}+\zeta'(-1) +\sum_{k=1}^\infty \bigg( \frac{B_{2k}}{ 2k(2k-1) } \frac{1}{N^{2k-1}} + \frac{ B_{2k+2} }{ 4k(k+1) } \frac{1}{N^{2k}} \bigg) ,
\end{split}
\end{align}
where $B_k$ is the Bernoulli number. 
\end{lem}

As previously mentioned, we shall use \eqref{ZN ref post} for the post-critical regime, and \eqref{ZN ref pre} for the pre-critical regime, see Figure~\ref{Fig_transition} for an illustration.

\begin{rem}[Total integration formula]
As an immediate consequence of \eqref{ZN ref post}, \eqref{ZN ref pre} and \eqref{ZN reference evaluations}, it follows that 
\begin{align}
\begin{split}
\int_{0}^\infty \Big( \partial_t \log Z_N(t,c) - \frac{2cN^2}{t} \mathbbm{1}_{t>a} \Big) \,dt & = (2c\log a)N^2+\log \Big( \frac{ Z_N^{ \rm Gin } }{ Z_N(0,c) } \Big)
\\
&= cN^2 \log(a^2N)+ \log \Big( \frac{ G(N+1) G(cN+1) }{ G(N+cN+1) } \Big). 
\end{split}
\end{align}
\end{rem}

\begin{proof}[Proof of Lemma~\ref{Lem_Reference Z}]
The formula \eqref{ZN ref post} is obvious. 
To see \eqref{ZN ref pre}, let us write
\begin{align}
\widehat{Z}_N(a,c):= \int_{\C^N} \prod_{j>k=1}^N |z_j-z_k|^{2} \prod_{j=1}^{N}  e^{-N \widehat{Q}(z_j) }  \,dA(z_j), \qquad 
\widehat{Q}(z) = |z|^2-2c \log \Big| \frac{z-a}{a} \Big|. 
\end{align}
Note that 
\begin{equation} \label{Q hat limit}
\lim_{a \to \infty} \widehat{Q}(z) = |z|^2,  \qquad \widehat{Z}_N(\infty,c) = Z_N^{ \rm Gin }.  
\end{equation}
Since $e^{-NQ(z)} = e^{-N \widehat{Q}(z) } \,a^{2cN}, $
it follows that 
\begin{equation} \label{ZN hat ZN rel}
\log Z_N(a,c) = (2c\log a)N^2 + \log \widehat{Z}_N(a,c) .
\end{equation}
Then we have 
\begin{equation}
\begin{split}
\log Z_N(a,c) & = (2c\log a)N^2+ \log \widehat{Z}_N(a,c)  
\\
&=  (2c\log a)N^2+ \log \widehat{Z}_N(\infty,c)  - \int_a^\infty \partial_t \log \widehat{Z}_N(t,c) \,dt
\\
&=  (2c\log a)N^2+ \log \widehat{Z}_N(\infty,c)  - \int_a^\infty \Big( \partial_t  \log Z_N(t,c) - \frac{2cN^2}{t} \Big)\,dt, 
\end{split}
\end{equation}
which leads to \eqref{ZN ref pre}. 

If $a=0$, the potential is radially symmetric, and consequently the associated orthogonal polynomial is monomial, i.e. $p_n(z)=z^n.$ 
The orthogonal norm is then computed as 
\begin{equation}
h_j= \int_\C |z|^{2j} e^{-NQ(z)}\,dA(z) = 2\int_0^\infty r^{2j+2cN+1} e^{-Nr^2}\,dr =  \frac{ \Gamma(j+cN+1)  }{ N^{j+cN+1}  }. 
\end{equation}
By using \eqref{ZN orthogonal norm}, we have
\begin{equation}
Z_N(0,c)=N! \prod_{k=0}^{N-1}  \frac{ \Gamma(k+cN+1)  }{ N^{k+cN+1}  }= N! \frac{G(N+cN+1)}{ G(cN+1) } N^{ -(c+\frac12)N^2-\frac{1}{2}N }. 
\end{equation}
This also gives \eqref{ZN reference evaluations} since $Z_N^{ \rm Gin }=Z_N(0,0)$. 
Now \eqref{Z ind Gin asy} and \eqref{Z Gin asy} follow from the asymptotic behaviour of the gamma function 
\begin{equation}
\log N! = \log N + \log \Gamma(N) = \Big(N+\frac12\Big) \log N - N +\frac12 \log(2\pi) +\sum_{k=1}^\infty \frac{B_{2k}}{ 2k(2k-1)N^{2k-1} }
\end{equation}
as $N \to \infty$, and of the Barnes $G$-function
\begin{align}
\begin{split} \label{Barnes G asymp}
\log G(z+1) & =\frac{z^2 \log z}{2} -\frac34 z^2+\frac{ \log(2\pi) z}{2}-\frac{\log z}{12}+\zeta'(-1)  + \sum_{k=1}^\infty \frac{ B_{2k+2} }{ 4k(k+1)  } \frac{1}{z^{2k}}, 
\end{split}
\end{align}
as $z \to \infty$, see e.g. \cite[Eqs.(5.11.1), (5.17.5)]{NIST}. \end{proof}



By Lemma~\ref{Lem_Reference Z}, we need to derive the asymptotic behaviour of $\frac{d}{dt} \log Z_N(t,c)$. 
For this purpose, we need to compute the derivatives of the coefficients in the expansion \eqref{Z expansion main}. 
In the following, we denote by $\partial_a$ and $\bar{\partial}_a$ the complex derivatives with respect to $a$ and $\bar{a}$, respectively. 
To distinguish the notations further, we also use $ \mathfrak{d}_a  = \partial_a +\bar{\partial}_a$ to represent the operator of differentiation with respect to the real variable $a$.
For instance, $\mathfrak{d}_a a^2=2a$, whereas $\bar{\partial}_a |a|^2=a$.

\begin{lem}[\textbf{Derivatives of the coefficients in the free energy expansion}]\label{Lem_deri of N2 1 terms}
For the pre-critical case, we have 
\begin{align}
\begin{split} \label{I pre deri}
  \mathfrak{d}_a \,  \mathcal{I}^{ \rm pre }(a,c) 
  =   \frac{1}{2a}+a - \frac{1}{ a q^2} - \frac{a^3q^4}{2}  = - \frac{ (1-a^2q^2)(2-q^2-a^2q^4) }{ 2a q^2 } 
\end{split}
\end{align}
and
\begin{align}
\begin{split}
  \mathfrak{d}_a \mathcal{F}^{ \rm pre }(a,c) & =    - \frac{q^2 (1 - a^4 q^4)^2}{8 a (1 - q^2) (1 - a^4 q^6)^2}. 
\end{split}
\end{align}
\end{lem}

\begin{proof}
Recall that $q \equiv q(a)$ is a function of $a.$ 
By \eqref{q equation}, we have 
 \begin{equation} \label{c in terms of qa}
    c = \frac{a^2q^2}{2} +\frac{1}{4a^2q^4}-\frac{a^2+2}{4}. 
 \end{equation} 
By differentiating \eqref{q equation} with respect to $a$ 
and using \eqref{c in terms of qa}, we have 
\begin{equation} \label{q' q}
\frac{q'}{q} 
= \frac{1}{2a} \frac{ 1+a^4q^4-2a^4q^6 }{  a^4q^6-1 }.
\end{equation}
Using \eqref{q' q}, one can express derivatives of \eqref{energy pre} and \eqref{det of Lap pre} in terms of $q$ and $a$.  
Then the lemma follows from straightforward computations. 
\end{proof}

\begin{rem} \label{Rem_positivity}
Note that by \eqref{I pre deri} and 
\eqref{c in terms of qa} , we have 
\begin{align} \label{deriv of energy diff}
  \mathfrak{d}_a \,\Big(   \mathcal{I}^{ \rm pre }(a,c)  -   \mathcal{I}^{ \rm post }(a,c) \Big)  & =  \frac{1}{2a}+a - \frac{1}{ a q^2} - \frac{a^3q^4}{2} + 2a c = \frac{(1-q^2)^2 (1-a^4q^4)}{2aq^4} >0. 
\end{align}
This implies the inequality \eqref{energy inequality}. 
\end{rem}

By combining Proposition~\ref{Prop_energy eval}, Lemma~\ref{Lem_Reference Z} and Lemma~\ref{Lem_deri of N2 1 terms}, Theorem~\ref{Thm_main} can be reduced to the following proposition.

\begin{prop}[\textbf{Asymptotic expansion of the derivative of free energies}] \label{Prop_ZN derivative}
As $N\to \infty$, we have the following. 
\begin{itemize}
    \item[\textup{(i)}] For the post-critical case, we have 
\begin{equation}
 \mathfrak{d}_a \log Z_N(a,c) = 2c a N^2  + O(N^{-\infty}). 
\end{equation}
 \item[\textup{(ii)}] For the pre-critical case, we have 
\begin{equation}
 \mathfrak{d}_a \log Z_N(a,c) = \frac{ (1-a^2q^2)(2-q^2-a^2q^4) }{ 2a q^2 }\,  N^2 - \frac{q^2 (1 - a^4 q^4)^2}{8 a (1 - q^2) (1 - a^4 q^6)^2} +O(N^{-1}).  
\end{equation}
\end{itemize}
\end{prop}

In the following subsection, we further reduce Proposition~\ref{Prop_ZN derivative} to certain asymptotic behaviour of the planar orthogonal polynomial \eqref{Planar OP}.


\subsection{Riemann-Hilbert problem and $\tau$-function} \label{Subsec_tau}

We shall implement the $\tau$-function \cite{BO09}, which arises in the context of the Riemann-Hilbert problem, cf. \cite{BBLM15,LY17,LY23}.  
The first important idea in \cite{BBLM15} for the asymptotic analysis of planar orthogonal polynomial $p_n$ is to demonstrate an equivalence between the planar orthogonality \eqref{Planar OP} and a certain orthogonality on a contour. Extensions of such an equivalence can be found in \cite{BKP23, LY19}. 
As a consequence, one can construct a standard Riemann-Hilbert problem for $p_n$, and we recall it here.

Let us mention that on some occasions, we write both $n$ and $N$. This distinction has usually been made in Riemann-Hilbert analysis to distinguish between the degree of the orthogonal polynomial and the number of particles.  

Since $Z_N(a, c) = Z_N(|a|, c)$, it suffices to consider real-valued $a$. However, when describing the Riemann-Hilbert problem, it is advantageous to consider the general complex $a \in \mathbb{C}$ and to distinguish between $a$ and its complex conjugate $\bar{a}$. 

Let $\Gamma$ be a simple closed curve enclosing $0$ and $a$.
Define the weight function 
\begin{equation}\label{def omega}
\omega_{N}(z)=\frac{(z-a)^{Nc}e^{-N  \bar{a}  z}}{z^{Nc+N }}. 
\end{equation}
In the proof of Lemma~\ref{Lem_tau function rep} below, we exploit the fact that the parameter \( a \) in \eqref{def omega} is complex-valued to trade derivatives with respect to the location of the singularity for derivatives with respect to the conjugate parameter \( \bar{a} \) appearing in the exponent.

Recall that the monic orthogonal polynomial $p_n$ satisfies the planar orthogonality \eqref{Planar OP}. It was shown in \cite[Section 3]{BBLM15} (see also \cite{LY19,BKP23}) that it also satisfies the contour orthogonality
\begin{equation} \label{def of contour ortho}
\int_\Gamma p_n(z) z^j \omega_N(z) \,dz=0, \qquad (j=0,1,\dots,n-1). 
\end{equation} 
In \cite{BBLM15}, the authors consider the case \( a > 0 \). However, as noted in \cite[p.115]{BBLM15}, the general case \( a \in \mathbb{C} \setminus\{0\} \) can be reduced to the real case by a rotation of coordinates. More precisely, when we consider a general \( a \in \mathbb{C} \setminus \{0\} \), if we follow the arguments for contour orthogonality given in \cite[Section~3]{BBLM15}, 
and replace $a$ in \cite[Eq.~(3.8)]{BBLM15} with $\overline a$, then the rest of the proof proceeds 
in the same way, and one arrives at the weight function of the form \eqref{def omega}, 
with $\overline{a}$ appearing only in the exponent.

Consequently, we consider the following Riemann-Hilbert problem. 

\begin{defi}[Riemann-Hilbert problem $Y$] \label{Def_RHP Y}
We consider the following Riemann-Hilbert problem for $Y\equiv Y_n$:
\begin{equation} \label{RHP Y}
\begin{cases}
Y(z)\mbox{\quad is holomorphic},&\quad z\in \C\setminus \Gamma,
\smallskip 
\\
 Y_+(z)=  Y_-(z)\begin{pmatrix}
1&\omega_{N}(z)\\0&1
\end{pmatrix},&\quad z\in \Gamma,
\smallskip 
\\
Y(z)=\big(I+{ O}(z^{-1})\big)
\begin{pmatrix}
z^n & 0
\\
0 & z^{-n}
\end{pmatrix} ,&\quad z\to\infty.\end{cases}
 \end{equation}
Here $Y_\pm(x)$ is the boundary values of $Y$ on $\Gamma$. 
\end{defi}

Let $q_{n}(z):=q_{n,N}(z)$ be the unique polynomial of degree $n-1$ such that
$$\frac{1}{2\pi i}\int_\Gamma\frac{q_{n}(w)\omega_{N}(w)}{w-z}dw=\frac{1}{z^n}\Big(1+{ O}\Big(\frac{1}{z}\Big)\Big).$$
Then the matrix function 
\begin{equation}\label{def Yz}
Y(z) \equiv Y_n(z) :=\begin{pmatrix}\displaystyle
p_{n}(z)&\displaystyle\frac{1}{2\pi i}\int_\Gamma\frac{p_{n}(w)\omega_{N}(w)}{w-z} \, d w
\smallskip 
\\ 
q_{n}(z)&\displaystyle\frac{1}{2\pi i}\int_\Gamma\frac{q_{n}(w)\omega_{N}(w)}{w-z} \, d w
\end{pmatrix}     
\end{equation}
is the unique solution to the Riemann-Hilbert problem \eqref{RHP Y}. 
For the latter purpose, let $\mathfrak{A}_{jk}  \equiv \mathfrak{A}_{jk}^{(n)}  $ be the coefficients of the $O(1/z)$-term in the large-$z$ expansion of $Y_n$:
\begin{equation}\label{Y asymtotic}
Y_n(z)=\bigg(I+\frac{1}{z}\begin{pmatrix}
\mathfrak{A}_{11}&\mathfrak{A}_{12}
\\
\mathfrak{A}_{21}& -\mathfrak{A}_{11}
\end{pmatrix}+O\Big(\frac{1}{z^{2}}\Big)\bigg) \begin{pmatrix}
z^n & 0
\\
0 & z^{-n} 
\end{pmatrix},\qquad z\to\infty.
\end{equation}  
 
We now define 
\begin{equation} \label{def of T}
\widetilde Y_n(z):=Y_n(z) T(z), \qquad T(z) \equiv T_N(z):  =\begin{pmatrix}
\Big(\dfrac{z-a}{z}\Big)^{Nc}{e^{-N  \bar{a}   z}}&0
\smallskip 
\\0&z^{N}
\end{pmatrix}.
\end{equation}  
Then $\widetilde Y_n$ satisfies the following Riemann-Hilbert problem:
\begin{equation}
\begin{cases}
\widetilde Y_n(z)\mbox{\quad is holomorphic},& \quad z\in \C\setminus ( \Gamma \cup [ 0 , a ] ) ,
\smallskip 
\\
\widetilde Y_{n,+}(z)=  \widetilde Y_{n,-}(z)\begin{pmatrix}
1&1\\0&1
\end{pmatrix},&\quad z\in \Gamma,\\
\widetilde Y_n(z)=\big(I+{ O}(z^{-1})\big) 
\begin{pmatrix}
z^n &0
\smallskip 
\\
0&z^{-n}
\end{pmatrix}T(z),&\quad z\to\infty.
\end{cases}
\end{equation} 
Note that by definition, $\widetilde{Y}_n T^{-1}=Y_n$ is invertible and holomorphic near the neighbourhood of $0$ and $a$.

The matrices $Y_N$ and $T_N$ can be used to express the derivative of the partition function.  

\begin{defi}[$\tau$-function] 
The $\tau$-function is defined by 
\begin{equation}\label{def tau function}
  \bar{\partial}_a   \log \tau:=  \underset{z=\infty}{\textup{Res}} \,\bigg[ \Tr \Big(Y(z)^{-1}\partial_z Y(z)  \bar{\partial}_a    T(z)T(z)^{-1}\Big) \bigg] .   
\end{equation}
Here $\tau \equiv \tau_n$ depends on the degree $n$ of the orthogonal polynomial $p_n$. 
\end{defi}
 
A similar definition of the $\tau$-function is given in \cite[Definition 3.2]{BO09}. As in \cite{BO09}, for our purposes, it suffices to define the derivative of the \( \tau \)-function.

By applying standard techniques to the deformation of the partition function (see e.g. \cite[Theorem 3.4]{BO09}), we shall prove the following lemma. The key idea is to compute the $\tau$-ratio $d \log (\tau_{n+1}/\tau_n)$ when the degree increases by 1. 

In the following statement, we allow \( a \) to be complex-valued in the definition of \( Z_N(a,c) \), as previously noted.

\begin{lem}[\textbf{Partition function in terms of the Riemann-Hilbert problem}] \label{Lem_tau function rep}
We have 
\begin{equation} \label{Z relation with Y T}
 \bar{\partial}_a   \log Z_N(a,c) =    -\bar{\partial}_a  \log \tau_N.  
\end{equation}
\end{lem}
\begin{proof}
  Since $\widetilde{Y}_n$ has a constant jump, $\bar{\partial}_a \widetilde{Y}_n \widetilde{Y}_n^{-1}$ is meromorphic with possible isolated singularities only at $0$ and $a$. Indeed, since $Y_n$ is invertible and holomorphic near the neighbourhood of $0$ and $a$, one can observe that  
\begin{align}
\begin{split} \label{def of Bn}
B_n(z) := \bar{\partial}_a \widetilde Y_n(z)  \widetilde Y_n(z)^{-1} & =  \bar{\partial}_a  Y_n(z) Y_n(z)^{-1}+Y_n(z) \bar{\partial}_a T(z) T(z)^{-1} Y_n(z)^{-1}    
\\
&= \bar{\partial}_a  Y_n(z) Y_n(z)^{-1} +Y_n(z) \begin{pmatrix}
-N z & 0
\\
0 & 0 
\end{pmatrix} Y_n(z)^{-1}   
\end{split}
\end{align}
is a polynomial of degree at most $1$.  
On the other hand, since $\widetilde{Y}_{n+1}$ and $\widetilde{Y}_n$ have the same constant jump, it follows that  $\widetilde{Y}_{n+1} = C_n \widetilde{Y}_n$, where $C_n(z)$ is a polynomial of degree 1. Furthermore, it follows from \eqref{Y asymtotic} that $C_n$ is of the form  
\begin{equation}\label{def of Cn form}
C_n(z)=\begin{pmatrix}
z+\mathfrak{A}_{11}^{(n+1)}-\mathfrak{A}_{11}^{(n)}&-\mathfrak{A}_{12}^{(n)}
\\
\mathfrak{A}_{21}^{(n+1)}& 0
\end{pmatrix}.
\end{equation}   

By the definition of the $\tau$-function in \eqref{def tau function} and the definition of the raising operator $C_n$, we have
\begin{align*}
\bar{\partial}_a \log \frac{\tau_{n+1}}{\tau_n}&=\underset{z=\infty}{\textup{Res}} \,\bigg[ \Tr \Big(Y_{n}(z)^{-1}C_n(z)^{-1}\partial_z C_n(z)Y_{n}(z) \bar{\partial}_a T(z)T(z)^{-1}\Big) \bigg]
\\
&= \underset{z=\infty}{\textup{Res}} \,\bigg[ \Tr \Big( C_n(z)^{-1}\partial_z C_n(z)Y_{n}(z) \bar{\partial}_a T(z)T(z)^{-1} Y_{n}(z)^{-1} \Big) \bigg]
\\
 &=\underset{z=\infty}{\textup{Res}} \,\bigg[ \Tr \Big(C_n(z)^{-1}\partial_z C_n(z)B_n(z)\Big)\bigg]-\underset{z=\infty}{\textup{Res}} \,\bigg[ \Tr \Big(C_n(z)^{-1}\partial_z C_n(z) \bar{\partial}_a Y_n(z)  Y_n(z)^{-1}\Big)\bigg]
\\
 &= -\underset{z=\infty}{\textup{Res}} \,\bigg[ \Tr \Big(C_n(z)^{-1}\partial_z C_n(z) \bar{\partial}_a Y_n(z)  Y_n(z)^{-1}\Big)\bigg],
\end{align*} 
where the second equality follows from the cyclic property of trace and the third equality follows from \eqref{def of Bn}. Here, the last equality follows from the fact that $B_n$ is a polynomial. 
Furthermore, it follows from \eqref{def of Cn form} that 
$$
 C_n(z)^{-1} \partial_z C_n(z) = \begin{pmatrix}
0 &0 
\\
-1/\mathfrak{A}_{12}^{(n)} & 0 
\end{pmatrix}.
$$
Then by \eqref{Y asymtotic}, we obtain 
$$
\bar{\partial}_a \log \frac{\tau_{n+1}}{\tau_n} =  - \bar{\partial}_a\log \mathfrak{A}_{12}^{(n)}. 
$$

We write 
\begin{equation}
\nu_k := \int_\Gamma z^k \omega_{N}(z)\,dz. 
\end{equation}
Since the monic orthogonal polynomial $p_n$ in \eqref{def of contour ortho} can be expressed in terms of the moment
$$p_n(z)=\frac{1}{\det [\nu_{j+k}]_{j,k=0}^{n-1}} \det\begin{bmatrix}
\nu_{0}&\nu_{1}&\cdots&\nu_{n} \\
\vdots&\vdots&\vdots&\vdots\\
\nu_{n-1}&\nu_{n}&\cdots&\nu_{2n-1}\\
1&z&\cdots&z^n\end{bmatrix},
$$
it follows from \eqref{def Yz} and \eqref{Y asymtotic} that 
$$
\bar{\partial}_a \log \frac{\tau_{n+1}}{\tau_n} = \bar{\partial}_a \log \Big(\det [\nu_{j+k}]_{j,k=0}^{n-1} \Big)  -\bar{\partial}_a \log  \Big(\det [\nu_{j+k}]_{j,k=0}^{n} \Big) .
$$
Using this identity inductively with the initial conditions $\bar{\partial}_a \log \tau_0=0$ and $\bar{\partial}_a \det \nu_{0}=0$, for $n=N$, we obtain 
$$
\bar{\partial}_a \log \tau_N= -\bar{\partial}_a \log \Big(\det [\nu_{j+k}]_{j,k=0}^{N-1}\Big).
$$

Let us define 
\begin{equation}
\mu_{jk} := \int_\C z^j \bar{z}^k e^{-NQ(z)}\,dA(z).  
\end{equation}
Then by the general theory on determinantal point process, we have 
\begin{equation}
Z_N(a,c) = N!\,\det [\mu_{jk}]_{j,k=0}^{N-1}. 
\end{equation}
Note that by \cite[(E.1)]{BBLM15}, we have 
\begin{equation} \label{3.38}
\frac{ \det [\mu_{jk}]_{j,k=0}^{N-1} }{ \det [\nu_{j+k}]_{j,k=0}^{N-1}   } = \frac{(-1)^{N(N-1)/2}}{\pi^N}\prod_{k=0}^{N-1} \frac{ \Gamma(cN+k+1) }{ 2i\,N^{ cN+k+1 }  }.
\end{equation}
Let us again emphasise that in \cite{BBLM15}, this identity was verified for real values of \( a \). However, the determinant \( \det [\nu_{j+k}]_{j,k=0}^{N-1} \) depends only on \( |a| \), since each entry \( \nu_{j+k} \) differs from its counterpart for real \( a \) only by a phase factor. These phase factors cancel out when computing the determinant. Alternatively, by letting $a=|a|e^{i\theta}$ and the change of variables $z_j= e^{i\theta}w_j$, one can observe that
\begin{align*}
\det [\nu_{j+k}]_{j,k=0}^{N-1} 
&=  \frac{1}{N!}\int_{ \Gamma^N } \prod_{ j<k } (z_j-z_k)^2  \prod_{j=1}^N \frac{(z_j-|a|e^{i\theta})^{Nc} e^{-N |a| e^{-i\theta} z_j } }{ z_j^{ N c+N } } \,dz_j 
\\
&=  \frac{1}{N!}\int_{ \Gamma^N } \prod_{ j<k } (w_j-w_k)^2  \prod_{j=1}^N \frac{(w_j-|a| )^{N c} e^{-N |a| w_j } }{ w_j^{ N c+N } } \,dw_j.
\end{align*} 
From this, it is again clear that $\det [\nu_{j+k}]_{j,k=0}^{N-1}$ depends only on $|a|.$

Since the right-hand side of \eqref{3.38} is independent of $a$ and $\bar a$, the lemma follows. 
\end{proof}

Recall that $ \mathfrak{d}_a  = \partial_a +\bar{\partial}_a$ is the differentiation with respect to the real variable $a$. 

\begin{prop}[\textbf{Partition functions and fine asymptotics of the orthogonal polynomial}] \label{Prop_ZN YN coeff}
We have 
\begin{equation}
 \mathfrak{d}_{a} \log Z_N(a,c) =   2N \, \mathfrak{A}_{11}, 
\end{equation}
where $\mathfrak{A}_{11}$ is defined by \eqref{Y asymtotic}.
\end{prop}
\begin{proof}  
By \eqref{def of T}, we have
$$
\bar{\partial}_a  T(z) T(z)^{-1} =  \begin{pmatrix}
-Nz & 0
\smallskip 
\\
0 & 0 
 \end{pmatrix} ,
$$
which gives 
\begin{align*}
 \Tr \Big(Y_N(z)^{-1}\partial_z Y_N(z)   \bar{\partial}_a   T(z) \,T(z)^{-1}\Big)  &=   -N  z \Big[ Y_N(z)^{-1} \partial_z Y_N(z) \Big]_{11} .
\end{align*}
Then it follows from the asymptotic behaviour in \eqref{Y asymtotic} that
\begin{equation*}
 \underset{z=\infty}{\textup{Res}} \,\bigg[ \Tr \Big(Y_N(z)^{-1}\partial_z Y_N(z)   \bar{\partial}_a  T_N(z) \,T_N(z)^{-1}\Big) \bigg] =\underset{z=\infty}{\textup{Res}} \,\Big[ -N^2+\frac{N \mathfrak{A}_{11}}{z}+O\big(\frac{1}{z^2}\big)\Big]=  - N \mathfrak{A}_{11}. 
\end{equation*} 
Combining this with \eqref{Z relation with Y T}, we have
$$
\bar{\partial}_a \log Z_N(a,c) =  N \, \mathfrak{A}_{11}.
$$
Since \( Z_N(a,c) \) is real-valued, it follows that $\mathfrak{d}_{a} \log Z_N(a,c) = 2 \, \bar{\partial}_a \log Z_N(a,c)$, which completes the proof. 
\end{proof}

Then by \eqref{def Yz}, Proposition~\ref{Prop_ZN derivative} can be further reduced to the following.

\begin{prop}[\textbf{Asymptotic behaviour of the coefficients}] \label{Prop_Asymptotic coefficient}
Let $\mathfrak{A}_{11}$ be the coefficient in \eqref{Y asymtotic}. Then as $z \to \infty$, we have 
\begin{equation} \label{OP asymptotic z inf}
p_N(z) = z^N + \mathfrak{A}_{11} z^{N-1} + O(z^{N-2}). 
\end{equation}
Furthermore, as $N \to \infty$, the coefficient $\mathfrak{A}_{11}$ satisfies the following asymptotic behaviour.  
\begin{itemize}
    \item For the post-critical case,
   \begin{equation} \label{A11 post}
    \mathfrak{A}_{11}= c a N +O(\frac{1}{N^{\infty}}).
   \end{equation}
    \item For the pre-critical case, 
    \begin{equation} \label{A11 pre}
    \mathfrak{A}_{11}=  \frac{ (1-a^2q^2)(2-q^2-a^2q^4) }{ 4a q^2 }\,  N - \frac{q^2 (1 - a^4 q^4)^2}{16 a (1 - q^2) (1 - a^4 q^6)^2} \frac{1}{N} +O(\frac{1}{N^2}). 
    \end{equation}
\end{itemize}
\end{prop}

In Theorem~\ref{Thm_OP fine asymp} below, we shall prove a stronger statement on the fine asymptotic behaviour of $p_N$.

\subsection{Fine asymptotic behaviour and proof of Theorem~\ref{Thm_main}} \label{Subsec_fine asymp}

We show the fine asymptotic behaviour of the orthogonal polynomial. 

\begin{thm}[\textbf{Fine asymptotic behaviour of the orthogonal polynomial}] \label{Thm_OP fine asymp}
Let $\mathcal B$ the curve defined in Definition~\ref{Def_motherbody} and $g$ be the function defined in Definition~\ref{def gfunction}.  
Then as $N\to \infty$, we have the following. 
\begin{itemize}
   \item \textup{\textbf{(Post-critical regime)}} We have 
   \begin{equation}\label{post asymptotics}
    p_{N}(z) = z^N \Big( \frac{z}{z-a} \Big)^{c N} \Big(1+O(N^{-\infty})\Big), \qquad z \in \ext(\mathcal B).
   \end{equation}
 
   \item  \textup{\textbf{(Pre-critical regime)}} We have 
\begin{equation}\label{pre asymptotics}
p_{N}(z)=\Big(\sqrt{ R F'(z)}(1+R_{11}(z))-\frac{\sqrt{\kappa F'(z)}}{F(z)-  q}R_{12}(z)\Big) e^{Ng(z)}  \Big(1+O(N^{-2})\Big),  \qquad z \in \ext(\mathcal B)  
\end{equation}
where $F(z)$ is given by \eqref{def of F}, and $R_{11}$ and $R_{12}$ are defined in Definition~\ref{Def_Rational functions}. 
\end{itemize}
Here the error terms are uniform over $z$ in compact sets of $\widehat\C \setminus \mathcal B$ where $\widehat{\C}:=\C\cup \{\infty\}$.
\end{thm}

We stress that Theorem~\ref{Thm_OP fine asymp} extends \cite[Theorems 1.3 and 1.4]{BBLM15} providing more precise error terms.  
Let us also mention that the asymptotic behaviour of \emph{orthonormal} polynomials outside the droplet has also been studied by Hedenmalm and Wennman in a more general context \cite{HW21,HW20,Hed24}. In these works, an algorithm computing the subleading correction terms was also introduced, but their explicit evaluations for a given potential require a separate analysis. In particular, Theorem~\ref{Thm_OP fine asymp} is not a direct consequence of the general theory developed by Hedenmalm and Wennman, as we need to compute the \emph{monic} orthogonal polynomial and also derive the asymptotic behaviour inside the droplet.

\medskip 

We are now ready to show Theorem~\ref{Thm_main}. 

\begin{proof}[Proof of Theorem~\ref{Thm_main}] 
As explained in the previous subsection, it suffices to show Proposition~\ref{Prop_Asymptotic coefficient}.

 
For the post-critical case, it follows from \eqref{post asymptotics} that 
$$
p_N(z)=\Big(z^N+Nac\,z^{N-1}+O(z^{N-2})\Big)\Big(1+O(N^{-\infty})\Big),\qquad  z\to \infty.
$$
This gives rise to \eqref{A11 post}. 

\medskip 

Next, let us show the pre-critical case. 
By Lemma~\ref{Lem_g asymptotic} below, we have 
\begin{equation}
e^{ N g(z) }=  z^N \bigg(1  - \Big(\frac{1}{4a}+\frac{a}{2} -\frac{1}{2aq^2} -\frac{a^3q^4}{4} \Big) \frac{N}{z} +O(\frac{1}{z^2}) \bigg),  \qquad   z\to\infty. 
\end{equation}
On the other hand, by using \cite[Eq.(1.28)]{BBLM15},  
\begin{equation}
F(z)=\frac{z}{R}+\frac{\kappa}{Rq}+\frac{\kappa}{z}+O\Big(\frac{1}{z^2}\Big),\qquad z\to\infty.
\end{equation}
Therefore we have 
\begin{equation}
R F'(z) = 1+O(\frac{1}{z^2}), \qquad z \to \infty.  
\end{equation}
Furthermore, by definition, we have 
\begin{equation}
R_{11}(z)=O(\frac{1}{z}),\qquad R_{12}(z)=O(\frac{1}{z}), \qquad z \to \infty. 
\end{equation}
Then the desired identity \eqref{A11 pre} follows from 
\begin{equation} \label{Residue of R11}
\underset{z=\infty}{\textup{Res}} \,\Big[ R_{11}(z)  \Big]= - \frac{q^2 (1 - a^4 q^4)^2}{16 a (1 - q^2) (1 - a^4 q^6)^2} \frac{1}{N}. 
\end{equation}

It remains to show \eqref{Residue of R11}. For this, note that as $z \to \infty,$ 
\begin{align*}
R_1(z)R_2(z)  &= \bigg(I+\frac{1}{z-\overline\beta}R_2(z)h_{11}R_2(z)^{-1}+\frac{1}{(z-\overline\beta)^2}R_2(z)h_{12}R_2(z)^{-1}  \bigg) \bigg( I+\frac{h_{21}}{z-\beta}+\frac{h_{22}}{(z-\beta)^2}\bigg)  . 
\end{align*}
Here $R_1, R_2$, $h_{11}$, $h_{12}$ and $h_{21}, h_{22}$ are given in Definition~\ref{Def_Rational functions}. 
Then 
\begin{align}
\begin{split}
\underset{z=\infty}{\textup{Res}} \,\Big[ R_{11}(z)  \Big] & =\frac{1}{N} \frac{1+i}{128 \sqrt{2} } \bigg(   \frac{3}{  (R\kappa)^{1/4} }\Big(\frac{1}{\gamma_{11}^{3/2}}-\frac{i}{\gamma_{21}^{3/2}}\Big)+10 (R\kappa)^{1/4} \Big(\frac{\gamma_{22}}{\gamma_{21}^{5/2}}-  \frac{i\gamma_{12}}{\gamma_{11}^{5/2}}\Big) \bigg)
\\
&= \frac{1+i}{ 128 \sqrt{2} } \frac{1}{( R \kappa )^{1/4} }  \bigg(  \frac{1}{\gamma_{11}^{3/2}} \Big( 3-10 (R\kappa)^{1/2} i \frac{ \gamma_{12} }{ \gamma_{11} }\Big) -i\, \overline{  \frac{1}{\gamma_{11}^{3/2}} \Big( 3-10 (R\kappa)^{1/2} i \frac{ \gamma_{12} }{ \gamma_{11} }\Big) }  \, \bigg),  
\end{split}
\end{align}
where $\gamma_{11}=\overline{\gamma}_{21}$ and $\gamma_{12}=\overline{\gamma}_{22}$ are given by \eqref{gamma 11 asym} and \eqref{gamma 12 asym}. 
We claim that
\begin{align}
\begin{split} \label{desired iden for main}
&\quad   \frac{1}{\gamma_{11}^{3/2}} \Big( 3-10 (R\kappa)^{1/2} i \frac{ \gamma_{12} }{ \gamma_{11} }\Big) -i\, \overline{  \frac{1}{\gamma_{11}^{3/2}} \Big( 3-10 (R\kappa)^{1/2} i \frac{ \gamma_{12} }{ \gamma_{11} }\Big) }  
\\
&=   - 4\sqrt{2}(1-i) ( R \kappa )^{1/4}   \frac{q^2 (1 - a^4 q^4)^2}{  a (1 - q^2) (1 - a^4 q^6)^2} , 
\end{split}
\end{align}
which leads to \eqref{Residue of R11}. 
By using \eqref{gamma 11 asym} and \eqref{gamma 12 asym}, we have 
\begin{align*}
&\quad \frac{1}{( R \kappa )^{1/4} } \frac{1}{\gamma_{11}^{3/2}} \bigg( 3 - 10 (R\kappa)^{1/2} i \frac{ \gamma_{12} }{ \gamma_{11} }\bigg) =  \frac{2}{( R \kappa )^{1/4} }   \frac{ (\overline{\beta}-a) \overline{\beta}  }{ a(\overline{\beta}-R/q) \sqrt{ \overline{\beta}-\beta }  }  \bigg( 3 -10 (R\kappa)^{1/2} i \frac{ \gamma_{12} }{ \gamma_{11} }\bigg) 
\\
&= \frac{2}{( R \kappa )^{1/4} }   \frac{ (\overline{\beta}-a) \overline{\beta}  }{ a(\overline{\beta}-R/q) \sqrt{ \overline{\beta}-\beta }  }  \bigg( 3 -8 (R\kappa)^{1/2} i \Big(  \frac{1}{\overline{\beta}-R/q}+\frac{1}{2(\overline{\beta}-\beta)}-\frac{1}{\overline{\beta}-a}-\frac{1}{ \overline{\beta} }  \Big) \bigg). 
\end{align*} 
Then 
\begin{align*}
&\quad \frac{1}{( R \kappa )^{1/4} }  \bigg(  \frac{1}{\gamma_{11}^{3/2}} \Big( 3-10 (R\kappa)^{1/2} i \frac{ \gamma_{12} }{ \gamma_{11} }\Big) -i\, \overline{  \frac{1}{\gamma_{11}^{3/2}} \Big( 3-10 (R\kappa)^{1/2} i \frac{ \gamma_{12} }{ \gamma_{11} }\Big) }  \, \bigg)
\\
&=  \frac{e^{ \frac{\pi i}{4} } }{  a (\kappa R)^{1/2} } \frac{ (\overline{\beta}-a) \overline{\beta}  }{ \overline{\beta}-R/q  }  \bigg( 3 -8 (R\kappa)^{1/2} i \Big(  \frac{1}{\overline{\beta}-R/q}+\frac{1}{2(\overline{\beta}-\beta)}-\frac{1}{\overline{\beta}-a} -\frac{1}{ \overline{\beta} }  \Big) \bigg)
\\
&\quad - \frac{e^{ \frac{\pi i}{4} } }{  a (\kappa R)^{1/2} }   \frac{ (\beta-a) \beta  }{  \beta-R/q   }  \bigg( 3 + 8 (R\kappa)^{1/2} i \Big(  \frac{1}{\beta-R/q}+\frac{1}{2(\beta-\overline{\beta})}-\frac{1}{\beta-a}-\frac{1}{ \beta }  \Big) \bigg). 
\end{align*}
Note that 
\begin{align*}
&\quad \frac{ (\overline{\beta}-a) \overline{\beta}  }{ \overline{\beta}-R/q  }   -   \frac{ (\beta-a) \beta  }{  \beta-R/q   }  = \frac{   (\overline{\beta}-a) \overline{\beta}  (\beta-R/q) -(\beta-a) \beta (\overline{\beta}-R/q  )   }{ ( \overline{\beta}-R/q )(\beta-R/q) }
 \\
 &= -q(\beta-\overline{\beta})  \frac{ q|\beta|^2 +(a-\beta-\overline{\beta}) R  }{ q^2 |\beta|^2 - q  (\beta+\overline{\beta})R+R^2 }  = -4 (\kappa R)^{1/2} i \,q\, \frac{ q|\beta|^2 +(a-\beta-\overline{\beta}) R  }{ q^2 |\beta|^2 - q  (\beta+\overline{\beta})R+R^2 } .
\end{align*}
On the other hand, we have 
\begin{align*}
& \quad \frac{ (\beta-a) \beta  }{  \beta-R/q   } \Big(  \frac{1}{\beta-R/q}+\frac{1}{2(\beta-\overline{\beta})}-\frac{1}{\beta-a}-\frac{1}{ \beta }  \Big)  +  \frac{ (\overline{\beta}-a) \overline{\beta}  }{ \overline{\beta}-R/q  }   \Big(  \frac{1}{\overline{\beta}-R/q}+\frac{1}{2(\overline{\beta}-\beta)}-\frac{1}{\overline{\beta}-a} -\frac{1}{ \overline{\beta} }  \Big)
\\
&= - \frac{ \beta^2+ (a-2\beta) R/q }{ (\beta-R/q)^2 } - \frac{ \overline{\beta}^2+ (a-2\overline{\beta}) R/q }{ (\overline{\beta}-R/q)^2 }  + \frac{q}{2} \frac{ q|\beta|^2 +(a-\beta-\overline{\beta}) R  }{ q^2 |\beta|^2 - q  (\beta+\overline{\beta})R+R^2 }. 
\end{align*}
Therefore we obtain 
\begin{align*}
&\quad \frac{1}{( R \kappa )^{1/4} }  \bigg(  \frac{1}{\gamma_{11}^{3/2}} \Big( 3-10 (R\kappa)^{1/2} i \frac{ \gamma_{12} }{ \gamma_{11} }\Big) -i\, \overline{  \frac{1}{\gamma_{11}^{3/2}} \Big( 3-10 (R\kappa)^{1/2} i \frac{ \gamma_{12} }{ \gamma_{11} }\Big) }  \, \bigg)
\\
&=  \frac{  4\sqrt{2}(1-i) }{  a   }  \bigg( 2 q\, \frac{ q|\beta|^2 +(a-\beta-\overline{\beta}) R  }{ q^2 |\beta|^2 - q  (\beta+\overline{\beta})R+R^2 } - \frac{ \beta^2+ (a-2\beta) R/q }{ (\beta-R/q)^2 } - \frac{ \overline{\beta}^2+ (a-2\overline{\beta}) R/q }{ (\overline{\beta}-R/q)^2 }  \bigg) .
\end{align*}
Then the desired identity \eqref{desired iden for main} follows from \eqref{def of beta} and \eqref{f conformal}.
\end{proof}

For our main theorem, it remains to prove Proposition~\ref{Prop_energy eval} and Theorem~\ref{Thm_OP fine asymp}.  
Before that, we discuss the critical case in a separate subsection.

\bigskip

\section{Potential theoretic preliminaries} \label{Section_energy}

In this section, we prove Proposition~\ref{Prop_energy eval}. 
Before the proof, let us first consider the radially symmetric case. 

\begin{rem}[Energy for the radially symmetric case]
We consider a general radially symmetric potential
\begin{equation}
    W(z)= w(|z|), \qquad w: \R_+ \to \R.
\end{equation}
Suppose that $\Delta W(z) >0 $ in $\C$. Then it follows from a general theory \cite{ST97} that the associated droplet $S_W$ is given by 
\begin{equation}
S_W= \{ z \in \C :  r_0 \le |z| \le r_1 \},
\end{equation}
where $(r_0,r_1)$ is the unique pair of constants satisfying 
\begin{equation}
r_0 w'(r_0)=0, \qquad r_1 w'(r_1)=2. 
\end{equation}
Furthermore, the energy $I_V[\sigma_Q]$ is given by 
\begin{equation} \label{energy radially sym}
I_W[\sigma_W]=w(r_1)-\log r_1 -\frac14 \int_{r_0}^{r_1} r w'(r)^2\,dr. 
\end{equation}
Using this, we have that for the potential $Q$ in \eqref{Q insertion} with $a=0$, 
\begin{equation}
I_Q[\sigma_Q] \Big|_{a=0}= \frac{3}{4}+\frac{3c}{2}+\frac{c^2}{2}\log c -\frac{(c+1)^2}{2}\log(c+1).
\end{equation}
This coincides with \eqref{energy post}. 
\end{rem}

\subsection{$g$-function}

In Riemann-Hilbert analysis, the first step is to normalise the problem so that it is close to the identity near infinity. 
For this, we need a proper function called the $g$-function, which indeed gives the logarithmic energy of limiting zeros of the orthogonal polynomial.

We first recall the motherbody defined in \cite[Section 1.1]{BBLM15}. 

\begin{defi}[Motherbody; limiting skeleton] \label{Def_motherbody}
The motherbody $\mathcal B$ is defined as follows.
\begin{itemize}
    \item For the post-critical case, let 
    \begin{align}\label{def beta and b}
    \beta=\frac{a^2+1-\sqrt{(1-a^2)^2-4a^2c}}{2a}, \qquad  b=\frac{a^2+1+\sqrt{(1-a^2)^2-4a^2c}}{2a}.
    \end{align}
    Then the simple closed curve $\mathcal B$ is defined by the following three conditions.
    \begin{itemize}
     \item $\mathcal{B}$ is a simple closed curve such that $\beta \in \mathcal{B}$; \smallskip 
\item $\mathcal{B}$ contains $0$ and $a$ in its interior; \smallskip 
\item The following inequality is satisfied on the curve $\mathcal{B}$
\begin{equation}
a^2 \frac{(z-\beta)^2(z-b)^2}{ z^2(z-a)^2 } \,dz <0, 
\end{equation}
where $dz$ is the standard differential.
    \end{itemize} 
    \smallskip 
    \item For the pre-critical case, let $\beta$ be given by \eqref{def of beta}. Let $b=R/q$, where $R$ and $q$ are given by \eqref{f conformal} and \eqref{q equation}.\footnote{There is a typo in \cite[Eq.(1.12)]{BBLM15}, where $b:=\alpha/\rho$ should be replaced with $b:=\rho/\alpha$.}   
    Then $\mathcal B$ is defined by the following three conditions.
    \begin{itemize}
        \item $\mathcal{ B} $ has the endpoints at $\beta$ and $\overline{\beta}$; \smallskip 
        \item  $\mathcal{ B}$ intersects the negative real axis; \smallskip 
        \item The following inequality is satisfied on the curve $\mathcal{B}$
    \begin{equation}
    \frac{ (z-b)^2(z-\beta) (z-\overline{\beta}) }{ (z-a)^2 z^2 } \,dz<0.  
    \end{equation}
    \end{itemize}
\end{itemize}
\end{defi}

Here, we use a slight abuse of notation by employing both $\beta$ and $b$ to denote variables in both post- and pre-critical regimes.
See Figure~\ref{Fig_motherbody} for an illustration of the motherbody.

\begin{figure}[h!]
	\begin{subfigure}{0.33\textwidth}
		\begin{center}		\includegraphics[height=0.8\textwidth]{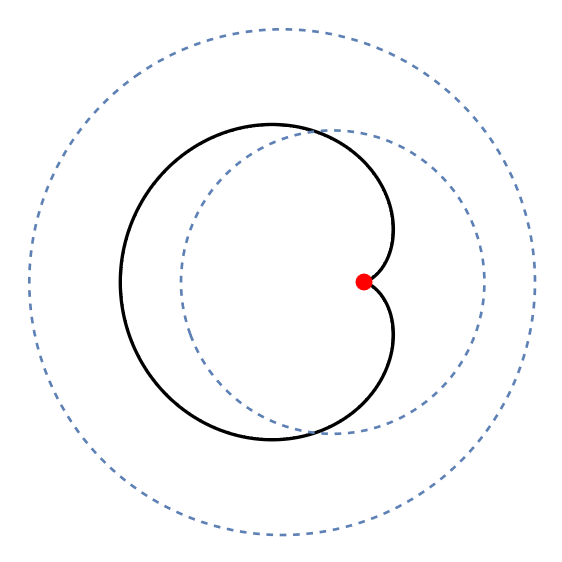}
		\end{center}
		\subcaption{Post-critical; $a=1/4$}
	\end{subfigure}	
 \quad
	\begin{subfigure}{0.33\textwidth}
		\begin{center}	
	 \includegraphics[height=0.8\textwidth]{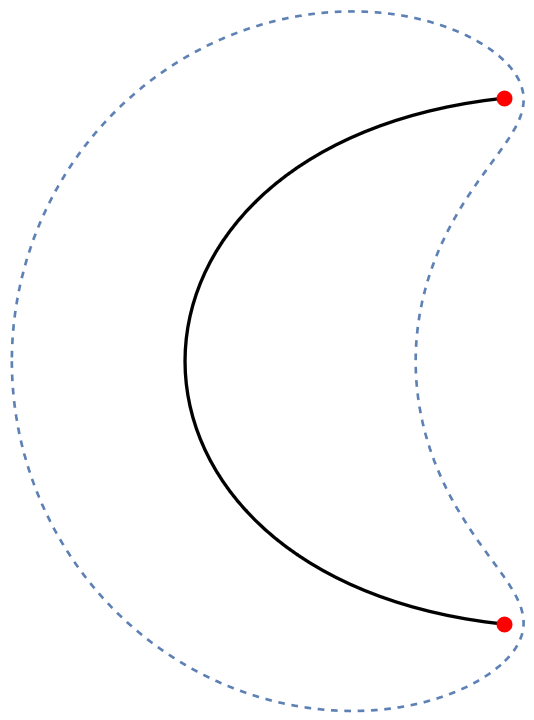}
		\end{center}
		\subcaption{Pre-critical; $a=1$}
	\end{subfigure}	
    \caption{ The solid black lines illustrate the motherbody in Definition~\ref{Def_motherbody}, where $c = 9/16$. The dashed lines represent the boundary of the droplet, while the red dots indicate $\beta$ in the post-critical case (A) and both $\beta$ and $\overline{\beta}$ in the pre-critical case (B). }
    \label{Fig_motherbody}
\end{figure}

We now recall the $g$-function defined in \cite[Definition 1.2]{BBLM15}.

\begin{defi}[$g$-function]\label{def gfunction}
The $g$-function is defined as follows.
\begin{itemize}
    \item For the post-critical case, 
    \begin{equation}
    g(z)=\begin{cases}
      \displaystyle  \log z+c \log \Big( \frac{z}{z-a} \Big),& z\in \ext(\mathcal B),
        \smallskip 
        \\
       az+\re\big((1+c)\log\beta-c\log(\beta-a)-\beta a \big),& z\in \mbox{int}(\mathcal B).
    \end{cases}
\end{equation} 
We also define 
\begin{equation}
\phi(z) = \int_{ \beta }^z y(s)\,ds, \qquad y(z)= \pm a \frac{(z-b)(z-\beta)}{z(z-a)}.
\end{equation}
\item For the pre-critical case, let 
\begin{equation}
V(z)= a z -c \log(z-a) + (c+1) \log z. 
\end{equation}
Write
\begin{equation} \label{def of y and phi}
y(z) = \frac{ a(z-R/q) \sqrt{ (z-\beta) (z-\overline{\beta})  }  }{ (z-a)z  },  \qquad 
\phi(z) = \int_{ \beta }^z y(s)\,ds, \qquad (z \in \C \setminus \mathcal B). 
\end{equation}
Then the $g$-function is defined by
\begin{equation} 
g(z) = \frac12 \Big( V(z)-\phi(z)+\ell \Big), \qquad (z \in \C \setminus \mathcal B), 
\end{equation}
where the constant $\ell \in \R$ is defined such that 
\begin{equation} \label{g log z asymptotic}
\lim_{z \to \infty} (g(z)-\log z) =0.  
\end{equation} 
\end{itemize}
\end{defi}

The $g$-function plays an important role in the asymptotic behaviour of orthogonal polynomials, see Theorem~\ref{Thm_OP fine asymp}.  
This indeed comes from the fact that if we write $\zeta_j$ for the zeros of $p_n$ and $\nu$ for their limiting distribution, we have formally
\begin{align*}
p_N(z) = \prod_{j=1}^N (z-\zeta_j) = \exp\Big( \sum_{j=1}^N \log(z-\zeta_j) \Big) \sim \exp \Big( N \int \log (z-\zeta) \,d\nu(\zeta) \Big).  
\end{align*}
Thus, the $g$-function is naturally defined by $\int \log(z-\zeta) \,d\nu(\zeta)$ in a more general context, and Definition~\ref{def gfunction} provides its evaluations in our present context.

For our purposes explained in Section~\ref{Section_Overall}, let us compute the asymptotic behaviour of $e^{Ng(z)}$ as $z \to \infty$.
In the end, this will provide the leading order asymptotic behaviour of the derivative of the partition function.

\begin{lem}[\textbf{Large-argument behaviour of the $g$-function}] \label{Lem_g asymptotic}
As $z \to \infty$, we have the following.
\begin{itemize}
    \item For the post-critical case, 
    \begin{equation}
e^{Ng(z)} = z^N \Big( 1+ ac\, \frac{N}{z} +O\Big(\frac{1}{z^2}\Big) \Big), \qquad z \to \infty.
\end{equation}
\item For the pre-critical case, 
\begin{equation}
e^{ N g(z) }=  z^N \bigg(1  - \Big(\frac{1}{4a}+\frac{a}{2} -\frac{1}{2aq^2} -\frac{a^3q^4}{4} \Big) \frac{N}{z} +O\Big(\frac{1}{z^2}\Big) \bigg),\qquad z \to \infty.
\end{equation}
\end{itemize}
\end{lem}

\begin{rem}
Notice that the coefficients of the $1/z$ term in Lemma~\ref{Lem_g asymptotic} coincide with the leading terms in Proposition~\ref{Prop_Asymptotic coefficient}. 
\end{rem}

\begin{proof}[Proof of Lemma~\ref{Lem_g asymptotic}]
This follows from the definition. 
For instance, for the pre-critical case, we have  
\begin{equation} \label{g' precritical}
2g'(z)= a- \frac{c}{z-a}+\frac{c+1}{z}-\frac{ a(z-R/q) \sqrt{ (z-f(z_+))(z-f(z_-)) }  }{ (z-a)z  }.
\end{equation}
Then it follows that   
\begin{equation}
g'(z)= \frac{1}{z} + \Big(\frac{1+2a^2}{4a} -\frac{1}{2aq^2} -\frac{a^3q^4}{4} \Big) \frac{1}{z^2} +O\Big(\frac{1}{z^3}\Big).
\end{equation}
This gives 
\begin{equation}
 g(z) =\log z  - \Big(\frac{1}{4a}+\frac{a}{2} -\frac{1}{2aq^2} -\frac{a^3q^4}{4} \Big) \frac{1}{z} +O\Big(\frac{1}{z^2}\Big),
\end{equation}
which completes the proof. 
\end{proof}

To compute the energy \eqref{energy pre}, we need to evaluate the $g$-function at the point $a.$

\begin{lem} \label{Lem_g(a) eval}
For the pre-critical case, we have 
\begin{align}
\begin{split} \label{re g(a) eval}
 \re g(a) 
 & = \frac{a^2q^2}{2} -\frac12 + \log \Big( \frac{1+a^2q^2}{2a q^2}\Big) + c \log \Big(  \frac{ 1+a^2q^2 }{1+a^2q^2-2a^2q^4}  \Big). 
\end{split}
\end{align}
\end{lem}

This lemma is shown in the next subsection. 

\begin{rem}
We encounter surprising cancellations when differentiating \eqref{re g(a) eval} with respect to $a$. Namely,  
\begin{equation}
\mathfrak{d}_a \re g(a)=    a q^2,
\end{equation} 
where we have used \eqref{q' q}. We currently lack a good intuition for such a significant simplification.
\end{rem}

\subsection{Robin's constant and energy}

Recall that the equilibrium measure $\mu_Q$ satisfies the variational equality (also known as the Euler–Lagrange equation):
\begin{align}
\label{vari eq_prec}
\int \log \frac{1}{|z-w|}\,d\mu_Q(z)+\frac{Q(w)}{2}=C, \qquad \text{if }w \in S_Q 
\end{align} 
where $C \equiv C(a)$ is called the modified Robin's constant, see \cite[p.27]{ST97}. 
Then we have 
\begin{equation}
\int_{ \mathbb{C}^2 } \log \frac{1}{ |z-w| }\, d\mu_Q(z)\, d\mu_Q(w)= C(a) -\frac12 \int Q(z) \,d\mu_Q(z)
\end{equation}
and 
\begin{align}
\begin{split} \label{energy Robin}
I_Q[\mu_Q] = C(a) + \frac12 \int Q(z) \,d\mu_Q(z). 
\end{split}
\end{align}

By Lemma~\ref{Lem_g(a) eval} and \eqref{energy Robin}, it suffices to show the following lemma to complete the proof of Proposition~\ref{Prop_energy eval}. 
Note that the inequality \eqref{energy inequality} was shown in Remark~\ref{Rem_positivity}.

\begin{lem}[\textbf{Robin's constant and energy}]  \label{Lem_energy}
We have the following. 
\begin{itemize}
    \item For the post-critical case, we have 
    \begin{equation} \label{Robin post}
C(a)= \frac{c+1}{2} -\frac{c+1}{2} \log(c+1)
\end{equation}
and 
\begin{equation}
\frac12 \int Q \,d\mu_Q = c+\frac14 +\frac{c^2 }{2}\log c -\frac{ c(1+c ) }{2} \log(1+c)-ca^2. 
\end{equation}
In particular, we have \eqref{energy post}. 
\smallskip 
  \item For the pre-critical case, we have 
  \begin{equation}
\begin{split} \label{Robin pre}
C(a) 
&=  c \log q -(c+1) \log \Big( \frac{1+a^2q^2}{2a q}\Big) +\frac{3}{4} + \frac{a^2}{4} + \frac{c}{2} - \Big(\frac{a^2}{4}+c+\frac{3}{4}\Big)a^2q^2 + \frac{a^4q^4}{2}
\end{split}
\end{equation}
and 
\begin{equation}
\begin{split} 
\frac12 \int Q \,d\mu_Q 
&= \frac{3}{8} + \frac{a^2}{4} + \frac{c}{2} - \Big( \frac{a^2}{4}+c+\frac12 \Big) a^2q^2+ \frac{3}{8} a^4q^4 -c \re g(a). 
\end{split}
\end{equation}
In particular, we have 
\begin{equation}
\begin{split}   
 I_Q[\mu_Q]
&= c \Big( \log q - \re g(a) \Big) -(c+1) \log \Big( \frac{1+a^2q^2}{2a q}\Big)+ \frac{5}{8}+ \frac{a^2}{4} + \frac{1}{4a^2q^4} - \frac{1}{2q^2} + \frac{a^2q^2}{4} - \frac{a^4q^4}{8}. 
\end{split}
\end{equation}
\end{itemize}
\end{lem}

\begin{rem}[Point charge insertion at infinity]
It is easy to check from \eqref{q equation} that 
\begin{equation} \label{q asymp}
q= \frac{1}{a} \Big( 1- \frac{c}{a^2} +O( \frac{1}{a^4} ) \Big), \qquad (a \to \infty).
\end{equation}
This in turn gives that 
\begin{equation}
R = 1+O(\frac{1}{a^4}), \qquad \kappa= \frac{c}{a^2} +O(\frac{1}{a^4}), \qquad  f(z)=z+O(\frac{1}{a}), \qquad (a \to \infty),
\end{equation}
where $R,\kappa$ and $f$ are given by \eqref{f conformal}. 
Therefore the droplet tends to the unit disc when $a \to \infty$, cf. Figure~\ref{Fig_transition}. 
Let us also mention that by \eqref{q asymp}, we have 
\begin{equation} \label{energy infintiy}
\lim_{a \to \infty} \Big( I_Q[\mu_Q] + 2c \log a \Big) = \frac34, \qquad 
\lim_{ a \to \infty }  \mathcal{F}(a,c)=0.
\end{equation}
From the viewpoint of the free energy expansions \eqref{Z expansion}, these asymptotic behaviours are consistent with \eqref{Z Gin asy} and \eqref{ZN hat ZN rel}. 
Furthermore, by \eqref{g log z asymptotic}  and \eqref{q asymp}, we also have  
\begin{equation}
\lim_{a \to \infty} \Big( C(a)+c \log a \Big)= \frac12, \qquad \lim_{a \to \infty} \Big( \frac12 \int Q(z) \,d\mu_Q(z)+c \log a \Big)= \frac14. 
\end{equation}
Note also that by combining Lemma~\ref{Lem_deri of N2 1 terms} and \eqref{q asymp}, one can check that as $a \to \infty,$
\begin{equation}
 \partial_a I_Q[\mu_Q]= -\frac{2c}{a}+O(a^{-3}), \qquad   \partial_a \mathcal{F}(a,c) =  -\frac{2c^2}{a^7}+O(a^{-9}) .
\end{equation} 
\end{rem}
 
We now prove Lemma~\ref{Lem_energy}.

\begin{proof}[Proof of Lemma~\ref{Lem_energy}]
The modified Robin constant was computed in \cite[Lemma 7.2]{BBLM15}. Let us mention that the Robin constant $\ell_{\rm 2D}$ in \cite{BBLM15} corresponds to 
$\ell_{\rm 2D}=-2 \, C(a).$

It remains to compute $\int Q \,d\mu_Q.$
We first show the post-critical case when $c< c_{\rm cri}$. 
By \cite[Lemma 2.4]{By23a}, we have that for $R>0$ and $p \in \mathbb{C}$, 
\begin{equation}
\int_{ \mathbb{D}(p,R) } \log|z-w| \,dA(z)=
\begin{cases}
\displaystyle R^2 \log|w-p| &\text{if }w \notin \mathbb{D}(p,R),
\smallskip 
\\
\displaystyle R^2 \log R-\frac{R^2}{2}+\frac{|w-p|^2}{2} &\text{otherwise}.
\end{cases}
\end{equation}
Using this, we have 
\begin{align*}
-2 c \int_S \log|w-a| \,dA(w)
 &=-2 c  \Big(  \int_{ \mathbb{D}(0,\sqrt{1+c}) } \log|w-a|\,dA(w) -  \int_{ \mathbb{D}(a,\sqrt{c}) } \log|w-a|\,dA(w) \Big) 
 \\
 &=  c^2 (\log c-1)  -2 c   \int_{ \mathbb{D}(0,\sqrt{1+c}) } \log|w-a|\,dA(w)   
 \\
 &= c +c^2 \log c - c(1+c ) \log(1+c) -c a^2.
\end{align*}
On the other hand, note that 
\begin{align*}
\begin{split}
\int_S |w|^2 \,dA(w) & = \int_{ \mathbb{D}(0,\sqrt{1+c}) } |w|^2 \,dA(w) - \int_{ \mathbb{D}(a,\sqrt{c}) } |w|^2 \,dA(w)
= \frac12 (c+1)^2 - \int_{ \mathbb{D}(a,\sqrt{c}) } |w|^2 \,dA(w). 
\end{split}
\end{align*}
Here by Green's formula, and using the map $z \mapsto w=\sqrt{c}z+a$, 
\begin{align*}
\begin{split}
\int_{ \mathbb{D}(a,\sqrt{c}) } |w|^2 \,dA(w) &= \frac{1}{2\pi i}  \int_{ |w-a|=\sqrt{c} } \frac{w \bar{w}^2 }{2} \,dw = \frac{1}{2\pi i} \int_{ \partial \mathbb{D} } \frac{ \sqrt{c} }{2} (\sqrt{c}z+a)( \sqrt{c}\bar{z}+a )^2  \,dz
\\
&= \frac{1}{2\pi i} \int_{ \partial \mathbb{D} } \frac{ \sqrt{c} }{2} (\sqrt{c}z+a)( \sqrt{c}/z+a )^2  \,dz  = \frac12 c^2+ca^2. 
\end{split}   
\end{align*}
This gives 
\begin{equation}
\int_S |w|^2 \,dA(w) = c+\frac12 -c a^2. 
\end{equation} 
We have shown the post-critical case.

\medskip 
Next, we consider the pre-critical case when $c > c_{\rm cri}$.
By Green's formula,
\begin{align*}
\int_S |z|^2 \,dA(z) & = \frac{1}{2\pi i} \int_{\partial S} \frac{ z \bar{z}^2   }{2}  \,dz 
= \frac{1}{2\pi i} \int_{\partial \mathbb{D}}  \frac{ f(w)f(1/w)^2 }{2} \,f'(w)\,dw, 
\end{align*}
where $f$ is given by \eqref{f conformal}.
Note that $0< q<1$ and
\begin{equation}
f(z)  = \frac{R\,z(z-z_0) }{z-q}, \qquad f(1/z) = \frac{R(1-z_0z)}{ z (1-qz) }, \qquad z_0= q+\frac{\kappa}{qR}. 
\end{equation}
Then by the residue calculus, 
\begin{align}
\begin{split}
\int_S |z|^2 \,dA(z) &=  \underset{w=q}{\textup{Res}} \Big[ \frac{ f(w)f(1/w)^2 }{2} \,f'(w) \Big]  +   \underset{w=0}{\textup{Res}} \Big[ \frac{ f(w)f(1/w)^2 }{2} \,f'(w) \Big] 
\\
& = \frac{1}{4} \Big(1 + \frac{1}{a^2 q^4} - \frac{2}{q^2} - a^4 q^4 + a^2 (1 + 2 q^2)\Big).
\end{split}
\end{align}
Here we have used \begin{equation} \label{1/q^4 trans}
    \frac{1}{q^4}= - 2a^4 q^2 + (a^2+4c+2) a^2 . 
 \end{equation}

Note that by \cite[Lemma 3.7]{BBLM15}, we have  
\begin{equation}
\int_S \log |z-w|\,dA(w) = \re g(z), \qquad z \not \in S. 
\end{equation}
This in particular gives
\begin{equation}
 \int_S \log|z-a| \,dA(z) =  \re g(a), 
\end{equation}
which completes the proof. 
\end{proof}

We now prove Lemma~\ref{Lem_g(a) eval}, which completes the proof of Proposition~\ref{Prop_energy eval}.

\begin{proof}[Proof of Lemma~\ref{Lem_g(a) eval}]
By \cite[Eq.(7.9)]{BBLM15}, it follows that 
\begin{align}
\begin{split}
\int_S \log |z-w| \,dA(w) = \re \int_{z_0}^z \mathcal{S}(\zeta)\,d\zeta - c \log |z-a| + \frac{ |z_0|^2 }{2} -C(a),
\end{split}
\end{align}
where $\mathcal{S}$ is the Schwarz function associated with the droplet, see \cite[Definition 2.1]{BBLM15}.
Choose $z_0=f(1).$ 
Then since 
\begin{equation}
\mathcal{S}(f(z))=f(1/z), 
\end{equation}
we have 
\begin{align}
\int_{z_0}^z \mathcal{S}(\zeta)\,d\zeta &= \int_{ 1 }^{ u } f(1/w) f'(w) \,dw, \qquad  f(u)=z. 
\end{align}
Note that 
\begin{align*}
f(1/z)f'(z) & = \frac{(1+a^2q^2)(1-a^2q^2+2a^2q^4)}{4a^2q^4} \frac{1}{z} + \frac{1-q^2-a^2q^2+a^2q^4}{2q} \frac{1}{(z-q)^2} + \frac{c}{z-1/q}  - \frac{c}{z-q} 
\\
&= \frac{c+1}{z} + \frac{1-q^2-a^2q^2+a^2q^4}{2q} \frac{1}{(z-q)^2} + \frac{c}{z-1/q}  - \frac{c}{z-q}. 
\end{align*}
On the other hand, 
\begin{align*}
&\quad \log|f(1)-a|-\log |z-a| =\log|f(1)-f(1/q)| -\log |f(u)-f(1/q)|
\\
&= \int_1^u \frac{1}{z-q} - \frac{1}{z-1/q} - \frac{1+a^2q^2}{(1+a^2q^2)z-2a^2q^3}  \,dz . 
\end{align*}
Combining the above, we have
\begin{align*}
&\quad \re \int_{z_0}^z \mathcal{S}(\zeta)\,d\zeta - c \log |z-a| + c \log|f(1)-a|
\\
&= \int_1^u \frac{c+1}{z} + \frac{1-q^2-a^2q^2+a^2q^4}{2q} \frac{1}{(z-q)^2} - \frac{c}{z-2a^2q^3/(1+a^2q^2)}  \,dz 
\end{align*}
Note that 
\begin{align*}
\int_1^{1/q}  \frac{c+1}{z}   \,dz = -(c+1) \log q, \qquad 
 \int_1^{1/q}  \frac{1-q^2-a^2q^2+a^2q^4}{2q} \frac{1}{(z-q)^2}   \,dz = \frac{1}{2q}-\frac{a^2q}{2},  
\end{align*}
and 
\begin{align*}
&\quad - \int_1^{1/q}   \frac{c}{z-2a^2q^3/(1+a^2q^2)}  \,dz - c \log |f(1)-a|
\\
&= -c \log \Big( \frac{1+a^2q^2-2a^2q^4}{ q(1+a^2q^2) } \Big) +c \log  \Big( \frac{1+a^2q^2-2a^2q^3}{ 1+a^2q^2 } \Big) -c \log  \Big( \frac{1+a^2q^2-2a^2q^3}{ 2aq^2 } \Big)
\\
&= -c \log \Big( \frac{1+a^2q^2-2a^2q^4}{ q(1+a^2q^2) } \Big) +c \log  \Big( \frac{2aq^2}{ 1+a^2q^2 } \Big) = c \log \Big(  \frac{ 2aq^3 }{1+a^2q^2-2a^2q^4}  \Big) . 
\end{align*}
Therefore it follows that 
\begin{align*}
&\quad \lim_{z \to a} \Big( \re \int_{z_0}^z S(\zeta)\,d\zeta - c \log |z-a| \Big)
\\
&= \frac{1}{2q}-\frac{a^2q}{2} -\frac{(1+a^2q^2)(1-a^2q^2+2a^2q^4)}{4a^2q^4} \log q  + c \log \Big(  \frac{ 2aq^3 }{1+a^2q^2-2a^2q^4}  \Big) .
\end{align*}
Then we have shown that
\begin{align*}
\re g(a)& =\frac{a^2}{8} + \frac{1}{8a^2q^4} - \frac{1}{4q^2} + \frac{a^2q^2}{2} -C(a)  -(c+1) \log q  + c \log \Big(  \frac{ 2aq^3 }{1+a^2q^2-2a^2q^4}  \Big). 
\end{align*}
Now the lemma follows from \eqref{Robin pre}. 
\end{proof}


\bigskip

\section{Riemann-Hilbert analysis and Fine asymptotic behaviour} \label{Section_RH analysis}

In this section, we show Theorem~\ref{Thm_OP fine asymp}.  
The leading-order term of $p_N$ was derived via the corresponding Riemann-Hilbert problem in \cite{BBLM15}. However, identifying a meaningful subleading term for $p_N$ is also significant and requires careful analysis. Here, we refine this analysis by employing the partial Schlesinger transform to derive the asymptotic behavior of $p_N$ with more precise error estimates.

Recall that $\Gamma$ be the contour in the Riemann-Hilbert problem \eqref{RHP Y}. 
We deform the contour $\Gamma$ so that it matches with the contour ${\mathcal{B}}$ in Definition~\ref{Def_motherbody}. 
We choose the contour $\Gamma_+$ and $\Gamma_-$ following the steepest descent paths from $\beta$ such that $\re \phi(z)<0$ on those contours. Here $\phi$ is given in Definition~\ref{def gfunction}. 
More precisely, we choose the one inside $\mbox{int}({\mathcal{B}})$ to be $\Gamma_+$, while the one inside $\ext({\mathcal{B}})$ to be $\Gamma_-$. The domains $\Omega_\pm$ are defined by the open sets enclosed by ${\mathcal{B}}$ and $\Gamma_\pm$ respectively.

\medskip 

Let us define the matrix function $A(z)$ by 
\begin{equation}\label{A transform}
    A(z):=\begin{cases}
        e^{-\frac{N\ell}{2}\sigma_3}Y(z)e^{-N(g(z)-\ell/2)\sigma_3},\quad &z\in\C\setminus(\Omega_+\cup\Omega_-),
        \smallskip 
        \\
        e^{-\frac{N\ell}{2}\sigma_3}Y(z)\begin{pmatrix}
1&0\\-1/\omega_{n,N}(z)&1
\end{pmatrix}e^{-N(g(z)-\ell/2)\sigma_3},\quad &z\in \Omega_+,
   \smallskip 
\\
  e^{-\frac{N\ell}{2}\sigma_3}Y(z)\begin{pmatrix}
1&0\\1/\omega_{n,N}(z)&1
\end{pmatrix}e^{-N(g(z)-\ell/2)\sigma_3},\quad &z\in \Omega_-,
    \end{cases}
\end{equation} 
where $Y$ and $\omega_{n,N}$ are given in Definition~\ref{Def_RHP Y}, while $\ell$ and $g(z)$ are given in Definition~\ref{def gfunction}. 
Here, 
\begin{equation}
\sigma_3= \begin{pmatrix}
1 & 0
\\
0 & -1
\end{pmatrix}
\end{equation}
is the third Pauli matrix. 
Then by \eqref{RHP Y}, one can check that $A$ satisfies the following Riemann-Hilbert problem:
\begin{equation}  \label{RHP for A pre}
\begin{cases}
 A_+(z)=  A_-(z)\begin{pmatrix}
1&0\\e^{N\phi(z)}&1
\end{pmatrix},&\quad z\in \Gamma_\pm,
\smallskip 
\\
A_+(z)=  A_-(z)\begin{pmatrix}
0&1\\-1&0
\end{pmatrix},&\quad z\in {\mathcal{B}},
\smallskip 
\\
A_+(z)=  A_-(z)\begin{pmatrix}
1&e^{-N\phi(z)}\\0&1
\end{pmatrix},&\quad z\in \Gamma\setminus{\mathcal{B}},
\smallskip 
\\
A(z)= I+{ O}(z^{-1}) ,&\quad z\to\infty,
\smallskip 
\\
A(z)\mbox{\quad is holomorphic},&\quad \mbox{otherwise}.
\end{cases}
\end{equation}
Note that $\re(\phi(z))$ is negative on $\Gamma_{\pm}$, while it is positive on $\Gamma\setminus {\mathcal{B}}$. 
Therefore, when $z$ is away from $ {\mathcal{B}}$, the jump matrices are close to identity exponentially.

Next, we define the global parametrix, cf. \cite[Eqs.(4.14), (5.1)]{BBLM15}.

\begin{defi}[Global parametrix] \label{Def_Global parametrix}
The global parametrix $\Phi$ is defined as follows.
\begin{itemize}
    \item For the post-critical case,
    \begin{equation} \label{def of Phi post}
\Phi(z):= \begin{cases}
I, &z\in \ext(\mathcal B),\\
    \begin{pmatrix}
    0&1 
    \\
    -1 & 0
\end{pmatrix}, &z\in \mbox{int} (\mathcal B).
\end{cases}
\end{equation}
\item For the pre-critical case, \begin{equation} \label{def of Phi}
\Phi(z):= \sqrt{ R F'(z)}\begin{pmatrix}
    1&\dfrac{\sqrt{\kappa R}}{RF(z)-  Rq}
    \smallskip 
    \\
    -\dfrac{\sqrt{\kappa R}}{RF(z)-  Rq} & 1
\end{pmatrix}
\end{equation}
where $F$ is given by \eqref{def of F}. 
\end{itemize}
\end{defi}

By construction, $\Phi(z)$ satisfies the following Riemann-Hilbert problem:
 \begin{equation}\label{global RHP}
\begin{cases}
\Phi(z)\mbox{\quad is holomorphic},&\quad z\in \C\setminus {\mathcal{B}},\\
\Phi_+(z)=  \Phi_-(z)\begin{pmatrix}
0&1\\-1&0
\end{pmatrix},&\quad z\in {\mathcal{B}},\\
\Phi(z)= I+{ O}(z^{-1}) ,&\quad z\to\infty.\end{cases}
 \end{equation}
For the remaining Riemann-Hilbert analysis, we should treat post- and pre-critical cases separately. 
We shall first discuss the pre-critical case, as it is the more difficult one.

\subsection{Pre-critical case}

We denote $D_w$ as a small disk centered at $w \in \C$ with a finite radius. Let us first define the local coordinates that will be used frequently.

\begin{defi}[Local coordinates]
Recall that $\phi$ is given by \eqref{def of y and phi}. 
The local coordinate $\zeta$ is defined by 
\begin{equation}
 \frac43 \zeta(z)^{3/2} := N(\phi(z)-\phi(\overline{\beta})), \quad z\in D_{\overline \beta},
\end{equation}
cf. \cite[Eq.(4.15)]{BBLM15}.  
Here $\zeta$ maps $\Gamma_{b\overline\beta}$ into $\R_+$, maps ${\mathcal B}$ into $\R_-$, maps $\Gamma_{+}$ into ray $\gamma_{\overline\beta}^+:=[0,e^{2\pi i/3\infty})$, and maps $\Gamma_{-}$ into ray $\gamma_{\overline\beta}^-:=[0,e^{-2\pi i/3\infty})$.
Similarly, we define the local coordinate $\xi$ by 
\begin{equation}
    \frac{4}{3}\xi(z)^{3/2}:=N\phi(z),\quad z\in D_\beta,
\end{equation} such that $\xi$ maps $\Gamma_{b\beta}$ into $\R_+$, maps ${\mathcal B}$ into $\R_-$, maps $\Gamma_{+}$ into ray $\gamma_{\beta}^+:=[0,e^{-2\pi i/3\infty})$, and maps $\Gamma_{-}$ into ray $\gamma_{\beta}^-:=[0,e^{2\pi i/3\infty})$.
\end{defi}

\begin{lem}[\textbf{Asymptotic of the local coordinates}] \label{Lem_local coord asymp}
We have 
\begin{align}
 \label{zeta(z) beta bar expansion}
\frac{ \zeta(z) }{ N^{2/3} }&=\gamma_{11}(z-\overline\beta)+\frac{\gamma_{12}}{2}(z-\overline\beta)^2+O(z-\overline\beta)^3, \qquad z \to \overline{\beta},
\\
 \label{xi(z) beta expansion}
\frac{ \xi(z) }{ N^{2/3} }&= \overline{\gamma_{11}}(z-\beta)+\frac{\overline{\gamma_{12}}}{2} (z-\beta)^2+O(z-\beta)^3, \qquad z \to \beta,
\end{align}
where 
\begin{align} \label{gamma 11 asym}
\gamma_{11} &=   \frac{1}{2^{2/3}}  \, \Big( \frac{ a(\overline{\beta}-R/q) \sqrt{ \overline{\beta}-\beta }  }{ (\overline{\beta}-a) \overline{\beta}  } \Big)^{2/3},
\\  \label{gamma 12 asym}
\gamma_{12}  &=  \frac45 \Big(  \frac{1}{\overline{\beta}-R/q}+\frac{1}{2(\overline{\beta}-\beta)}-\frac{1}{\overline{\beta}-a}-\frac{1}{ \overline{\beta} }  \Big) \gamma_{11}.  
\end{align}
\end{lem}
\begin{proof} 
By definition, we have 
\begin{align*}
\zeta'(z) &= \frac23 \Big(\frac{3N}4 \Big)^{2/3} \, \Big(    \int_{ \overline{\beta} }^z y(s)\,ds\Big)^{-1/3} \, y(z),   
\qquad 
\frac{\zeta''(z)}{ \zeta'(z) } = -\frac13 \Big(    \int_{ \overline{\beta} }^z y(s)\,ds\Big)^{-1} \, y(z) +   \frac{ y'(z) }{ y(z) } . 
\end{align*}
Let us write 
\begin{align*}
y(z) = h(z)  \sqrt{z-\overline{\beta}}, \qquad h(z):=\frac{ a(z-R/q) \sqrt{ z-\beta }  }{ (z-a)z  }. 
\end{align*}
Then we have 
\begin{align*}
y(z) 
& =  h( \overline{\beta} ) (z-\overline{\beta})^{1/2}  \Big( 1+ \frac{ h'( \overline{\beta} )  }{ h( \overline{\beta} )   } (z-\overline{\beta}) +O((z-\overline{\beta})^2) \Big),
\\
y'(z) 
& = \frac12 h( \overline{\beta} ) (z-\overline{\beta})^{-1/2} \Big( 1+ 3\frac{ h'( \overline{\beta} )  }{ h( \overline{\beta} )   } (z-\overline{\beta}) +O((z-\overline{\beta})^2) \Big). 
\end{align*}
Using these, direct computations give
\begin{align*}
\Big( \int_{ \overline{\beta} }^z y(s)\,ds \Big)^{-1} &= \frac{3}{2} \frac{1}{ h( \overline{\beta} ) } (z-\overline{\beta})^{-3/2} \Big( 1- \frac{3}{5} \frac{ h'( \overline{\beta} ) }{ h( \overline{\beta} ) } (z-\overline{\beta})  +O((z-\overline{\beta})^2) \Big),
\\
\frac{y'(z)}{y(z)} &= \frac12 (z- \overline{\beta} )^{-1} + \frac{ h'( \overline{\beta} )  }{ h( \overline{\beta} )   }  +O( z-\overline{\beta} ),
\end{align*}
which leads to 
$$
 -\frac13 \Big(    \int_{ \overline{\beta} }^z y(s)\,ds\Big)^{-1} \, y(z)  =   -\frac12 (z- \overline{\beta} )^{-1} -\frac15 \frac{ h'( \overline{\beta} )  }{ h( \overline{\beta} )   }  +O( z-\overline{\beta} ). 
$$
Then the lemma follows from 
 \begin{equation*}
\frac{h'(z)}{h(z)}= \frac{1}{z-R/q}+\frac{1}{2(z-\beta)}-\frac{1}{z-a}-\frac{1}{z}. 
\end{equation*}
\end{proof} 

We now define the local parametrices. 

\begin{defi}[Riemann-Hilbert problems for local parametrices] \label{Def_RHP}
We consider the Riemann-Hilbert problems.
\begin{itemize}
    \item 
Inside $D_{\overline\beta}$, we define the following Riemann-Hilbert problem:
\begin{equation}
    \begin{cases}
    [P_{\overline\beta}(\zeta)]_+=\begin{pmatrix}
0&-1\\1&0
\end{pmatrix}[P_{\overline\beta}(\zeta)]_-\begin{pmatrix}
0&1\\-1&0
\end{pmatrix},\quad &\zeta\in \R_-,
\smallskip 
\\
        [P_{\overline\beta}(\zeta)]_+=[P_{\overline\beta}(\zeta)]_-\begin{pmatrix}
1&0\\e^{\frac{4}{3}\zeta^{3/2}}&1
\end{pmatrix},\quad &z\in \gamma_{\overline \beta}^{\pm},
\smallskip 
\\
[P_{\overline\beta}(\zeta)]_+=[P_{\overline\beta}(\zeta)]_-\begin{pmatrix}
1&e^{-\frac{4}{3}\zeta^{3/2}}\\0&1
\end{pmatrix},\quad &\zeta\in \R_+,
\smallskip 
\\
P_{\overline\beta}(\zeta)=I+O(\zeta^{-3/2}),\quad & \zeta\to\infty,
\smallskip 
\\
P_{\overline\beta}(\zeta)=O(\zeta^{-1/4}),\quad & \zeta\to 0.
\end{cases}
\end{equation} 
Then $\Phi(z)P_{\overline\beta}(z)$ satisfies the same jump conditions of $A(z)$.
\smallskip 
\item Inside $D_{\beta}$, we define the following Riemann-Hilbert problem:
\begin{equation}
    \begin{cases}
    [\widehat{P}_\beta(\xi)]_+=\begin{pmatrix}
0&1\\-1&0
\end{pmatrix}[\widehat{P}_\beta(\xi)]_-\begin{pmatrix}
0&-1\\1&0
\end{pmatrix},\quad &\xi\in \R_-,
\smallskip 
\\
        [\widehat{P}_\beta(\xi)]_+=[\widehat{P}_\beta(\xi)]_-\begin{pmatrix}
1&0\\-e^{\frac{4}{3}\xi^{3/2}}&1
\end{pmatrix},\quad &\xi\in \gamma_{\beta}^{\pm},
\smallskip 
\\
[\widehat{P}_\beta(\xi)]_+=[\widehat{P}_\beta(\xi)]_-\begin{pmatrix}
1&-e^{-\frac{4}{3}\xi^{3/2}}\\0&1
\end{pmatrix},\quad &\xi\in \R_+,
\smallskip 
\\
\widehat{P}_\beta(\xi)=I+O(\xi^{-3/2}),\quad & \xi\to\infty,
\smallskip 
\\
\widehat{P}_\beta(\xi)=O(\xi^{-1/4}),\quad & \xi\to 0.\end{cases}
\end{equation} 
Then $\Phi(z)\widehat{P}_{\beta}(z)$ satisfies the same jump conditions of $A(z)$. 
\end{itemize}

The jump contours and the corresponding orientations can be seen from Figure \ref{Fig_jump contours}.
\end{defi}

\begin{figure}[t]  
\centering
\begin{tikzpicture}[scale=0.7]

\draw[very thick,black, postaction={decorate, decoration={markings, mark = at position 0.5 with {\arrow{>}}}} ] (-4,0) -- (0,0);
\draw[very thick,black,postaction={decorate, decoration={markings, mark = at position 0.5 with {\arrow{>}}}} ] (0,0) -- (4,0);

\draw[very thick,black,postaction={decorate, decoration={markings, mark = at position 0.5 with {\arrow{>}}}} ] (-2,3.464) -- (0,0);
\draw[very thick,black,postaction={decorate, decoration={markings, mark = at position 0.5 with {\arrow{>}}}} ] (-2,-3.464) -- (0,0);

\foreach \Point/\PointLabel in {(0,-0.1)/0, (-1,2.6)/\gamma_{\overline \beta}^+,(-1,-2)/\gamma_{\overline \beta}^-}
\draw[fill=black]  
 \Point node[below right] {$\PointLabel$};
 \end{tikzpicture}
 \qquad  \qquad 
 \begin{tikzpicture}[scale=0.7]

\draw[very thick,black, postaction={decorate, decoration={markings, mark = at position 0.5 with {\arrow{>}}}} ] (-4,0) -- (0,0);
\draw[very thick,black,postaction={decorate, decoration={markings, mark = at position 0.5 with {\arrow{>}}}} ] (0,0) -- (4,0);

\draw[very thick,black,postaction={decorate, decoration={markings, mark = at position 0.5 with {\arrow{>}}}} ] (-2,3.464) -- (0,0);
\draw[very thick,black,postaction={decorate, decoration={markings, mark = at position 0.5 with {\arrow{>}}}} ] (-2,-3.464) -- (0,0);

\foreach \Point/\PointLabel in {(0,-0.1)/0, (-1,2.6)/\gamma_{\beta}^-,(-1,-2)/\gamma_{ \beta}^+}
\draw[fill=black]  
 \Point node[below right] {$\PointLabel$};
 \end{tikzpicture}
 \caption{The jump contours of $P_{\overline\beta}(\zeta)$ in $D_{\overline\beta}$ and those of $\widehat{P}_{\beta}(\xi)$ in $D_\beta$. }\label{Fig_jump contours}
 \end{figure}

The solution to the above Riemann-Hilbert problems can be constructed using the Airy parametrices, which we now recall.

\begin{defi}[Airy parametrices] \label{Def_Airy parametrices}
Let 
\begin{equation}
y_j(\zeta)= \omega^j \Ai(\omega^j \zeta), \qquad \omega=e^{2\pi i/3}, \quad (j=0,1,2). 
\end{equation}
Then we define\footnote{Here, there is a minor typo in \cite[Eq.(4.23)]{BBLM15}, where $5\pi/3$ should be replaced with $4\pi/3$.}
\begin{equation} \label{def of Airy parametrix}
    {\mathcal A}(\zeta):=\sqrt{2\pi}e^{-\frac{i\pi}{4}}\begin{cases}
        \begin{pmatrix}
            y_0(\zeta)&-y_2(\zeta)
            \\
            y_0'(\zeta)&-y_2'(\zeta)
        \end{pmatrix} e^{\frac{2}{3}\zeta^{3/2}\sigma_3},\quad &\arg\zeta\in(0,2\pi/3),
        \smallskip 
        \\
     \begin{pmatrix}
            -y_1(\zeta)&-y_2(\zeta)\\
            -y_1'(\zeta)&-y_2'(\zeta)
        \end{pmatrix}e^{\frac{2}{3}\zeta^{3/2}\sigma_3}, \quad &\arg\zeta\in(2\pi/3,\pi),
        \smallskip 
        \\  
        \begin{pmatrix}
            -y_2(\zeta)&y_1(\zeta)\\
            -y_2'(\zeta)&y_1'(\zeta)
        \end{pmatrix}e^{\frac{2}{3}\zeta^{3/2}\sigma_3}, \quad &\arg\zeta\in(\pi,4\pi/3),
        \smallskip 
        \\ \begin{pmatrix}
            y_0(\zeta)&y_1(\zeta)\\
            y_0'(\zeta)&y_1'(\zeta)
        \end{pmatrix}e^{\frac{2}{3}\zeta^{3/2}\sigma_3}, \quad &\arg\zeta\in(4\pi/3,2\pi) 
    \end{cases}
\end{equation}
and 
\begin{equation}
    \widehat{{\mathcal A}}(\xi):=\sqrt{2\pi}e^{-\frac{i\pi}{4}}\begin{cases}
        \begin{pmatrix}
            y_0(\xi)&y_2(\xi)\\
            y_0'(\xi)&y_2'(\xi)
        \end{pmatrix}e^{\frac{2}{3}\xi^{3/2}\sigma_3},\quad &\arg\xi\in(0,2\pi/3),
        \smallskip 
        \\
     \begin{pmatrix}
            -y_1(\xi)&y_2(\xi)\\
            -y_1'(\xi)&y_2'(\xi)
        \end{pmatrix}e^{\frac{2}{3}\xi^{3/2}\sigma_3}, \quad &\arg\xi\in(2\pi/3,\pi),
        \smallskip 
        \\  
        \begin{pmatrix}
            -y_2(\xi)&-y_1(\xi)\\
            -y_2'(\xi)&-y_1'(\xi)
        \end{pmatrix}e^{\frac{2}{3}\xi^{3/2}\sigma_3}, \quad &\arg\xi\in(\pi,4\pi/3),
        \smallskip 
        \\ \begin{pmatrix}
            y_0(\xi)&-y_1(\xi)\\
            y_0'(\xi)&-y_1'(\xi)
        \end{pmatrix}e^{\frac{2}{3}\xi^{3/2}\sigma_3}, \quad &\arg\xi\in(4\pi/3,2\pi). 
    \end{cases}
\end{equation}
\end{defi}

Using the standard Airy Riemann-Hilbert problem, one can show that the solution to the above Riemann-Hilbert problems in Definition~\ref{Def_RHP} is given by 
\begin{align}  \label{def of P beta bar}
P_{\overline\beta}(\zeta)& =e^{\frac{i\pi}{4}\sigma_3}
\frac{1}{\sqrt 2}\begin{pmatrix}
        1&-1\\
        1&1
    \end{pmatrix}\zeta^{\frac{\sigma_3}{4}}{\mathcal A}(\zeta), 
\\  \label{def of P beta hat}
 \widehat{P}_\beta(\xi)& =\begin{pmatrix}
        e^{\frac{i\pi}{4}}&0\\
        0&-e^{-\frac{i\pi}{4}}
    \end{pmatrix}\frac{1}{\sqrt 2}\begin{pmatrix}
        1&-1\\
        1&1
    \end{pmatrix}\xi^{\frac{\sigma_3}{4}}\widehat{{\mathcal A}}(\xi).
\end{align} 
See also \cite[Eq.(4.22)]{BBLM15}.

\begin{lem}[\textbf{Asymptotics of local parametrices}] The solutions $P_{ \overline{\beta} }$ and $\widehat{P}_\beta$ in \eqref{def of P beta bar} and \eqref{def of P beta hat} satisfy the following. 
\begin{itemize}
    \item As $\zeta \to \infty$, we have 
\begin{equation} \label{P beta bar asymp}
    P_{\overline\beta}(\zeta)=I+\frac{\Pi_{11}}{\zeta^{3/2}}+\frac{\Pi_{12}}{\zeta^{3}}+O\Big(\frac{1}{\zeta^{9/2}}\Big),
\end{equation}
where
\begin{equation}
\Pi_{11}=\frac{1}{8}\begin{pmatrix}
1/6&i\\i&-1/6
\end{pmatrix},\qquad \Pi_{12}=\frac{35}{384}\begin{pmatrix}
-1/12&i\\-i&-1/12
\end{pmatrix}.
\end{equation}
\item  As $\xi \to \infty$, we have
\begin{equation} \label{P beta asymp}
    \widehat{P}_{\beta}(\xi)=I+\frac{\Pi_{21}}{\xi^{3/2}}+\frac{\Pi_{22}}{\xi^{3}}+O\Big(\frac{1}{\xi^{9/2}}\Big),
\end{equation}
where
\begin{equation}
\Pi_{21}=\frac{1}{8}\begin{pmatrix}
1/6&-i\\-i&-1/6
\end{pmatrix},\qquad \Pi_{22}=\frac{35}{384}\begin{pmatrix}
-1/12&-i\\i&-1/12
\end{pmatrix}.
\end{equation}
\end{itemize}
\end{lem}
\begin{proof}

Recall the following behaviour of the Airy function: for $|\arg z|<\pi$, 
\begin{equation}
\Ai(z) \sim \frac{ e^{-\frac23 z^{3/2}  } }{ 2\sqrt{\pi} z^{1/4} } \sum_{k=0}^\infty (-1)^k \frac{u_k}{ ( \frac23 z^{3/2}  )^k}, \qquad \Ai'(z) \sim -\frac{z^{1/4} e^{- \frac23 z^{3/2}  } }{ 2\sqrt{\pi} } \sum_{k=0}^\infty (-1)^k \frac{v_k}{ ( \frac23 z^{3/2}  )^k },
\end{equation}
where $u_0=v_0=1$ and for $k=1,2,\dots,$
\begin{equation}
u_k= \frac{(2k+1)(2k+3)\dots (6k-1)}{ (216)^k \, k! }, \qquad v_k= \frac{6k+1}{1-6k} u_k, 
\end{equation}
see \cite[Eqs.(9.7.5), (9.7.6)]{NIST}. 
Using these, the lemma follows from direct computations. 
\end{proof}

Following the definition of the global parametrix in \eqref{def of Phi}
we define 
\begin{align}  \label{def of H hat beta}
\widehat{\textup{\textbf{\textup{H}}}}_\beta(z)&: =\Phi(z)\widehat{S}(z), \qquad 
\widehat{S}(z) : = \begin{pmatrix}
        e^{\frac{i\pi}{4}}&0\\
        0&-e^{-\frac{i\pi}{4}}
    \end{pmatrix}\frac{1}{\sqrt 2}\begin{pmatrix}
    1&-1\\
    1&1
\end{pmatrix}\xi(z)^{\frac{\sigma_3}{4}},
\\   \label{def of H beta bar}
 \textbf{\textup{H}}_{\overline\beta}(z)&:=\Phi(z) S(z), \qquad S(z):= e^{\frac{i\pi\sigma_3}{4}}\frac{1}{\sqrt 2}\begin{pmatrix}
    1&-1\\
    1&1
\end{pmatrix}\zeta(z)^{\frac{\sigma_3}{4}}. 
\end{align}
We show that they are holomorphic and find the $N$-dependence.

\begin{lem}[\textbf{Asymptotics of the global parametrices}] 
The global parametrices in Definition~\ref{Def_Global parametrix} satisfy the following.
\begin{itemize}
    \item As $z \to \overline{\beta}$, we have 
\begin{align}
\begin{split}
\textup{\textbf{H}}_{\overline\beta}(z) & =\frac{(\overline\beta-\beta)^{1/4}}{2}\begin{pmatrix}\displaystyle
(1+i)\gamma_{11}^{1/4}N^{1/6}&\displaystyle\frac{-i}{\sqrt2(R\kappa\gamma_{11})^{1/4}N^{1/6}}
    \smallskip 
    \\
    \displaystyle
    (1-i)\gamma_{11}^{1/4}N^{1/6}&\displaystyle\frac{1}{\sqrt2(R\kappa\gamma_{11})^{1/4}N^{1/6}}
\end{pmatrix}
\\
&\quad -\frac{\gamma_{11}-2i \sqrt{R \kappa}\gamma_{12}}{16(i\sqrt{R\kappa}\gamma_{11})^{3/4}}\begin{pmatrix}\displaystyle
  iN^{1/6} & \displaystyle\frac{1}{2(i\sqrt{R\kappa}\gamma_{11})^{1/2}N^{1/6}} 
  \\\displaystyle
  N^{1/6} &\displaystyle \frac{i}{2(i\sqrt{R\kappa}\gamma_{11})^{1/2}N^{1/6}} 
\end{pmatrix} (z-\overline{\beta}) +O\big((z-\overline{\beta})^2\big). 
\end{split}
\end{align}
\item As $z \to \beta,$ we have
\begin{align}
\begin{split}
\widehat{\textup{\textbf{H}}}_\beta(z) & = \frac{(\beta-\overline\beta)^{1/4}}{2}\begin{pmatrix}
\displaystyle    (1+i)\overline\gamma_{11}^{1/4}N^{1/6}& \displaystyle  \frac{-1}{\sqrt2(R\kappa\overline\gamma_{11})^{1/4}N^{1/6}}
    \smallskip 
    \\
\displaystyle      -(1-i)\overline\gamma_{11}^{1/4}N^{1/6}& 
\displaystyle  \frac{i}{\sqrt2(R\kappa\overline\gamma_{11})^{1/4}N^{1/6}}
\end{pmatrix}
\\
& \quad -   \frac{(i)^{1/4}(\overline\gamma_{12}+2i\sqrt{R\kappa}\overline\gamma_{11})}{16(R\kappa)^{3/8}\overline\gamma_{12}^{3/4}}\begin{pmatrix}\displaystyle N^{1/6} &\displaystyle \frac{i}{2(R\kappa)^{1/4}\overline\gamma_{12}^{1/2}}
  \\
   \displaystyle iN^{1/6} & \displaystyle\frac{1}{2(R\kappa)^{1/4}\overline\gamma_{12}^{1/2}}
\end{pmatrix}       (z-\beta) +O\big((z-\beta)^2\big). 
\end{split}
\end{align}
\end{itemize}
\end{lem}
\begin{proof}
This follows from a long but straightforward computations using Lemma~\ref{Lem_local coord asymp},  \eqref{def of beta} and \eqref{def of F}. 
\end{proof}

We shall apply the partial Schlesinger transform to improve the local parametrix.
A similar feature in deriving the fine asymptotic behaviour in Riemann-Hilbert analysis appears in \cite{BB15}. Additionally, techniques to refine the challenging matching condition for local parametrices in the steepest descent analysis have been explored in \cite{Mo21}.

For our present purpose, we will implement the following steps:
\begin{itemize}
    \item Construct rational matrix functions $R_1(z)$ with the only pole at $\overline\beta$ and $R_2(z)$ with the only pole at $\beta$;
    \smallskip
    \item Construct holomorphic matrix functions $H_1(z)$ and $H_2(z)$ such that the modified global parametrix $R_1(z)R_2(z)\Phi(z)$ matches the local parametrix $R_2(z)H_1(z)\Phi(z)P_{\overline\beta}(\zeta(z))$ along $\partial D_{\overline\beta}$ and matches the local parametrix $R_1(z)H_2(z)\Phi(z)\widehat{P}_{\beta}(\xi(z))$ along $\partial D_{\beta}$.
\end{itemize}

We now define the rational functions. At first glance, the definition of these rational functions may appear quite complex, making it challenging to grasp the underlying motivation. However, this definition is crafted with a specific purpose: to ensure that in the proof of the theorem below, the error matrix satisfies a small norm Riemann-Hilbert problem. Notably, this formulation guarantees the elimination of all terms at the order of $O(N^{-1})$ within this norm, resulting in an error term of the order $O(N^{-2})$. 
This becomes clear in the proof of Theorem~\ref{Thm_OP fine asymp} below.

\begin{defi}[Rational functions] \label{Def_Rational functions}
Let 
\begin{align} \label{def of H1}
  H_1(z)& :=I-\frac{1}{48}\textbf{\textup{H}}_{\overline\beta}(z)\begin{pmatrix}
0&5/\zeta(z)^2\\-7/\zeta(z)&0
\end{pmatrix}\textbf{\textup{H}}_{\overline\beta}(z)^{-1}+\frac{h_{11}}{z-\overline\beta}+\frac{h_{12}}{(z-\overline\beta)^2}, 
\\
 H_2(z)&:=I-\frac{1}{48}\widehat{\textbf{\textup{H}}}_\beta(z)\begin{pmatrix}
0&5/\xi(z)^2\\-7/\xi(z)&0
\end{pmatrix}\widehat{\textbf{\textup{H}}}_\beta(z)^{-1}+\frac{h_{21}}{z-\beta}+\frac{h_{22}}{(z-\beta)^2},   \label{def of H2}
\end{align}
where 
\begin{align}
\begin{split}
h_{11}& =\frac{1+i}{128\sqrt2(R\kappa)^{1/4}\gamma_{11}^{5/2}N}\begin{pmatrix}\displaystyle
3\gamma_{11}-10i\sqrt{R\kappa}\gamma_{12}&\displaystyle\frac{19\gamma_{11}+ 30i\sqrt{R\kappa}\gamma_{12}}{3}\\ \displaystyle
\frac{19\gamma_{11}+30i\sqrt{R\kappa}\gamma_{12}}{3} &\displaystyle -3\gamma_{11}+ 10i\sqrt{R\kappa}\gamma_{12}
\end{pmatrix},
\\
h_{12}&= \frac{5(R\kappa)^{1/4}}{48\sqrt2\gamma_{11}^{3/2}N}\begin{pmatrix}
    -1+i&1+i\\ 1+i &  1-i
\end{pmatrix},
\qquad h_{21}=\overline{h_{11}},
\qquad 
h_{22}= \overline{h_{12}}.
\end{split}
\end{align}
These coefficient $h_{jk}$'s are defined in a way that the matrix functions $H_1$ and $H_2$ are holomorphic at $\overline\beta$ and at $\beta$ respectively. 

We further define 
\begin{equation}\label{def of R2}
  R_2(z):=I+\frac{h_{21}}{z-\beta}+\frac{h_{22}}{(z-\beta)^2}
\end{equation}
and 
\begin{equation}\label{def of R1}
  R_1(z) :=I+\frac{1}{z-\overline\beta}R_2(z)h_{11}R_2(z)^{-1}+\frac{1}{(z-\overline\beta)^2}R_2(z)h_{12}R_2(z)^{-1}.
\end{equation}
Here $R_1(z)$ is a rational matrix function with the only pole at $\overline \beta$, and $R_2(z)$ is a rational matrix function with the only pole at $\beta$.  
Let us also write
\begin{equation} \label{R1 R2 product}
R_1(z)R_2(z):=I+\begin{pmatrix}
    R_{11}(z)&R_{12}(z)\\
    R_{21}(z)&R_{22}(z)
\end{pmatrix}. 
\end{equation}
\end{defi}

\begin{defi}[Strong asymptotics for the pre-critical case] \label{Def_strong asymptotic A}
We define the strong asymptotics of $A(z)$ by 
\begin{equation} \label{A inft strong asymp pre}
    A^{\infty}(z)=\begin{cases}
        R_1(z)R_2(z)\Phi(z),  & z\notin D_{\overline\beta} \cup D_\beta,
        \smallskip 
        \\
        R_2(z)H_1(z)\Phi(z)P_{\overline\beta}(\zeta(z)), & z\in D_{\overline \beta},
        \smallskip 
        \\
        R_1(z)H_2(z)\Phi(z)\widehat{P}_{\beta}(\xi(z)), & z\in D_{\beta}.
        \end{cases}
\end{equation}  
\end{defi}

We note that the strong asymptotic in \cite{BBLM15} is stated in a similar manner but without the rational functions. Then, the error analysis in \cite{BBLM15} can only be applied to derive the leading-order asymptotic behaviour of the orthogonal polynomials. 
In other words, it is crucial to construct the rational functions in Definition~\ref{Def_Rational functions} to define \eqref{A inft strong asymp pre}, which yields the fine asymptotic behaviour.

\begin{proof}[Proof of Theorem~\ref{Thm_OP fine asymp} for the pre-critical case]

We aim to show that 
\begin{equation} \label{A A inft small}
A(z)= \Big(I+O(N^{-2})\Big)A^\infty(z), \qquad z \in \ext(\mathcal B).
\end{equation}
For this, let us define the error function
\begin{equation}
{\mathcal E}(z)=A^\infty(z)A^{-1}(z)
\end{equation}
and verify that $\mathcal{E}$ satisfies a small-norm Riemann-Hilbert problem with error $O(N^{-2})$. 

Note that when $z$ is away from $\partial D_{ \overline \beta }$ and $\partial D_\beta$, the jump of the error matrix is exponentially small in $N$, see \cite[Eq.(4.27)]{BBLM15}.
On the other hand, by \eqref{A inft strong asymp pre} and \eqref{RHP for A pre}, we have
\begin{equation}
{\mathcal E}_+(z){\mathcal E}_-(z)^{-1} = 
\begin{cases}
R_1(z)\Big(H_2(z)\Phi(z)\widehat{P}_{\beta}(\xi(z))\Phi(z)^{-1}R_2(z)^{-1}\Big)R_1(z)^{-1}, & z \in \partial D_\beta,
\smallskip 
\\
\Big(R_2(z)H_1(z)\Phi(z)P_{\overline\beta}(  \zeta(z) )\Phi(z )^{-1}R_2(z)^{-1}\Big)R_1(z)^{-1} , & z \in \partial D_{\overline \beta} .
\end{cases}
\end{equation}

We first discuss the case $z \to \partial D_\beta$. 
Note that by \eqref{def of H hat beta} and \eqref{P beta asymp}, we have
\begin{equation}
\begin{aligned} 
H_2(z)\Phi(z)\widehat{P}_{\beta}(\xi(z) )\Phi(z)^{-1}
&= H_2(z)\widehat{\textbf{H}}_\beta(z)\Big(S(z)^{-1}\widehat{P}_{\beta}(\xi(z))S(z)\Big)\widehat{\textbf{H}}_\beta(z)^{-1}
\\
&=H_2(z)\widehat{\textbf{H}}_\beta(z)\bigg(I+\frac{1}{48} \begin{pmatrix}
0&5/\xi(z)^2\\-7/\xi(z)&0
\end{pmatrix}+O\Big(\frac{1}{\xi(z)^{3}}\Big)\bigg) \widehat{\textbf{H}}_\beta(z)^{-1}
\\
&=H_2(z)\bigg (I+\frac{1}{48}\widehat{\textbf{H}}_\beta(z)\begin{pmatrix}
0&5/\xi(z)^2\\-7/\xi(z)&0
\end{pmatrix}\widehat{\textbf{H}}_\beta(z)^{-1}+O\Big(\frac{1}{\xi(z)^{3}}\Big)\bigg).    
\end{aligned}
\end{equation}
Then it follows from the definition \eqref{def of H2} of $H_2$ that 
\begin{equation}
    \begin{aligned} 
        H_2(z)\Phi(z)\widehat{P}_{\beta}(\xi(z))\Phi(z)^{-1}
        =I+\frac{h_{21}}{z-\beta}+\frac{h_{22}}{(z-\beta)^2}+O\Big(\frac{1}{\xi(z)^{3}},\frac{1}{N^2}\Big).    
    \end{aligned}
\end{equation}
Furthermore, by the definition \eqref{def of R2} of $R_2$, we have that as $z \to \partial D_\beta$, 
\begin{align} \label{small norm D beta}
R_1(z)\Big(H_2(z)\Phi(z)\widehat{P}_{\beta}(\xi(z))\Phi(z)^{-1}R_2(z)^{-1}\Big)R_1(z)^{-1} &= R_1(z)\Big( I+O(N^{-2})\Big)R_1(z)^{-1} = I+O(N^{-2}). 
\end{align}
Here, we have used the fact that $R_1(z)$ is analytic in $D_{\beta}$.

\medskip 

Next, we discuss the case $z \to \partial D_{ \overline \beta }$.
Note that by \eqref{def of H hat beta},  
\begin{align}
\begin{split}
\label{HP asymptotic beta bar}
& \quad  R_2(z)H_1(z)\Big(\Phi(z)P_{\overline\beta}(  \zeta(z) )\Phi(z)^{-1}\Big)R_2(z)^{-1}
\\
&=R_2(z)H_1(z) \textbf{H}_{\overline\beta}(z)\Big(\widehat S(z)^{-1} P_{\overline\beta}(  \zeta(z) )\widehat S(z)\Big)\textbf{H}_{\overline\beta}(z)^{-1}
R_2(z)^{-1}
\\
&=R_2(z)H_1(z)\bigg(I+\frac{1}{48}\textbf{H}_{\overline\beta}(z)\begin{pmatrix}
0&5/\zeta(z)^2\\-7/\zeta(z)&0
\end{pmatrix}\textbf{H}_{\overline\beta}(z)^{-1}+O\Big(\frac{1}{\zeta(z)^{3}}\Big)\bigg)R_2(z)^{-1} .    
\end{split}
\end{align} 
Furthermore, by \eqref{def of H1}, we have
\begin{equation}
    \begin{aligned} 
&\quad R_2(z)H_1(z)\Big(\Phi(z) P_{\overline\beta}(\zeta(z)) \Phi(z)^{-1}\Big)R_2(z)^{-1}
\\
&=R_2(z)\bigg(I+\frac{h_{11}}{z-\overline\beta}+\frac{h_{12}}{(z-\overline\beta)^2}+O\Big(\frac{1}{\zeta(z)^{3}}\Big)\bigg)R_2(z)^{-1}
\\
&=\bigg(I+\frac{1}{z-\overline\beta}R_2(z)h_{11}R_2(z)^{-1}+\frac{1}{(z-\overline\beta)^2}R_2(z)h_{12}R_2(z)^{-1}+O\Big(\frac{1}{\zeta(z)^{3}}\Big)\bigg)
.    
    \end{aligned}
\end{equation}
Then by the definition \eqref{def of R1} of $R_1$, one can see that as $ z \to D_{ \overline \beta }$, 
\begin{align}
\begin{split} \label{small norm D beta bar}
\Big(R_2(z)H_1(z)\Phi(z)P_{\overline\beta}(  \zeta(z) )\Phi(z )^{-1}R_2(z)^{-1}\Big)R_1(z)^{-1} = I+O(N^{-2}). 
\end{split}
\end{align} 

Combining \eqref{small norm D beta} and \eqref{small norm D beta bar}, we conclude \eqref{A A inft small}. 
As a consequence of \eqref{A A inft small}, for $z \in \C\setminus(\Omega_+\cup\Omega_-)$, the explicit transform \eqref{A transform} together with the definition \eqref{A inft strong asymp pre} gives rise to  
\begin{equation}
    \begin{aligned}
Y(z)&=e^{N\ell/2\sigma_3}A(z)e^{tN(g(z)-\ell/2)\sigma_3} =e^{N\ell/2\sigma_3}\Big(I+O(N^{-2})\Big)A^\infty(z)e^{N(g(z)-\ell/2)\sigma_3}\\
&=e^{ N\ell/2\sigma_3}\Big(I+O(N^{-2})\Big)R_1(z)R_2(z)\Phi(z)e^{ N(g(z)-\ell/2)\sigma_3}. 
    \end{aligned}
\end{equation}
Then by \eqref{def of Phi} and \eqref{R1 R2 product}, this in turn implies that 
\begin{align}
\begin{split}
Y(z)& =e^{N\ell/2\sigma_3}\Big(I+O(N^{-2})\Big)\begin{pmatrix}
    1+R_{11}(z)&R_{12}(z)\\
    R_{21}(z) &1+R_{22}(z)
\end{pmatrix}\begin{pmatrix}
    \sqrt{ R F'(z)}&\dfrac{\sqrt{\kappa F'(z)}}{F(z)-  q}
    \smallskip 
    \\
    -\dfrac{\sqrt{\kappa F'(z)}}{F(z)-  q}&\sqrt{ R F'(z)}
\end{pmatrix}e^{ N(g(z)-\ell/2)\sigma_3}.
\end{split}
\end{align}
In particular, since $ p_n(z)=[Y(z)]_{11}$, we obtain the desired asymptotic behaviour \eqref{pre asymptotics}. 
\end{proof}

\subsection{Post-critical case}

In this subsection, we now discuss the post-critical case, which is much simpler compared to the pre-critical case.

Recall that for the post-critical case, the global parametrix $\Phi$ is given by \eqref{def of Phi post}. 
Similar to the pre-critical case, we also apply the partial Schlesinger transform to improve the local parametrix. 
Namely, we derive a rational matrix functions $R(z)$ with the only pole at $\beta$,  and a holomorphic matrix function $H(z)$ such that the modified global parametrices 
\begin{equation}
\Phi(z)R(z), \qquad \Phi(z)\begin{pmatrix}
0&-1\\1&0
\end{pmatrix}R(z)\begin{pmatrix}
0&1\\-1&0
\end{pmatrix}
\end{equation}
match with the local parametrix $\Phi(z)P(\zeta(z))$ along $\partial D_{\beta}$.  The construction of such $H(z)$ and $R(z)$ with a simple pole at $\beta$ was described in \cite[Section 5]{BBLM15}. 

For the post-critical case, we define the local coordinate $\zeta$ inside $D_\beta$, 
 \begin{equation}
    \zeta(z)^2=2N \begin{cases}
         \phi(z), & z\in \ext(\mathcal B),
         \smallskip 
         \\
         -\phi(z),& z\in \mbox{int}(\mathcal B). 
     \end{cases}
 \end{equation}
Here $\phi$ is defined in Definition~\ref{def gfunction}. 
Here the sign is $+$ for $ z \in \ext(\mathcal{B})$ and $-$ for $z \in \rm{int}(\mathcal{B})$. 
Note that $\zeta$ maps ${\mathcal B}$ into rays $\gamma_+\cup \gamma_-$, where $\gamma_+:=[0,e^{\pi i/4\infty})$ and $\gamma_-:=[0,e^{-\pi i/4\infty})$, maps $\Gamma_{+}$ into imaginary axis $i\R$. 

It can be observed from the definition that as $z \to \beta$, we have the expansion
\begin{equation}
    \frac{\zeta(z)}{\sqrt N}=\frac{1}{\gamma_1}(z-\beta)(1+O(z-\beta)), \qquad  \gamma_1=\sqrt{\frac{\beta(\beta-a)}{a(b-\beta)}},
\end{equation} 
where $\beta$ and $b$ are given in \eqref{def beta and b}.
Inside $D_\beta$, we define the following Riemann-Hilbert problem
\begin{equation}
    \begin{cases}
    [P(\zeta)]_+=\begin{pmatrix}
0&-1\\1&0
\end{pmatrix}[P(\zeta)]_-\begin{pmatrix}
0&1\\-1&0
\end{pmatrix},\quad &\zeta\in \gamma_+\cup \gamma_-,
\smallskip 
\\
        [P(\zeta)]_+=[P(\zeta)]_-\begin{pmatrix}
1&0\\e^{\zeta^2}&1
\end{pmatrix},\quad &\zeta \in i\R,
\smallskip 
\\
P(\zeta)=I+O(\zeta^{-1}),\quad & \zeta\to\infty. \end{cases}
\end{equation} 
Then $\Phi(z)P(\zeta(z))$ satisfies the same jump conditions of $A(z)$.
Furthermore, the solution to this Riemann-Hilbert problem is given by
\begin{equation}
 P(\zeta(z))=   \begin{cases}
        H(z)F(\zeta(z)),& z\in \ext(\mathcal B),
        \smallskip 
        \\
        \begin{pmatrix}
0&-1\\1&0
\end{pmatrix}H(z)F(\zeta(z))\begin{pmatrix}
0&1\\-1&0
\end{pmatrix},& z\in \mbox{int}(\mathcal B),
    \end{cases}
\end{equation}
where 
\begin{equation}
    F(\zeta(z)):=\begin{pmatrix}
1&\displaystyle-\frac{1}{2\pi i}\int_{-i\infty}^{i\infty} \frac{e^{s^2/2}}{s-\zeta(z)} ds 
\smallskip 
\\
0&1
\end{pmatrix} 
\end{equation}
and 
\begin{equation}
    H(z):=R(z)\bigg(I-\frac{1}{\sqrt{2\pi}\zeta(z)}\begin{pmatrix}
0&1\\0&0
\end{pmatrix}\bigg),\qquad R(z):=I+\frac{\gamma_1}{\sqrt{2\pi N}}\frac{1}{z-\beta}\begin{pmatrix}
0&1\\0&0
\end{pmatrix}.
\end{equation}

\begin{figure}[t]  
\centering
 \begin{tikzpicture}[scale=0.7]

\draw[very thick,black, postaction={decorate, decoration={markings, mark = at position 0.5 with {\arrow{>}}}} ] (0,-4) -- (0,0);
\draw[very thick,black,postaction={decorate, decoration={markings, mark = at position 0.5 with {\arrow{>}}}} ] (0,0) -- (0,4);

\draw[very thick,black,postaction={decorate, decoration={markings, mark = at position 0.5 with {\arrow{>}}}} ] (0,0) -- (4,4);
\draw[very thick,black,postaction={decorate, decoration={markings, mark = at position 0.5 with {\arrow{>}}}} ] (0,0) -- (4,-4);

\foreach \Point/\PointLabel in {(-0.6,0)/0, (-1,2.6)/i\R_+,(-1,-2)/i\R_-,(3,2.6)/\gamma_+,(3,-2.3)/\gamma_- }
\draw[fill=black]  
 \Point node[below right] {$\PointLabel$};
 \end{tikzpicture}
 \caption{The jump contours of $P(\zeta)$ in $D_{\beta}$. }
 \end{figure}
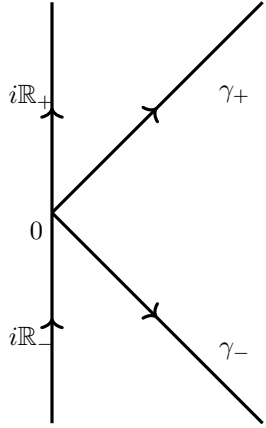

\begin{defi}[Strong asymptotics for the post-critical case]
We define the strong asymptotics of $A(z)$ by 
\begin{equation}
    A^{\infty}(z)=\begin{cases}
        \Phi(z)R(z),  & z\in \ext(\mathcal B)\setminus D_{\beta},
        \smallskip 
        \\
        \Phi(z)\begin{pmatrix}
0&-1\\1&0
\end{pmatrix}R(z)\begin{pmatrix}
0&1\\-1&0
\end{pmatrix}, & z\in \mbox{int}(\mathcal B)\setminus D_{\beta},
        \smallskip 
        \\
        \Phi(z)P(\zeta(z)), & z\in D_{\beta}.
        \end{cases}
\end{equation}\end{defi}

We are now ready to complete the proof of Theorem~\ref{Thm_OP fine asymp}. 

\begin{proof}[Proof of Theorem~\ref{Thm_OP fine asymp} for the post-critical case]
As before, we define the error function
\begin{equation}
{\mathcal E}(z)=A^\infty(z)A^{-1}(z).
\end{equation}
Note that as $z \to D_\beta$, 
\begin{equation}
   F(\zeta(z))= I+\frac{1}{\sqrt{2\pi}\zeta(z)}\begin{pmatrix}
0&1\\0&0
\end{pmatrix}+O\Big(\frac{1}{\zeta(z)^{3}}\Big). 
\end{equation}
Therefore, when $z$ is on the boundary of $\partial D_{\beta}\cap \ext(\mathcal B)$, we have
\begin{align*}
&\quad {\mathcal E}_+(z){\mathcal E}_-(z)^{-1} =\Phi(z) P(\zeta(z))R(z)^{-1} \Phi(z)^{-1} =H(z)F(\zeta(z))R(z)^{-1}
\\
&=R(z)\bigg(I-\frac{1}{\sqrt{2\pi}\zeta}\begin{pmatrix}
0&1\\0&0
\end{pmatrix}\bigg)\bigg(I+\frac{1}{\sqrt{2\pi}\zeta}\begin{pmatrix}
0&1\\0&0
\end{pmatrix}+O\Big(\frac{1}{\zeta^{3}}\Big)\bigg)R(z)^{-1} =I+O(N^{-3/2}). 
\end{align*} 
Similarly, one can check that the same error bound holds for $z\in \partial D_{\beta}\cap \mbox{int}(\mathcal B)$. One can also check that the error bounds are exponentially small for $z$ in other regions. Using the small norm theorem, we obtain $$A(z)=\Big(I+O(N^{-3/2})\Big)A^\infty(z).$$

Then for $z\in \ext(\mathcal B)\setminus D_{\beta}$, it follows that
\begin{align*}
\begin{split}
Y(z)&=e^{ N\ell/2\sigma_3}A(z)e^{ N(g(z)-\ell/2)\sigma_3}
=e^{ N\ell/2\sigma_3}\Big(I+O(N^{-3/2})\Big)A^\infty(z)e^{N(g(z)-\ell/2)\sigma_3}
\\
&=e^{ N\ell/2\sigma_3}\Big(I+O(N^{-3/2})\Big)\Phi(z)R(z)e^{ N(g(z)-\ell/2)\sigma_3}. 
\end{split}
\end{align*}
This gives 
\begin{equation}
Y(z) = e^{ N\ell/2\sigma_3}\Big(I+O(N^{-3/2})\Big)\begin{pmatrix}
1&\dfrac{1}{\sqrt{2\pi N}}\dfrac{\gamma_1}{z-\beta}
\smallskip 
\\0&1
\end{pmatrix}\begin{pmatrix}
\dfrac{z^{Nc+N}}{(z-a)^{Nc}}&0
\smallskip 
\\
0&\dfrac{(z-a)^{Nc}}{z^{Nc+N}}
\end{pmatrix}e^{-N\ell/2\sigma_3}, 
\end{equation}
which in particular yields 
$$p_{N}(z)=[Y(z)]_{11} = z^N \Big( \frac{z}{z-a} \Big)^{c N} \Big(1+O(N^{-3/2})\Big).$$

Note that the above partial Schlesinger transform only updates the $(1,2)$-entry in $Y(z)$, but not the $(1,1)$-entry. Thus, by continuing to use the partial Schlesinger transform, we can improve the error bound so that it becomes $O(N^{-m})$ for any $m > 0$. Therefore, we conclude \eqref{post asymptotics}.
\end{proof}

\bigskip

\section{Duality and critical case} \label{Subsec_critical}

In this section, we establish the free energy expansion for the critical case.
Let us first discuss the duality formula \eqref{equivalence btw three}. 

\begin{proof}[Proof of Proposition~\ref{Prop_equivalence}]
    
We consider the LUE of size $(cN)\times (cN)$ with joint probability distribution proportional to 
\begin{equation}
\prod_{j>k=1}^{cN} |\widetilde{\lambda}_j- \widetilde{\lambda}_k|^2 \prod_{j=1}^{cN} \widetilde{\lambda}_j^{N} e^{-N \widetilde{\lambda}_j}, \qquad (\widetilde{\lambda}_{cN} > \dots > \widetilde{\lambda}_1 >0).  
\end{equation}
Then by \cite[Proposition 3.1]{DS22} with $k=cN$, 
\begin{equation}
 \mathbb{P}\Big[ \widetilde{\lambda}_1 > x^2 \Big]= \mathbb{E} \Big| \det(\textbf{\textup{G}}_N - x ) \Big|^{2cN}  N^{ cN^2  }  e^{ -cN^2 x^2 }  \frac{ G(cN+1) G(N+1) }{ G(cN+N+1) } . 
\end{equation}
Furthermore by \eqref{ZN reference evaluations} and \eqref{rel char poly ZN}, it can be rewritten as 
\begin{align}
\begin{split}
\mathbb{P}\Big[ \widetilde{\lambda}_1 > x^2 \Big] &= \mathbb{E} \Big| \det(\textbf{\textup{G}}_N - x ) \Big|^{2cN} \, e^{ -cN^2 x^2 } \,  \frac{ Z_N^{ \rm Gin } }{ Z_N(0,c) } =  e^{ -cN^2 x^2 } \,  \frac{ Z_N(x,c) }{ Z_N(0,c) }.
\end{split}
\end{align}
We make a change of variable $\widetilde{\lambda}_j = c\, \widehat{\lambda}_j $. Then $\widehat{\lambda}_j$ follows the distribution proportional to  
\begin{equation}
\prod_{j>k=1}^{cN} |\widehat{\lambda}_j- \widehat{\lambda}_k|^2 \prod_{j=1}^{cN} \widehat{\lambda}_j^{N} e^{-cN \widehat{\lambda}_j}, \qquad (\widehat{\lambda}_{cN} > \dots > \widehat{\lambda}_1 >0).  
\end{equation}
If we relabelling $cN \mapsto N$, this follows the LUE \eqref{LUE} with $\alpha=1/c$. 
Then the duality formula \eqref{equivalence btw three} follows from 
\begin{equation}
\mathbb{P}\Big[ \lambda_1 > \frac{x^2}{c} \Big]= \mathbb{P}\Big[ \widehat{\lambda}_1 > \frac{x^2}{c} \Big] \bigg|_{N\to N/c}, \qquad 
 \mathbb{P}\Big[ \widehat{\lambda}_1 > \frac{x^2}{c} \Big]   = \mathbb{P}\Big[ \widetilde{\lambda}_1 > x^2 \Big] =    e^{ -cN^2 x^2 } \,  \frac{ Z_N(x,c) }{ Z_N(0,c) }.
\end{equation}
This completes the proof. 
\end{proof}

We now prove Proposition~\ref{Prop_critical expansion}. 

\begin{proof}[Proof of Proposition~\ref{Prop_critical expansion}]
Note that by \eqref{a cri asymp}, 
\begin{align*}
\frac{a^2}{c} 
&= \lambda_- - \frac{  ( \sqrt{c+1}-\sqrt{c})^{4/3} }{  c^{7/6} (c+1)^{1/6} } \,s\, N^{-2/3} +O(N^{-4/3}). 
\end{align*}
Notice here that the Marchenko-Pastur law \eqref{MP} satisfies the behaviour
\begin{equation}
\frac{1}{2\pi}\frac{\sqrt{(\lambda_{+}-x) (x-\lambda_{-})  }}{x} \sim \frac{\delta}{\pi} \sqrt{ x-\lambda_- }, \qquad x \to \lambda_-,
\end{equation}
where
\begin{equation} \label{def of delta}
\delta = \frac{1}{2} \frac{ \sqrt{ \lambda_+-\lambda_- } }{ \lambda_- } = \frac{ (\alpha+1)^{1/4}}{ (\sqrt{\alpha+1}-1)^2 } =   \frac{ c^{3/4} (c+1)^{1/4}}{  (\sqrt{c+1}- \sqrt{c} )^2 }.  
\end{equation}
Then by the edge universality of the Hermitian unitary ensembles \cite{DKMVZ99}, it follows that 
\begin{equation}
\mathbb{P}\Big[  (\lambda_1-\lambda_-) (\delta N)^{2/3}   > -y \Big] \to F_{ \rm TW }(y) . 
\end{equation}
(In our present case, this also follows from the classical Plancherel-Rotach asymptotic formula for the Laguerre polynomials.)
To be more precise, we have 
\begin{align*}
\mathbb{ P }\Big[ \lambda_1 > \frac{a^2}{c} \Big] & = \mathbb{ P }\Big[ \lambda_1 > \lambda_- - \frac{  ( \sqrt{c+1}-\sqrt{c})^{4/3} }{  c^{7/6} (c+1)^{1/6} } \,s\, N^{-2/3} +O(N^{-4/3}) \Big]
\\
& = \mathbb{P} \Big[ (\lambda_1-\lambda_-) (\delta N)^{2/3} >  -\delta^{2/3} \frac{  ( \sqrt{c+1}-\sqrt{c})^{4/3} }{  c^{7/6} (c+1)^{1/6} }  s   \Big] +O(N^{-2/3})
\\
&= \mathbb{P} \Big[ (\lambda_1-\lambda_-) (\delta N)^{2/3} >  -  \frac{ s }{  c^{2/3}  }    \Big] +O(N^{-2/3}).
\end{align*}
Therefore we obtain 
\begin{equation}
\mathbb{ P }\Big[ \lambda_1 > \frac{a^2}{c} \Big] = F_{\rm TW}(c^{-2/3} s) +O(N^{-2/3}).
\end{equation}
Then the proposition follows from Proposition~\ref{Prop_equivalence} and \eqref{Z ind Gin asy}. 
\end{proof}

While the proof of Proposition~\ref{Prop_critical expansion} follows easily from the duality relation as well as the well-established Hermitian random matrix theory, our overall strategy presented in Subsections~\ref{Subsec_Deform}, \ref{Subsec_tau} and \ref{Subsec_fine asymp} also works for the critical case.
For this case, again by using the partial Schlesinger transform, one can obtain the fine asymptotic behaviour of the orthogonal polynomial, which can be written in terms of the solution to the Painlev\'e II equation 
\begin{equation}
\mathfrak{q}''(s)=s\mathfrak{q}(s)+2\mathfrak{q}(s)^3
\end{equation} with the behaviour 
\begin{equation}
\mathfrak{q}(s)=\begin{cases}
 {\rm Ai}(s)+{ O}\Big(\dfrac{e^{-\frac{4}{3}s^{3/2}}}{s^{1/4}}\Big),
 & s\to +\infty ,
 \smallskip 
 \\
 \sqrt{-\frac{1}{2}s}\Big(1+{ O}\big(s^{-3}\big)\Big),& s\to -\infty.
\end{cases}
\end{equation}
An additional difficulty arises from the fact that, regardless of the deformation \eqref{ZN ref post} or \eqref{ZN ref pre} we employ, in the critical case, we also need to address either the post-critical or pre-critical regimes. Nonetheless, our approach leads to the free energy expansion expressed in terms of the Painlevé II solution, which should also coincide with the Tracy-Widom distribution.
This would lead to an alternative representation of the Tracy-Widom distribution, see also \cite{Bar24} for a recently found another representation.

\subsection*{Acknowledgements} 
Sung-Soo Byun was partially supported by the National Research Foundation of Korea grant (RS-2023-00301976, RS-2025-00516909).
Seong-Mi Seo was partially supported by the National Research Foundation of Korea (NRF-2022R1I1A1A01072052).
Meng Yang was partially supported by the Research Grant RCXMA23007 and the Start-up funding YJKY230037 at Great Bay University.

The authors are greatly indebted to Seung-Yeop Lee for several valuable suggestions on the partial Schlesinger transform during the Riemann-Hilbert analysis. 
Special thanks are extended to Nick Simm for bringing the duality formula to our attention, to Christian Webb for insightful discussions concerning Remark~\ref{Rem_small insertion}, and to Christophe Charlier for engaging in numerous extensive discussions, particularly regarding Remark~\ref{Rem_1D partition}.
In addition, we would like to express our gratitude to Gernot Akemann, Yacin Ameur, Paul Bourgade, Peter J. Forrester, Kohei Noda, Sylvia Serfaty, and Aron Wennman for their interest and stimulating discussions.

\end{document}